\documentclass{article}

\usepackage{arxiv}

\usepackage[utf8]{inputenc}
\usepackage[T1]{fontenc}
\usepackage[american]{babel}
\usepackage{csquotes}
\usepackage{microtype}
\usepackage{hyperref}
\usepackage{url}
\usepackage{graphicx}
\usepackage{amsmath,amssymb,amsfonts,mathtools}
\usepackage{amsthm}
\usepackage{bm}
\usepackage{booktabs}
\usepackage{array}
\usepackage{adjustbox}
\usepackage{longtable}
\usepackage{tabularx}
\usepackage{multirow}
\usepackage{multicol}
\usepackage{enumitem}
\usepackage{xcolor}
\usepackage{comment}
\usepackage{caption}
\usepackage{pdflscape}
\usepackage{setspace}
\usepackage[style=apa,backend=biber,natbib=true,sortcites=true,maxcitenames=2,maxbibnames=99]{biblatex}
\allowdisplaybreaks
\hypersetup{
  colorlinks=true,
  linkcolor=blue,
  citecolor=blue,
  urlcolor=blue,
  pdftitle={Doubly robust estimation of while-alive estimands in individually-randomized and cluster-randomized trials},
  pdfauthor={Xi Fang, Da Zhao, and Fan Li}
}

\theoremstyle{plain}
\newtheorem{theorem}{Theorem}
\newtheorem{lemma}{Lemma}
\newtheorem{proposition}{Proposition}
\newtheorem{corollary}{Corollary}

\theoremstyle{definition}

\newtheorem{assumption}{Assumption}

\theoremstyle{remark}
\newtheorem{remark}{Remark}

\newcommand{\bZ}{\mathbf{Z}}
\newcommand{\bN}{\mathbf{N}}
\newcommand{\bw}{\mathbf{w}}
\newcommand{\bh}{\mathbf{h}}
\newcommand{\bO}{\mathbf{O}}

\newcommand{\bbeta}{\boldsymbol{\beta}}
\newcommand{\balpha}{\boldsymbol{\alpha}}
\newcommand{\bgamma}{\boldsymbol{\gamma}}

\newcommand{\Ebb}{\operatorname{E}}
\newcommand{\Pbb}{\operatorname{P}}
\providecommand{\Var}{}
\renewcommand{\Var}{\operatorname{Var}}

\newcommand{\ind}{\mathbb{I}}
\newcommand{\dd}{\,d}
\providecommand{\tr}{}
\renewcommand{\tr}{^{\mathsf T}}
\newcommand{\Pn}{\mathbb{P}_n}
\newcommand{\PM}{\mathbb{P}_M}
\newcommand{\Gn}{\mathbb{G}_n}

\newcommand{\R}{\mathrm{R}}
\newcommand{\cF}{\mathcal{F}}

\newcommand{\cO}{\mathcal{O}}
\newcommand{\cC}{\mathcal{C}}

\newcommand{\cV}{\mathcal{V}}

\title{Doubly robust estimation of while-alive estimands in individually-randomized and cluster-randomized trials}

\author{
Xi Fang\\
Data Science Institute\\
Medical College of Wisconsin\\
Milwaukee, WI, USA\\
\texttt{xfang@mcw.edu}
\And
Da Zhao\\
Department of Biostatistics\\
Yale School of Public Health\\
New Haven, CT, USA\\
\texttt{dzhao@yale.edu}
\And
Fan Li\\
Department of Biostatistics\\
Yale School of Public Health\\
New Haven, CT, USA\\
\texttt{fan.f.li@yale.edu}\\
Corresponding author
}

\begin{document}
\maketitle
\onehalfspacing

\begin{abstract}
Randomized trials in chronic disease settings often measure treatment benefit through recurrent nonfatal events that are truncated by death, where conventional summaries either discard recurrences, conflate the treatment effect with survival, or treat death as censoring and forfeit a causal interpretation. While-alive estimands measure event burden per unit time alive, but doubly robust estimation for the exposure-weighted while-alive rate remains undeveloped, particularly in cluster-randomized trials (CRTs). We develop a doubly robust estimator based on a local Nelson--Aalen representation, with augmented estimating equations targeting the marginal hazard of the terminal event and the weighted recurrent event rate among those alive; the estimator accommodates multiple event types through prespecified clinical weights and remains consistent if either the censoring model or the outcome working models are correctly specified. For CRTs, we define a new pair of individual-average and cluster-average estimands under informative cluster size, with inference based on cluster-level influence functions. We establish component-wise double robustness and asymptotic normality, corroborate the theory in simulations, and illustrate the methods with reanalyses of data from two completed randomized trials.
\end{abstract}



\section{Introduction} \label{sec:introduction}

A growing number of randomized trials measure treatment benefit through a stream of recurrent nonfatal events, such as repeated hospitalizations in heart failure, repeated exacerbations in pulmonary disease, repeated relapses in neurology, or repeated falls and injuries in health system studies. In each of  these settings, the recurrent events are subject to a terminal event, typically death, which usually complicates the definition and interpretation of a summary treatment effect. Death is not an ordinary censoring event. It disproportionately affects the patients who are most prone to events, so that the sickest patients contribute both the largest event counts and the shortest follow-up. The following two traditional analyses therefore each have their own limitations. Reducing the outcome to the time of the first event ignores the recurrences that are meaningful for patient disease trajectory, whereas summarizing it through the cumulative number of events confounds the treatment effect with survival, since an arm that keeps patients alive longer is penalized with more events \citep{ema2020qualification,cook2002analysis,mao2022whilealive}. On the other hand, a summary that measures the burden of events relative to the time a patient is actually alive is of increasing interest. 

Statistical methods for recurrent event analysis have been well developed under individually-randomized trials (IRTs). For example, built upon counting process arguments \citep{nelson1972hazard,aalen1978nonparametric,andersen1982cox,gill1990product,fleming1991counting,andersen1993statistical}, early developments include the conditional regression models of \citet{prentice1981regression}, the robust marginal rate and mean function methods of \citet{lawless1995simple}, and the semiparametric regression models of \citet{lin2000semiparametric}. When the recurrent events are subject to a terminal event such as death, two major approaches can be considered. The first targets the marginal rate or mean function of recurrent events, acknowledging that no recurrent events can occur after death \citep{cook1997marginal,ghosh2000nonparametric,ghosh2002marginal,cook2009robust}. The second approach jointly models the recurrent event process and the terminal event process, typically through a shared frailty inducing dependence between processes \citep{wang2001analyzing,huang2004joint,liu2004shared,lin2017multitype}. Relatedly, \citet{mao2016semiparametric} proposed semiparametric regression for a weighted composite endpoint of recurrent and terminal events, where prespecified weights reflect the relative clinical importance of each event type. More recently, recurrent events with a terminal event have also been studied from a causal inference perspective. \citet{genetti2025efficient} developed efficient estimators of the marginal mean of recurrent events in randomized trials. \citet{huang2026model} proposed model-assisted causal inference for the treatment effect on recurrent events in the presence of terminal events. \citet{mork2025structural} introduced a structural nested rate model for time varying exposures. In parallel, \citet{lyu2023principal} and \citet{ohnishi2025principal} adopted the principal stratification framework \citep{frangakis2002principal} and restricted inference to the subpopulation of individuals who would always survive under either treatment assignment. However, none of these methods directly addresses the event burden accrued per unit of time alive. The marginal and joint modeling approaches characterize how often events occur, whether through a cumulative mean function, a rate function, or a regression coefficient, and are therefore inevitably influenced by the distribution of the terminal event. The principal stratification approach instead redefines the target population free of the terminal event, and hence the resulting inference no longer applies to the original, full trial population.

The while-alive estimand directly addresses this limitation. Under the estimands framework of the ICH E9(R1) addendum \citep{iche9r12019}, death is treated as an intercurrent event, and the while-alive strategy defines the outcome only over the period during which a participant remains alive \citep{ema2020qualification,schmidli2023estimands}. By normalizing the event burden by the survival time, the while-alive estimand avoids rewarding a treatment merely for shortening survival. Two versions of the while-alive estimand are available, and they differ in how each participant contributes to the population summary. The patient-weighted estimand averages each participant's own ratio of event count to time alive and therefore assigns equal weight to every participant. The exposure-weighted estimand takes the ratio of the expected event burden to the expected survival time and therefore weights each participant in proportion to the length of survival \citep{mao2022whilealive,wei2023properties}. Since the formalization of the estimands framework, several estimation and inference techniques have been developed. \citet{mao2022whilealive} first studied a general class of while-alive estimands under a nonparametric framework. \citet{wei2023properties} examined the properties of two specific while-alive estimands and their estimators. \citet{fang2025whilealive} developed regression methods for while-alive summaries of composite survival endpoints. \citet{ragni2026patient} derived the efficient influence function and proposed an efficient nonparametric estimator for the patient-weighted estimand. Despite these advances, robust methods for estimating the exposure-weighted while-alive estimand remain relatively underdeveloped. This estimand is particularly relevant when the scientific objective is to quantify the population-level event burden per unit of aggregate survival time, as may be appropriate for assessing recurrent morbidity, healthcare utilization, or multiple event types assigned prespecified weights according to their clinical importance. We therefore focus on developing robust estimators for this estimand. The work most closely related to ours is \citet{baer2025causal}, who embedded recurrent events and death within a multiply robust causal framework by targeting a landmark-indexed vector of expected event counts and survival probabilities. Their estimators, however, are constructed directly on the survival function scale, and the resulting survival estimates are not guaranteed to be monotone across time. We instead formulate the problem through local hazard increments and develop a Nelson-Aalen-type construction. By building the survival function from an estimated cumulative hazard, this approach preserves its required monotonicity while accommodating multiple, prespecified clinically weighted event types. Table \ref{tab:lit_compare} transparently situates this contribution relative to the existing literature.

A separate complication arises when randomization occurs at the level of clusters rather than individuals, such as when hospitals, clinics, primary care practices, or communities are assigned to intervention groups. In cluster-randomized trials (CRTs), treatment effects can be defined under different weighting schemes. The cluster-average treatment effect assigns equal weight to each cluster regardless of its size, whereas the individual-average treatment effect assigns equal weight to each participant across all clusters. These two estimands coincide when cluster sizes are constant or only randomly varying, but they diverge in the presence of informative cluster size, where cluster size is marginally associated with the cluster specific treatment effect. The choice between the two estimands therefore carries direct implications for interpretation and inference \citep{kahan2023estimands,kahan2023ics,kahan2024demystifying}. This recognition has motivated a growing body of causal inference methods for CRTs with noncensored outcomes \citep{balzer2019hierarchical,balzer2023twostage,wang2024modelrobust,li2025standardization} and, more recently, with right-censored survival outcomes \citep{fang2026estimands,fang2026rmtif}. In addition, \citet{grant2025recurrentcrt} recently investigated the analysis of recurrent events in CRTs. However, no existing method for CRTs targets a while-alive summary of the recurrent event burden, nor provides doubly robust estimation for such a summary under censoring that depends on baseline covariates. As a result, the analysis of a CRT with recurrent and terminal events must confront two challenges simultaneously. A while-alive estimand is required to appropriately account for death, and the distinction between cluster-average and individual-average estimands is required to better reflect the scientific question in mind.

In this article, we resolve both challenges within a single framework. We focus on an exposure-weighted while alive recurrent event rate, defined as the ratio of the expected weighted recurrent event count to the restricted mean survival time, where the weights are prespecified to reflect the clinical importance of each event type. In the IRT setting, we develop an estimator based on a local representation of the Nelson-Aalen type, in which one augmented estimating equation targets the marginal hazard of the terminal event and a second targets the marginal weighted event rate among participants who remain alive. The resulting estimator is doubly robust in the sense that it remains consistent when either the censoring model or the outcome models are correctly specified, but not necessarily both. We then extend the same construction to CRTs by formalizing a pair of while-alive cluster-average and individual-average estimands. 
Specifically, the remainder of the paper is organized as follows. Section \ref{sec:irt_estimand_estimation} introduces the proposed exposure-weighted while-alive estimands for IRTs, establishes their identification, and develops the proposed doubly robust local Nelson–Aalen estimator, and establishes the large-sample properties of the proposed estimator. Section \ref{sec:crt} extends the proposed framework to CRTs, develops the corresponding doubly robust estimators, and establishes the asymptotic properties of the proposed estimators under cluster randomization. Section \ref{sec:simulation} evaluates the finite-sample performance of the proposed methods through simulation studies. Section \ref{sec:data_illustration} illustrates the proposed methodology through two real data applications, and Section \ref{sec:discussion} concludes with a discussion.

\begin{table}[tbp]
\caption{Related recurrent-event methods.}
\label{tab:lit_compare}
\centering
\renewcommand{\arraystretch}{1.16}
\setlength{\tabcolsep}{3.2pt}
\begin{adjustbox}{width=\textwidth,center}
\begin{tabular}{
p{2.65cm}
p{5.75cm}
>{\centering\arraybackslash}p{1.15cm}
>{\centering\arraybackslash}p{1.35cm}
>{\centering\arraybackslash}p{1.10cm}
>{\centering\arraybackslash}p{1.35cm}
>{\centering\arraybackslash}p{1.20cm}
}
\toprule
Method
& Estimands
& \shortstack{WA\\rate}
& \shortstack{Cov.-dep.\\cens.}
& DR/MR
& \shortstack{Event\\weights}
& CRT \\
\midrule

\citet{schmidli2023estimands}
& Recurrent event estimand strategies with death, including while-alive strategies
& \(\triangle\)
& \(\text{--}\)
& \(\text{--}\)
& \(\triangle\)
& \(\times\) \\[2pt]

\citet{wei2023properties}
& Exposure-weighted and composite exposure-weighted while-alive estimands
& \(\surd\)
& \(\times\)
& \(\times\)
& \(\triangle\)
& \(\times\) \\[2pt]

\citet{mao2022whilealive}
& Generalized exposure-weighted while-alive loss rate
& \(\surd\)
& \(\times\)
& \(\times\)
& \(\surd\)
& \(\times\) \\[2pt]

\citet{ragni2026patient}
& Patient-weighted while-alive estimand
& \(\surd\)
& \(\triangle\)
& \(\triangle\)
& \(\times\)
& \(\times\) \\[2pt]

\citet{fang2025whilealive}
& Generalized exposure-weighted while-alive loss rate for composite survival endpoints
& \(\surd\)
& \(\surd\)
& \(\times\)
& \(\surd\)
& \(\triangle\) \\[2pt]

\citet{baer2025causal}
& Landmarked recurrent event means and survival probabilities
& \(\triangle\)
& \(\surd\)
& \(\surd\)
& \(\times\)
& \(\times\) \\[2pt]

\midrule

\textbf{Proposed method}
& Exposure-weighted while-alive recurrent event rate for clinically weighted recurrent events
& \(\surd\)
& \(\surd\)
& \(\surd\)
& \(\surd\)
& \(\surd\) \\
\bottomrule
\end{tabular}
\end{adjustbox}
\par\vspace{3pt}
\begin{minipage}{\textwidth}
\(\surd\), directly addressed; \(\triangle\), partially addressed, restricted, or indirectly obtainable;
\(\times\), not a primary feature; \(\text{--}\), not applicable.
WA, while-alive; Cov.-dep. cens., covariate-dependent censoring; DR/MR, doubly or multiply robust;
Event weights, prespecified clinically motivated weights for recurrent-event types or composite-event components;
CRT, cluster-randomized trial support, especially explicit while-alive cluster-average estimand and while-alive individual-average estimand.
\end{minipage}
\end{table}

\section{Exposure-weighted while-alive estimand and estimator in individually-randomized trials}
\label{sec:irt_estimand_estimation}

\subsection{Notation and estimands}\label{sec:irt_estimand}
We first consider an IRT with participants indexed by \(i=1,\ldots,n\). Let \(A_i\in\{0,1\}\) denote treatment assignment, where \(A_i=1\) denotes the intervention arm and \(A_i=0\) denotes the control arm, and let \(\bZ_i\in\R^p\) denote baseline covariates measured before randomization. Each participant may experience $K$ distinct types of recurrent events as well as a terminal event such as death. Under the potential outcomes framework \citep{rubin1974estimating,hernan2020causal}, let $D_i^a$ denote the potential terminal event time under treatment arm $a\in\{0,1\}$, and for event type $k=1,\ldots,K$, let $T_{ik1}^a<T_{ik2}^a<\cdots$ denote the potential occurrence times of type $k$ recurrent events. Following the counting process notation for event history \citep{andersen1993statistical}, define the potential terminal event counting process $N_i^{D,a}(t)=\ind(D_i^a\le t)$ and the potential at risk indicator $Y_i^a(t)=\ind(D_i^a\ge t)$. Because no recurrent events can occur after the terminal event, the potential type $k$ recurrent event counting process is $N_{ik}^a(t)=\sum_{j\ge 1}\ind(T_{ikj}^a\le t\wedge D_i^a)$, so that $N_{ik}^a(t)=N_{ik}^a(t\wedge D_i^a)$ \citep{cook2002analysis,cook2007statistical,ghosh2000nonparametric,schmidli2023estimands}. We collect the $K$ processes into the vector $\bN_i^a(t)=\{N_{i1}^a(t),\ldots,N_{iK}^a(t)\}^{\mathsf T}$, pre-specify a nonnegative weight vector $\bw=(w_1,\ldots,w_K)^{\mathsf T}$ with $w_k\ge 0$ encoding the relative clinical importance of the different event types \citep{mao2016semiparametric,mao2022whilealive}, and define the weighted recurrent event process $N_i^{a,\bw}(t)=\bw^{\mathsf T}\bN_i^a(t)=\sum_{k=1}^K w_k N_{ik}^a(t)$.

For a fixed time horizon \(\tau>0\), the two primitive arm-specific causal quantities are the marginal weighted recurrent event burden
\begin{equation}
  \mu_a(\tau;\bw)=\Ebb\{N_i^{a,\bw}(\tau)\},
  \label{eq:irt_mu}
\end{equation}
and the restricted mean survival time \(\nu_a(\tau)=\Ebb(D_i^a\wedge \tau)=\int_0^\tau S_a(t)\dd t\) with \(S_a(t)=\Pbb(D_i^a>t)\) \citep{uno2014moving}. 
The exposure-weighted while-alive recurrent-event rate is
\begin{equation}
  \psi_a(\tau;\bw)=\frac{\mu_a(\tau;\bw)}{\nu_a(\tau)},\label{eq:irt_psi}
\end{equation}
and the treatment contrast is \(\Delta(\tau;\bw)=\psi_1(\tau;\bw)-\psi_0(\tau;\bw)\).  The estimand $\psi_a(\tau;\bw)$ is the expected clinically weighted number of recurrent events accumulated by time $\tau$ per unit of restricted mean survival time under arm $a$. Because both the numerator and the denominator are population averages, $\psi_a(\tau;\bw)$ is an exposure-weighted while-alive estimand, and it differs from a patient-weighted while-alive estimand, which averages participant specific ratios of the form $N_i^{a,\bw}(\tau)/(D_i^a\wedge\tau)$ and thereby describes the event experience of the average participant \citep{schmidli2023estimands,wei2023properties,mao2022whilealive,ragni2026patient}. Because the occurrence of recurrent events is intrinsically tied to the time spent alive, the exposure-weighted version provides a natural summary of disease burden over the period during which events can occur, and it is the focus of this article. In addition, $\mu_a(\tau;\bw)$ and $\nu_a(\tau)$ remain separately interpretable marginal causal quantities, so reporting them alongside $\psi_a(\tau;\bw)$ gives a complete description of the treatment effect on the event and survival components.

The estimand $\psi_a(\tau;\bw)$ consists of two components, the weighted recurrent event burden $\mu_a(\tau;\bw)$ and the restricted mean survival time $\nu_a(\tau)$. We now show that both components, and hence the estimand itself, are determined by two marginal local quantities. Define the marginal terminal event hazard increment and the marginal weighted recurrent event rate increment among participants who remain alive as
\begin{align}
  \dd\Lambda_a^D(t)&=\frac{\Ebb\{\dd N_i^{D,a}(t)\}}{\Ebb\{Y_i^a(t)\}},\label{eq:irt_LamD}\\
  \dd\Lambda_a^{R,\bw}(t)&=\frac{\Ebb\{\dd N_i^{a,\bw}(t)\}}{\Ebb\{Y_i^a(t)\}},\label{eq:irt_LamR}
\end{align}
where $\dd\Lambda_a^D(t)$ is the instantaneous risk of the terminal event at time $t$ among participants alive at $t$, and $\dd\Lambda_a^{R,\bw}(t)$ is the expected weighted number of recurrent events at time $t$ per participant alive at $t$. Because recurrent events accrue only while participants remain alive, the numerator $\mu_a(\tau;\bw)$ accumulates the rate among survivors weighted by the probability of remaining alive. The following proposition formalizes these relations, with proof given in Appendix \ref{supp:sec:irt_dr_detailed}.
 
\begin{proposition}[Full data Nelson--Aalen representation]
\label{prop:full_data_na}
Suppose $\Ebb\{N_i^{a,\bw}(\tau)\}<\infty$, $\Ebb(D_i^a\wedge\tau)<\infty$, and $\inf_{0\le t\le\tau}R_a(t-)>0$, where $R_a(t)=\Ebb\{Y_i^a(t)\}$. Then
\[
  S_a(t)=\prod_{0<u\le t}\{1-\dd\Lambda_a^D(u)\},
  \qquad
  \nu_a(\tau)=\int_0^\tau S_a(t)\dd t,
  \qquad
  \mu_a(\tau;\bw)=\int_0^\tau S_a(t-)\dd\Lambda_a^{R,\bw}(t).
\]
\end{proposition}

The first identity recovers the survival function from the terminal event hazard through the product integral, which reduces to $S_a(t)=\exp\{-\Lambda_a^D(t)\}$ when $\Lambda_a^D$ is continuous \citep{gill1990product,andersen1993statistical}. The third identity builds the cumulative burden from local rate increments over the surviving population, in the same way that the Nelson--Aalen estimator builds a cumulative hazard from local increments over the at risk population \citep{nelson1972hazard,aalen1978nonparametric}, and we accordingly refer to it as a local Nelson--Aalen representation. Combining the three identities, the while-alive estimand admits the representation
\[\psi_a(\tau;\bw)=\frac{\int_0^\tau S_a(t-)\dd\Lambda_a^{R,\bw}(t)}{\int_0^\tau S_a(t)\dd t},\]
in which every element is a functional of $\Lambda_a^D$ and $\Lambda_a^{R,\bw}$. Estimation of $\psi_a(\tau;\bw)$ therefore reduces to estimation of these two local quantities. 

\subsection{Doubly robust estimation}\label{sec:irt_estimation}
In practice, the terminal event and recurrent event processes are subject to right censoring due to administrative end of follow up or loss to follow up. Let $C_i^a$ denote the potential censoring time for participant $i$ under treatment arm $a\in\{0,1\}$. We assume the following assumptions for identification.

\begin{assumption}[Consistency]
\label{asm:irt_consistency}
There is no interference between participants and no hidden versions of treatment. If participant \(i\) is assigned \(A_i=a\), then \(D_i=D_i^a\), \(C_i=C_i^a\), and \(\bN_i(t)=\bN_i^a(t)\) for all \(t\in[0,\tau]\).
\end{assumption}

\begin{assumption}[Randomization]
\label{asm:irt_randomization}
The treatment assignment probability is known by design and bounded away from zero. Write \(\pi_a=\Pbb(A_i=a)\) for \(a \in \{0,1\}\). Then \(0<\pi_a<1\), and \(  A_i\perp\!\!\!\perp\{D_i^0,D_i^1,C_i^0,C_i^1,\bN_i^0(\cdot),\bN_i^1(\cdot)\}\mid \bZ_i\). For simple equal randomization, \(\pi_0=\pi_1=1/2\).
\end{assumption}

Under Assumptions \ref{asm:irt_consistency} and \ref{asm:irt_randomization}, only the potential outcomes under the assigned arm are observed, subject to right censoring, with $D_i=D_i^{A_i}$, $C_i=C_i^{A_i}$, and $\bN_i(t)=\bN_i^{A_i}(t)$, and we write $T_{ikj}=T_{ikj}^{A_i}$ for the corresponding event times. The observed processes are restricted by censoring through the observed follow up time. Let $X_i=\min(D_i,C_i,\tau)$ denote the observed follow up time, with terminal event indicator $\delta_i=\ind(D_i\le C_i\wedge\tau)$ and censoring indicator $\delta_i^C=\ind(C_i<D_i\wedge\tau)$, where $\tau$ represents the target follow up window. For the terminal event, define the observed counting process $N_i^D(t)=\ind(X_i\le t,\delta_i=1)$, the observed censoring counting process $N_i^C(t)=\ind(X_i\le t,\delta_i^C=1)$, and the observed at risk indicator $Y_i(t)=\ind(X_i\ge t)$. For the recurrent events, the observed type $k$ counting process is $N_{ik}(t)=\sum_{j\ge 1}\ind(T_{ikj}\le t\wedge X_i)$, which equals $N_{ik}^{A_i}(t\wedge C_i)$ on $[0,\tau]$ because the potential process is stopped at the terminal event. Let $\bN_i(t)=\{N_{i1}(t),\ldots,N_{iK}(t)\}^{\mathsf T}$ and $N_i^{\bw}(t)=\bw^{\mathsf T}\bN_i(t)$, so that $\dd N_i^{\bw}(t)=\ind(C_i\ge t)\dd N_i^{A_i,\bw}(t)$. When the terminal event and censoring are tied, we adopt the convention that the terminal event precedes censoring, and the censoring risk indicator is $Y_i^\dagger(t)=\ind(X_i>t,\delta_i=1)+\ind(X_i\ge t,\delta_i\ne 1)$, which removes a participant from the censoring risk set at an observed terminal event. To ensure identifiability of the two local quantities under censoring, we impose the following assumption.

\begin{assumption}[Covariate dependent censoring]
\label{asm:irt_censoring}
For each treatment arm \(a\in\{0,1\}\), \(  C_i^a\perp\!\!\!\perp\{D_i^a,\bN_i^a(\cdot)\}\mid \bZ_i\), and the conditional censoring survival \(K_i^a(t)=\Pbb(C_i^a\ge t\mid\bZ_i)\) is uniformly bounded away from zero on \([0,\tau]\).
\end{assumption}

Assumption \ref{asm:irt_censoring} allows the censoring distribution to depend on the treatment arm and the baseline covariates, but requires censoring to be conditionally independent of the terminal and recurrent event processes given the covariates. The observed data are $\bO_i=(A_i,\bZ_i,X_i,\delta_i,\delta_i^C,\bN_i)$ for $i=1,\ldots,n$, assumed independent and identically distributed across participants, and we write $\Pn f=n^{-1}\sum_{i=1}^n f(\bO_i)$ for the empirical measure.

We now construct a doubly robust estimator of $\psi_a(\tau;\bw)$ by estimating the two local quantities in \eqref{eq:irt_LamD} and \eqref{eq:irt_LamR}. For the terminal event component, the augmented Nelson--Aalen estimating equation for the marginal terminal event hazard can be expressed as
\begin{equation}
\Pn\!\left[
\begin{aligned}
&\xi_i^a\{K_i^a(t-)\}^{-1}
\{\dd N_i^D(t)-Y_i(t)\dd\Lambda_a^D(t)\} \\
&\quad+\{1-\xi_i^a U_i^a(t)\}\{\dd q_i^D(t;a)-r_i^a(t)\dd\Lambda_a^D(t)\}
\end{aligned}
\right]=0.
\label{eq:irt_DR_eq_D}
\end{equation}
Each quantity in \eqref{eq:irt_DR_eq_D} is defined as follows. $\xi_i^a=\ind(A_i=a)/\pi_a$ is the inverse probability of treatment weight. Under Assumption \ref{asm:irt_randomization}, $\Ebb(\xi_i^a\mid\bZ_i)=1$, and the propensity score is known by design. $K_i^a(t)$ is the conditional censoring survival function introduced in Assumption \ref{asm:irt_censoring}, with conditional cumulative hazard $\Lambda_i^C(t;a)=-\log K_i^a(t)$. $U_i^a(t)$ is the censoring martingale augmentation process \(U_i^a(t)=1-\int_{(0,t)}\frac{\dd M_i^C(u;a)}{K_i^a(u-)H_i^a(u-)}\), where $\dd M_i^C(t;a)=\dd N_i^C(t)-Y_i^\dagger(t)\dd\Lambda_i^C(t;a)$ is the censoring martingale increment and $H_i^a(t)=\Pbb(D_i^a\ge t\mid\bZ_i)$ is the conditional terminal event survival function, with conditional cumulative hazard $\Lambda_i^D(t;a)$. This process mirrors the inverse probability of censoring weighting structure in the semiparametric theory for censored data \citep{robins1994estimation,robins2000correcting,tsiatis2006semiparametric}. Denote $r_i^a(t)$ and $\dd q_i^D(t;a)$ as the conditional at risk probability and the conditional terminal event increment, where \(r_i^a(t)=H_i^a(t-)=\Ebb\{Y_i^a(t)\mid\bZ_i\}\), and  \(\dd q_i^D(t;a)=H_i^a(t-)\dd\Lambda_i^D(t;a)=\Ebb\{\dd N_i^{D,a}(t)\mid\bZ_i\}\). Thus, in \eqref{eq:irt_DR_eq_D}, the first term is an inverse probability of censoring weighted version of the Nelson--Aalen estimating function, and the second term augments it with conditional full data increments to recover information from censored participants and to confer robustness. In particular, under Assumptions \ref{asm:irt_consistency} through \ref{asm:irt_censoring}, the estimating function in \eqref{eq:irt_DR_eq_D} has mean zero at the true $\dd\Lambda_a^D(t)$ when either the censoring related quantities or the terminal event related quantities are evaluated at their true values, which is the source of the double robustness established below.

In practice, the nuisance functions $K_i^a$, $H_i^a$, and $\Lambda_i^D(\cdot;a)$ are unknown. We replace them with estimates $\widehat K_i^a$, $\widehat H_i^a$, and $\widehat\Lambda_i^D(\cdot;a)$ obtained from working regression models described at the end of this section, and write $\widehat U_i^a$, $\widehat r_i^a$, and $\widehat q_i^D$ for the resulting plug in versions of $U_i^a$, $r_i^a$, and $q_i^D$. Substituting these estimates into \eqref{eq:irt_DR_eq_D} and solving for $\dd\Lambda_a^D(t)$ gives
\begin{equation}
\dd\widehat\Lambda_a^{D,\mathrm{DR}}(t)=\frac{\Pn\!\left[\xi_i^a\{\widehat K_i^a(t-)\}^{-1}\dd N_i^D(t)+\{1-\xi_i^a\widehat U_i^a(t)\}\dd\widehat q_i^D(t;a)\right]}{\Pn\!\left[\xi_i^a\{\widehat K_i^a(t-)\}^{-1}Y_i(t)+\{1-\xi_i^a\widehat U_i^a(t)\}\widehat r_i^a(t)\right]}.
\label{eq:irt_LamD_DR}
\end{equation}
Following the first two identities of Proposition \ref{prop:full_data_na}, the marginal survival function and restricted mean survival time estimators are
\begin{align}
  \widehat S_a^{\mathrm{DR}}(t)&=\prod_{0<u\le t}\{1-\dd\widehat\Lambda_a^{D,\mathrm{DR}}(u)\},\notag\\
  \widehat\nu_a^{\mathrm{DR}}(\tau)&=\int_0^\tau \widehat S_a^{\mathrm{DR}}(t-)\dd t.
  \label{eq:irt_S_nu_DR}
\end{align}

For the recurrent event component, we propose the analogous augmented Nelson--Aalen estimating equation for the marginal weighted recurrent event rate,
\begin{equation}
\Pn\!\left[
\begin{aligned}
&\xi_i^a\{K_i^a(t-)\}^{-1}\{\dd N_i^{\bw}(t)-Y_i(t)\dd\Lambda_a^{R,\bw}(t)\} \\
&\quad+\{1-\xi_i^a U_i^a(t)\}\{\dd q_i^{R,\bw}(t;a)-r_i^a(t)\dd\Lambda_a^{R,\bw}(t)\}
\end{aligned}
\right]=0,
\label{eq:irt_DR_eq_R}
\end{equation}
in which the only new quantity is the conditional weighted recurrent event increment \(\dd q_i^{R,\bw}(t;a)=H_i^a(t-)\dd\Lambda_i^{R,\bw}(t;a)=\Ebb\{\dd N_i^{a,\bw}(t)\mid\bZ_i\}\), where $\dd\Lambda_i^{R,\bw}(t;a)$ denotes the conditional weighted recurrent event rate among participants alive at $t$. Replacing the nuisance functions by their estimates, with $\widehat q_i^{R,\bw}$ denoting the plug in version of $q_i^{R,\bw}$, and solving \eqref{eq:irt_DR_eq_R} gives
\begin{equation}
\dd\widehat\Lambda_a^{R,\mathrm{DR},\bw}(t)=\frac{\Pn\!\left[\xi_i^a\{\widehat K_i^a(t-)\}^{-1}\dd N_i^{\bw}(t)+\{1-\xi_i^a\widehat U_i^a(t)\}\dd\widehat q_i^{R,\bw}(t;a)\right]}{\Pn\!\left[\xi_i^a\{\widehat K_i^a(t-)\}^{-1}Y_i(t)+\{1-\xi_i^a\widehat U_i^a(t)\}\widehat r_i^a(t)\right]}.
\label{eq:irt_LamR_DR}
\end{equation}
The denominators of \eqref{eq:irt_LamD_DR} and \eqref{eq:irt_LamR_DR} are identical, so the terminal event hazard and the weighted recurrent event rate are estimated on the same marginal alive risk set. Following the third identity of Proposition \ref{prop:full_data_na}, the doubly robust estimator of the recurrent event burden is
\[
  \widehat\mu_a^{\mathrm{DR}}(\tau;\bw)=\int_0^\tau\widehat S_a^{\mathrm{DR}}(t-)\dd\widehat\Lambda_a^{R,\mathrm{DR},\bw}(t).
\]
The arm specific while-alive rate and the treatment contrast at each time horizon $\tau$ are estimated by
\begin{align}
   \widehat\psi_a^{\mathrm{DR}}(\tau;\bw)&=\frac{\widehat\mu_a^{\mathrm{DR}}(\tau;\bw)}{\widehat\nu_a^{\mathrm{DR}}(\tau)},\notag\\
   \widehat\Delta^{\mathrm{DR}}(\tau;\bw)&=\widehat\psi_1^{\mathrm{DR}}(\tau;\bw)-\widehat\psi_0^{\mathrm{DR}}(\tau;\bw).
  \label{eq:irt_psi_delta_DR}
\end{align}

We now formalize the double robustness of the proposed estimator. Let $K_i^{a,\star}(t)$, $H_i^{a,\star}(t)$, $\dd q_i^{D,\star}(t;a)$, and $\dd q_i^{R,\bw,\star}(t;a)$ denote the uniform probability limits of $\widehat K_i^a(t)$, $\widehat H_i^a(t)$, $\dd\widehat q_i^D(t;a)$, and $\dd\widehat q_i^{R,\bw}(t;a)$, respectively. These limits coincide with the corresponding true conditional quantities under correct specification, and are otherwise pseudo true limits under misspecification. The censoring model is correctly specified for arm $a$, denoted by $\mathcal{C}_a$, if $K_i^{a,\star}(t)=K_i^a(t)$ uniformly on $[0,\tau]$. The terminal event outcome model is correctly specified, denoted by $\mathcal{O}_a^D$, if $H_i^{a,\star}(t)=H_i^a(t)$ and $\dd q_i^{D,\star}(t;a)=\dd q_i^D(t;a)$ uniformly on $[0,\tau]$. The recurrent event outcome model is correctly specified for the weighted full data increment, denoted by $\mathcal{O}_a^R$, if $\dd q_i^{R,\bw,\star}(t;a)=\dd q_i^{R,\bw}(t;a)$ uniformly on $[0,\tau]$.
 
\begin{theorem}[Componentwise double robustness]
\label{thm:irt_double_robustness}
Under Assumptions \ref{asm:irt_consistency}--\ref{asm:irt_censoring} and the regularity conditions in Section \ref{sec:irt_asymptotics}, the following properties hold for treatment arm $a$. If $\mathcal{C}_a$ or $\mathcal{O}_a^D$ holds, then $\widehat\nu_a^{\mathrm{DR}}(\tau)\overset{p}{\longrightarrow}\nu_a(\tau)$. If $\mathcal{C}_a$ or $(\mathcal{O}_a^D\cap\mathcal{O}_a^R)$ holds, then $\widehat\mu_a^{\mathrm{DR}}(\tau;\bw)\overset{p}{\longrightarrow}\mu_a(\tau;\bw)$. If, in addition, $\nu_a(\tau)>0$, then $\widehat\psi_a^{\mathrm{DR}}(\tau;\bw)\overset{p}{\longrightarrow}\psi_a(\tau;\bw)$. If the same condition holds for both treatment arms, then $\widehat\Delta^{\mathrm{DR}}(\tau;\bw)\overset{p}{\longrightarrow}\Delta(\tau;\bw)$.
\end{theorem}
 
The numerator and denominator have different robustness requirements because they depend on different full data quantities. The restricted mean survival time depends only on the terminal event distribution, so it is consistently estimated if either the censoring model or the terminal event outcome model is correct. The recurrent event burden depends on survival up to time $t$ and on the recurrent event process among those alive at time $t$, so outcome based consistency for the numerator requires joint correctness of the terminal event and recurrent event full data increments. Correct specification of the censoring model alone is sufficient for both components. The proof of Theorem \ref{thm:irt_double_robustness} follows by evaluating the probability limits of the two local estimators \eqref{eq:irt_LamD_DR} and \eqref{eq:irt_LamR_DR}. Under $\mathcal{C}_a$, the inverse probability weighted terms identify the full data local terminal event and recurrent event increments. Under the relevant outcome model conditions, the augmentation terms recover those same full data increments even if the censoring working model converges to a pseudo true limit. This is the censored data analogue of augmented inverse probability weighted double robustness \citep{robins1994estimation,bang2005doubly,vdlrobins2003unified,tsiatis2006semiparametric}. The details are provided in Appendix \ref{supp:sec:irt_dr_detailed}. 

The nuisance estimates entering \eqref{eq:irt_LamD_DR} and \eqref{eq:irt_LamR_DR} are generated by semiparametric working models for the survival and recurrent event outcomes \citep{cox1972regression,lin2000semiparametric}. For the censoring distribution, we fit an arm-specific working Cox model \citep{cox1972regression},
\[
  \lambda_C^a(t\mid\bZ_i)=\lambda_{C0}^a(t)\exp\{\balpha_C^{a\mathsf T}\bh_C(\bZ_i)\},
\]
where $\lambda_{C0}^a(t)$ is an unspecified baseline hazard function, $\balpha_C^a$ is an unknown regression coefficient vector, and $\bh_C(\cdot)$ is a prespecified vector of covariate functions, and let $\widehat\Lambda_i^C(t;a)$ and $\widehat K_i^a(t)=\exp\{-\widehat\Lambda_i^C(t;a)\}$ denote the fitted conditional cumulative hazard and survival function. For the terminal event, we similarly fit an arm specific working Cox model,
\[
  \lambda_D^a(t\mid\bZ_i)=\lambda_{D0}^a(t)\exp\{\bbeta_D^{a\mathsf T}\bh_D(\bZ_i)\},
\]
where $\lambda_{D0}^a(t)$, $\bbeta_D^a$, and $\bh_D(\cdot)$ are defined analogously, with fitted conditional cumulative hazard $\widehat\Lambda_i^D(t;a)$ and fitted conditional survival function $\widehat H_i^a(t)=\exp\{-\widehat\Lambda_i^D(t;a)\}$. 
For the recurrent event rates, we adopt an arm specific marginal proportional rates working model \citep{lin2000semiparametric}, also called the LWYY model. For event type $k=1,\ldots,K$,
\begin{equation}
  \Ebb\{\dd N_{ik}^{a}(t)\mid\bZ_i\}=H_i^a(t-)\exp\{\bgamma_{Rk}^{a\mathsf T}\bh_R(\bZ_i)\}\dd\Lambda_{Rk0}^a(t),
  \label{eq:irt_lwyy_model}
\end{equation}
where $\Lambda_{Rk0}^a(t)$ is an unspecified baseline cumulative rate function for type $k$ events, $\bgamma_{Rk}^a$ is an unknown regression coefficient vector, and $\bh_R(\cdot)$ is a prespecified vector of covariate functions. Unlike a conditional intensity model, the specification in \eqref{eq:irt_lwyy_model} concerns only the marginal mean increment of the recurrent event process given the baseline covariates, and therefore requires neither independent recurrent event increments nor a frailty distribution. The regression coefficients and the baseline cumulative rate function are estimated by the LWYY estimating equation and the associated Breslow estimator \citep{breslow1974covariance}. Let $\dd\widehat\Lambda_i^{R_k}(t;a)=\exp\{(\widehat{\bgamma}_{Rk}^{\,a})^{\mathsf T}\bh_R(\bZ_i)\}\dd\widehat\Lambda_{Rk0}^a(t)$ denote the fitted type $k$ increment, and define the fitted weighted increment $\dd\widehat\Lambda_i^{R,\bw}(t;a)=\sum_{k=1}^K w_k\dd\widehat\Lambda_i^{R_k}(t;a)$. The fitted quantities entering the estimators are then $\widehat r_i^a(t)=\widehat H_i^a(t-)$, $\dd\widehat q_i^D(t;a)=\widehat H_i^a(t-)\dd\widehat\Lambda_i^D(t;a)$, and $\dd\widehat q_i^{R,\bw}(t;a)=\widehat H_i^a(t-)\dd\widehat\Lambda_i^{R,\bw}(t;a)$. 
Regularity conditions for the fitted censoring, terminal event, and recurrent event working models, including their probability limits under possible misspecification, are given in Appendix \ref{supp:sec:irt_asymptotic_detailed}.

Our proposed estimator is closely related to the one step estimator of \citet{baer2025causal}. When the future mean $F_a(u,\tau\mid\bZ)$, namely $\Ebb\{\ind(D_i^a>u)N_i^{a,\bw}(\tau)\mid\bZ_i\}$, in their construction is induced by the same conditional terminal event survival and local recurrent event increment nuisance functions used here, the two procedures are first order equivalent. 
The formal statement, under which the differences between the corresponding estimators of $\mu_a(\tau;\bw)$, $\nu_a(\tau)$, and $\psi_a(\tau;\bw)$ are all $o_p(n^{-1/2})$, is given in Appendix \ref{supp:sec:baer_detailed_equivalence}. The two estimators nevertheless differ in finite samples, and the local Nelson--Aalen construction offers two practical advantages. First, the proposed estimators respect the range and shape constraints of their targets. Because $\widehat S_a^{\mathrm{DR}}$ is a product integral, it is a genuine nonincreasing survival function, so $\widehat\nu_a^{\mathrm{DR}}(\tau)$ is nondecreasing in $\tau$, and $\widehat\mu_a^{\mathrm{DR}}(\tau;\bw)$ is nonnegative and nondecreasing in $\tau$ whenever the estimated local rate increments are nonnegative. One step estimators, in contrast, may violate such global constraints across time horizons and require a subsequent isotonic projection to restore monotonicity \citep{baer2025causal}. Second, the proposed estimators are simple to compute. Every quantity is a closed form functional of fitted Cox and LWYY regressions evaluated at the observed event times, and estimates over all horizons $\tau\in(0,\tau_{\max}]$ are obtained in a single pass from the same local increments. The one step construction instead requires the future mean $F_a(u,\tau\mid\bZ)$ at every censoring time $u$, which entails fitting separate remaining outcome regressions over a grid of time points and is computationally burdensome without further approximation \citep{baer2025causal}.

\subsection{Asymptotic properties}
\label{sec:irt_asymptotics}

We next establish the asymptotic distribution of the proposed estimator, using counting process arguments and semiparametric estimating equation theory \citep{fleming1991counting,andersen1993statistical,tsiatis2006semiparametric,vandervaart1998asymptotic}. Throughout, the time horizon $\tau$, the treatment arm $a$, and the event type weight vector $\bw$ are fixed, and all limits are taken as $n\to\infty$. The regularity conditions for the asymptotic analysis, together with the probability limits of the two local estimators under possible working model misspecification, are stated in Appendix \ref{supp:sec:irt_asymptotic_detailed}. Let $\phi_{\Lambda^D,a,i}(t)$ and $\phi_{\Lambda^{R,\bw},a,i}(t)$ denote the influence processes of $\widehat\Lambda_a^{D,\mathrm{DR}}(t)$ and $\widehat\Lambda_a^{R,\mathrm{DR},\bw}(t)$, which include both the empirical estimating equation terms and the first order effects of estimating the censoring and outcome working models, and let $\phi_{\nu,a,i}$ and $\phi_{\mu,a,i}$ denote the induced influence functions of $\widehat\nu_a^{\mathrm{DR}}(\tau)$ and $\widehat\mu_a^{\mathrm{DR}}(\tau;\bw)$, obtained from the product integral delta method for the survival function and the Stieltjes integral representation of the recurrent event burden \citep{aalen1978nonparametric,gill1990product}. Explicit expressions for these influence functions, together with product integral versions allowing jumps in the terminal event cumulative hazard, are given in Appendices \ref{supp:sec:local_ratio_derivative_detailed} and \ref{supp:sec:functional_delta_detailed}.

\begin{theorem}[Asymptotic normality]
\label{thm:irt_asymptotic_normality}
Suppose Assumptions \ref{asm:irt_consistency}--\ref{asm:irt_censoring} and the regularity conditions in Appendix \ref{supp:sec:irt_asymptotic_detailed} hold. Suppose also that, for each treatment arm, $\mathcal{C}_a$ or $(\mathcal{O}_a^D\cap\mathcal{O}_a^R)$ holds, and that $\nu_a(\tau)>0$. Then
\begin{align*}
  n^{1/2}\left[
    \begin{pmatrix}
      \widehat\mu_a^{\mathrm{DR}}(\tau;\bw) \\
      \widehat\nu_a^{\mathrm{DR}}(\tau)
    \end{pmatrix}
    -
    \begin{pmatrix}
      \mu_a(\tau;\bw) \\
      \nu_a(\tau)
    \end{pmatrix}
  \right]
  &=n^{-1/2}\sum_{i=1}^n
  \begin{pmatrix}
    \phi_{\mu,a,i} \\
    \phi_{\nu,a,i}
  \end{pmatrix}
  +o_p(1).
\end{align*}
Thus, \(n^{1/2}\{\widehat\psi_a^{\mathrm{DR}}(\tau;\bw)-\psi_a(\tau;\bw)\}=n^{-1/2}\sum_{i=1}^n\phi_{\psi,a,i}+o_p(1)\), where $\phi_{\psi,a,i}=\frac{\phi_{\mu,a,i}}{\nu_a(\tau)}-\frac{\mu_a(\tau;\bw)}{\{\nu_a(\tau)\}^2}\phi_{\nu,a,i}$. For the treatment contrast, $\phi_{\Delta,i}=\phi_{\psi,1,i}-\phi_{\psi,0,i}$, and hence \(  n^{1/2}\{\widehat\Delta^{\mathrm{DR}}(\tau;\bw)-\Delta(\tau;\bw)\} = n^{-1/2}\sum_{i=1}^n\phi_{\Delta,i}+o_p(1)\). If $\sigma_\Delta^2=\Ebb(\phi_{\Delta,i}^2)$ is finite and nonzero, then \(  n^{1/2}
  \{\widehat\Delta^{\mathrm{DR}}(\tau;\bw)-\Delta(\tau;\bw)\}\overset{d}{\longrightarrow}\mathrm{N}(0,\sigma_\Delta^2)\). 
\end{theorem}
 Theorem \ref{thm:irt_asymptotic_normality} follows from the asymptotic linearity of the nuisance estimators, the first order expansion of the two local ratio estimators, the differentiability of the product integral survival map, and the delta method for the ratio $(\mu,\nu)\mapsto\mu/\nu$ \citep{andersen1993statistical,tsiatis2006semiparametric,vandervaart1998asymptotic}, with the full derivation given in Appendices \ref{supp:sec:irt_asymptotic_detailed} through \ref{supp:sec:functional_delta_detailed}.
 
\begin{corollary}[Variance estimation]
\label{cor:irt_variance}
Let $\widehat\phi_{\Delta,i}$ be the estimated influence function obtained by replacing unknown quantities in $\phi_{\Delta,i}$ with their sample analogues and fitted nuisance quantities \citep{tsiatis2006semiparametric,vandervaart1998asymptotic}. Define $\overline{\widehat\phi}_{\Delta}=n^{-1}\sum_{i=1}^n\widehat\phi_{\Delta,i}$ and $\widehat\sigma_\Delta^2=\frac{1}{n-1}\sum_{i=1}^n\left(\widehat\phi_{\Delta,i}-\overline{\widehat\phi}_{\Delta}\right)^2$. Under the conditions of Theorem \ref{thm:irt_asymptotic_normality}, $\widehat{\operatorname{se}}\{\widehat\Delta^{\mathrm{DR}}(\tau;\bw)\}=\left(\widehat\sigma_\Delta^2/n\right)^{1/2}$ is a consistent standard error estimator, and an asymptotic $95\%$ Wald confidence interval is
\[
  \widehat\Delta^{\mathrm{DR}}(\tau;\bw)\pm1.96\,\widehat{\operatorname{se}}\{\widehat\Delta^{\mathrm{DR}}(\tau;\bw)\}.
\]
The same construction yields standard errors for $\widehat\psi_a^{\mathrm{DR}}(\tau;\bw)$, $\widehat\mu_a^{\mathrm{DR}}(\tau;\bw)$, and $\widehat\nu_a^{\mathrm{DR}}(\tau)$.
\end{corollary}
 
The influence function in Theorem \ref{thm:irt_asymptotic_normality} is the first order influence function of the proposed local Nelson--Aalen estimator, and is not, in general, the semiparametric efficient influence function in the unrestricted recurrent event model. The efficient influence function involves conditional expectations of the remaining recurrent event burden and the remaining survival time given the full event history available just before censoring, whereas the proposed estimator replaces these history specific regressions with baseline regressions among participants who remain alive, induced by the terminal event and recurrent event working models. This replacement is a deliberate tradeoff. The history specific remaining outcome regressions are high dimensional, must be estimated at every censoring time, and require the continuous time recurrent event history, information that is not always recorded or readily available in practice \citep{baer2025causal}. The baseline regressions used here are standard Cox and LWYY fits that are stable in moderate samples, and they preserve the componentwise double robustness and the closed form computation of the proposed estimator. The efficiency cost is confined to settings in which the past recurrent event history carries additional prognostic information, beyond the baseline covariates and survival status, for the remaining numerator and denominator outcomes. Under the corresponding history sufficiency condition, the first order influence function coincides with the efficient influence function and the proposed estimator attains the semiparametric efficiency bound. The exact efficient influence function and the resulting efficiency gap are given in Appendix \ref{supp:sec:eif}.

\section{Extension to cluster-randomized trials}\label{sec:crt}

\subsection{Notation and estimands}\label{sec:crt_estimand}

We now extend the while-alive construction to cluster randomized trials. In contrast to the IRT setting, the independent unit of randomization and sampling is the cluster, whereas the recurrent event and terminal event processes are measured on participants within clusters. This distinction affects both the causal estimand and the unit used for inference. When cluster size is associated with prognosis, survival, recurrent event rates, or treatment effect heterogeneity, the average causal effect for the population of participants and that for the population of clusters are generally different causal targets \citep{kahan2023estimands,kahan2023ics,kahan2024demystifying,li2025standardization,fang2026estimands}. The CRT formulation therefore requires a matched pair of while-alive estimands, an individual average estimand that gives equal weight to participants in the target population, and a cluster average estimand that gives equal weight to clusters in the target population. Consider \(M\) independent clusters indexed by \(i=1,\ldots,M\), where cluster \(i\) contains \(n_i\) participants indexed by \(j=1,\ldots,n_i\). Let \(A_i\in\{0,1\}\) denote the cluster assigned treatment, with \(A_i=1\) for the intervention arm and \(A_i=0\) for the control arm. Let \(\mathbf L_i\) denote baseline cluster level covariates and \(\bZ_{ij}\) baseline individual level covariates, and write \(\mathcal V_{ij}=(n_i,\mathbf L_i,\bZ_{ij})\) for the all baseline information. Including \(n_i\) in \(\mathcal V_{ij}\) allows us to account for informative cluster size, because cluster size may carry prognostic information and may modify treatment effects \citep{kahan2023ics,kahan2024demystifying,li2025standardization}. For treatment arm \(a\in\{0,1\}\), let \(D_{ij}^a\) denote the potential terminal event time, and define the potential terminal event counting process \(N_{ij}^{D,a}(t)=\ind(D_{ij}^a\le t)\) and the potential at risk indicator \(Y_{ij}^a(t)=\ind(D_{ij}^a\ge t)\). For event type \(k=1,\ldots,K\), let \(T_{ijk1}^a<T_{ijk2}^a<\cdots\) denote the potential occurrence times of type \(k\) recurrent events for participant \(j\) in cluster \(i\). As in Section \ref{sec:irt_estimand}, no recurrent events can occur after the terminal event, so the potential type \(k\) recurrent event counting process is \(N_{ijk}^a(t)=\sum_{l\ge 1}\ind(T_{ijkl}^a\le t\wedge D_{ij}^a)\), which satisfies \(N_{ijk}^a(t)=N_{ijk}^a(t\wedge D_{ij}^a)\). We collect the \(K\) processes into the vector \(\bN_{ij}^a(t)=\{N_{ij1}^a(t),\ldots,N_{ijK}^a(t)\}^{\mathsf T}\), and with the same pre-specified weight vector \(\bw=(w_1,\ldots,w_K)^{\mathsf T}\) as in Section \ref{sec:irt_estimand}, define the weighted recurrent event process \(N_{ij}^{a,\bw}(t)=\bw^{\mathsf T}\bN_{ij}^a(t)=\sum_{k=1}^K w_kN_{ijk}^a(t)\).

To formalize the while-alive cluster-average estimand and the while-alive individual-average estimand within a single notation, we introduce a target population index \(\ell\in\{\mathrm{clus},\mathrm{ind}\}\) and define the target weight
\(\omega_i^\ell\), which equals \(1\) when \(\ell=\mathrm{clus}\) (the cluster-average target) and \(n_i\) when \(\ell=\mathrm{ind}\) (the individual-average target). The choice \(\omega_i^{\mathrm{clus}}=1\) first averages participants within a cluster and then averages clusters equally, so that each cluster contributes equally to the estimand. The choice \(\omega_i^{\mathrm{ind}}=n_i\) weights each cluster in proportion to its size, so that each participant contributes equally. This weighting formulation follows the estimand distinction used in model robust CRT standardization and CRT survival estimands \citep{li2025standardization,fang2026estimands}. For fixed \(\tau>0\), define
\begin{align}
  \mu_{a,\ell}(\tau;\bw)&=\frac{\Ebb\left[{\omega_i^\ell}\sum_{j=1}^{n_i}N_{ij}^{a,\bw}(\tau)/n_i\right]}{\Ebb(\omega_i^\ell)},
  \label{eq:crt_mu}\\
  \nu_{a,\ell}(\tau)&=\frac{\Ebb\left[{\omega_i^\ell}\sum_{j=1}^{n_i}(D_{ij}^a\wedge \tau)/n_i\right]}{\Ebb(\omega_i^\ell)}.
  \label{eq:crt_nu}
\end{align}
In both definitions, the inner average \(n_i^{-1}\sum_{j=1}^{n_i}\) is the within cluster mean of the participant-level outcome, and the outer ratio \(\Ebb(\omega_i^\ell\,\cdot)/\Ebb(\omega_i^\ell)\) averages these cluster-level summaries over the target population. Thus \(\mu_{a,\ell}(\tau;\bw)\) is the target specific marginal weighted recurrent event burden and \(\nu_{a,\ell}(\tau)\) is the corresponding restricted mean survival time. The exposure-weighted while-alive recurrent event rate and the treatment contrast are
\begin{equation*}
  \psi_{a,\ell}(\tau;\bw)=\frac{\mu_{a,\ell}(\tau;\bw)}{\nu_{a,\ell}(\tau)},\qquad
  \Delta_{\ell}(\tau;\bw)=\psi_{1,\ell}(\tau;\bw)-\psi_{0,\ell}(\tau;\bw).
\end{equation*}
We refer to \(\Delta_{\mathrm{clus}}(\tau;\bw)\) as the treatment effect for the while-alive cluster-average estimand and to \(\Delta_{\mathrm{ind}}(\tau;\bw)\) as the treatment effect for the while-alive individual-average estimand. The former compares the average cluster under intervention with the average cluster under control, whereas the latter compares the average participant under intervention with the average participant under control. If cluster sizes are constant, or if cluster size is unrelated to the relevant potential outcome distribution, the two targets may coincide, and when every cluster has size \(n_i=1\), both reduce to the IRT estimands in \eqref{eq:irt_mu} through \eqref{eq:irt_psi}. Under informative cluster size, however, the two estimands answer different causal questions and should not be used interchangeably \citep{kahan2023estimands,kahan2023ics,kahan2024demystifying}.

As in the IRT setting, the two components of \(\psi_{a,\ell}(\tau;\bw)\) are determined by two marginal local quantities, now averaged over clusters according to \(\omega_i^\ell\). For target \(\ell\), define the \(\ell\) weighted marginal terminal event hazard increment and the \(\ell\) weighted marginal recurrent event rate increment among participants who remain alive as
\begin{align*}
  \dd\Lambda_{a,\ell}^D(t)=\frac{\Ebb\left[\frac{\omega_i^\ell}{n_i}\sum_{j=1}^{n_i}\dd N_{ij}^{D,a}(t)\right]}{\Ebb\left[\frac{\omega_i^\ell}{n_i}\sum_{j=1}^{n_i}Y_{ij}^a(t)\right]},\qquad
  \dd\Lambda_{a,\ell}^{R,\bw}(t)=\frac{\Ebb\left[\frac{\omega_i^\ell}{n_i}\sum_{j=1}^{n_i}\dd N_{ij}^{a,\bw}(t)\right]}{\Ebb\left[\frac{\omega_i^\ell}{n_i}\sum_{j=1}^{n_i}Y_{ij}^a(t)\right]}.
\end{align*}
Let \(S_{a,\ell}(t)=\Ebb\left[\frac{\omega_i^\ell}{n_i}\sum_{j=1}^{n_i}Y_{ij}^a(t)\right]/\Ebb(\omega_i^\ell)\) denote the \(\ell\) weighted marginal survival function, namely the proportion of the target population remaining alive at time \(t\), which is recovered from \(\Lambda_{a,\ell}^D\) through the product integral \(S_{a,\ell}(t)=\prod_{0<u\le t}\{1-\dd\Lambda_{a,\ell}^D(u)\}\). Here, \(\dd\Lambda_{a,\ell}^D(t)\) is the instantaneous risk of the terminal event at time \(t\) among participants alive at \(t\) in the target population, and \(\dd\Lambda_{a,\ell}^{R,\bw}(t)\) is the corresponding expected weighted number of recurrent events at time \(t\) per participant alive at \(t\). Paralleling Proposition \ref{prop:full_data_na},
\begin{align}
  \nu_{a,\ell}(\tau)&=\int_0^\tau S_{a,\ell}(t)\dd t,\label{eq:crt_nu_NA}\\
  \mu_{a,\ell}(\tau;\bw)
  &=\int_0^\tau S_{a,\ell}(t-)\dd\Lambda_{a,\ell}^{R,\bw}(t).\label{eq:crt_mu_NA}
\end{align}
The CRT estimands therefore have the same local Nelson--Aalen structure as in the IRT setting, with alive risk sets and event increments averaged over clusters according to \(\omega_i^\ell\) \citep{aalen1978nonparametric,andersen1993statistical,li2025standardization,fang2026estimands}. The derivation of \eqref{eq:crt_nu_NA} and \eqref{eq:crt_mu_NA} is given in Appendix \ref{supp:sec:crt_dr_detailed}, and estimation of \(\psi_{a,\ell}(\tau;\bw)\) again reduces to estimation of the two local quantities.

\subsection{Doubly robust estimation}\label{sec:crt_estimation}

In practice, the terminal event and recurrent event processes are subject to right censoring. Let \(C_{ij}^a\) denote the potential censoring time for participant \(j\) in cluster \(i\) under treatment arm \(a\in\{0,1\}\). We impose the following assumptions for identification.

\begin{assumption}[Consistency]
\label{asm:crt_consistency}
There is no interference between clusters and no hidden versions of the cluster level treatment. If cluster \(i\) is assigned \(A_i=a\), then, for each \(j=1,\ldots,n_i\), \(D_{ij}=D_{ij}^a\), \(C_{ij}=C_{ij}^a\), and \(\bN_{ij}(t)=\bN_{ij}^a(t)\) for all \(t\in[0,\tau]\). Participants within the same cluster may be arbitrarily dependent.
\end{assumption}

\begin{assumption}[Cluster randomization]
\label{asm:crt_randomization}
For each cluster \(i\) and treatment arm \(a\in\{0,1\}\), let \(\pi_i^a\) denote the known randomization probability, satisfying \(\pi_i^0+\pi_i^1=1\) and \(0<\epsilon_\pi\le \pi_i^a\le 1-\epsilon_\pi<1\). The cluster treatment assignment satisfies
\(A_i\perp\!\!\!\perp\left\{D_{ij}^0,D_{ij}^1,C_{ij}^0,C_{ij}^1,\bN_{ij}^0(\cdot),\bN_{ij}^1(\cdot):j=1,\ldots,n_i\right\}\mid n_i,\mathbf L_i,\bZ_{i1},\ldots,\bZ_{in_i}\).
\end{assumption}

The randomization probability \(\pi_i^a\) is known from the trial design and is not estimated in the primary analysis. Under simple cluster randomization, \(\pi_i^a=\pi_a\) is constant, and under equal allocation, \(\pi_i^0=\pi_i^1=1/2\). Under stratified, blocked, pair matched, or covariate adaptive cluster randomization, \(\pi_i^a\) denotes the corresponding known design probability and may vary across clusters according to the prespecified randomization scheme. The conditional independence statement in Assumption \ref{asm:crt_randomization} allows the randomization design to use baseline cluster level and individual level covariates, while excluding dependence of treatment assignment on the potential event histories after conditioning on those baseline variables.

Under Assumptions \ref{asm:crt_consistency} and \ref{asm:crt_randomization}, only the potential outcomes under the assigned arm are observed, subject to right censoring, with \(D_{ij}=D_{ij}^{A_i}\), \(C_{ij}=C_{ij}^{A_i}\), and \(\bN_{ij}(t)=\bN_{ij}^{A_i}(t)\). The observed quantities are defined exactly as in Section \ref{sec:irt_estimation}, applied to each participant within each cluster. Let \(X_{ij}=\min(D_{ij},C_{ij},\tau)\) denote the observed follow up time, with terminal event indicator \(\delta_{ij}=\ind(D_{ij}\le C_{ij}\wedge\tau)\) and censoring indicator \(\delta_{ij}^C=\ind(C_{ij}<D_{ij}\wedge\tau)\). For the terminal event, define \(N_{ij}^D(t)=\ind(X_{ij}\le t,\delta_{ij}=1)\), \(N_{ij}^C(t)=\ind(X_{ij}\le t,\delta_{ij}^C=1)\), and \(Y_{ij}(t)=\ind(X_{ij}\ge t)\). For the recurrent events, the observed type \(k\) counting process is \(N_{ijk}(t)=N_{ijk}^{A_i}(t\wedge C_{ij})\), and the observed weighted process is \(N_{ij}^{\bw}(t)=\bw^{\mathsf T}\bN_{ij}(t)\) with \(\bN_{ij}(t)=\{N_{ij1}(t),\ldots,N_{ijK}(t)\}^{\mathsf T}\). When the terminal event and censoring are tied, the terminal event precedes censoring, and the censoring risk indicator is \(Y_{ij}^\dagger(t)=\ind(X_{ij}>t,\delta_{ij}=1)+\ind(X_{ij}\ge t,\delta_{ij}\ne1)\). To ensure identifiability of the two local quantities under censoring, we impose the following assumption.

\begin{assumption}[Covariate dependent censoring]
\label{asm:crt_censoring}
Let \(\mathcal F_i=(n_i,\mathbf L_i,\bZ_{i1},\ldots,\bZ_{in_i})\) denote the full baseline information in cluster \(i\). For each arm \(a\in\{0,1\}\), the joint censoring process is independent of the joint terminal event and recurrent event processes conditional on \(\mathcal F_i\), namely \(\{C_{i1}^a,\ldots,C_{in_i}^a\}\perp\!\!\!\perp\{D_{ij}^a,\bN_{ij}^a(\cdot):j=1,\ldots,n_i\}\mid\mathcal F_i\). The conditional censoring survival function \(K_{ij}^a(t)=\Pbb(C_{ij}^a\ge t\mid\mathcal F_i)\) is uniformly bounded away from zero on \([0,\tau]\).
\end{assumption}

Assumption \ref{asm:crt_censoring} is an individual level identifying condition within a clustered design. It does not require independence of event times, censoring times, or recurrent event processes among participants in the same cluster. Such dependence is addressed by treating clusters as the independent sampling units in estimation and inference \citep{balzer2019hierarchical,balzer2023twostage,wang2024modelrobust,li2025standardization,fang2026estimands}. The observed data for cluster \(i\) are \(\bO_i=\left(A_i,n_i,\mathbf L_i,\{\bZ_{ij},X_{ij},\delta_{ij},\delta_{ij}^C,\bN_{ij}:j=1,\ldots,n_i\}\right)\), assumed independent across clusters, and the empirical measure in the CRT setting is over clusters, \(\mathbb P_M f=M^{-1}\sum_{i=1}^M f(\bO_i)\).

We now construct doubly robust estimators of \(\psi_{a,\ell}(\tau;\bw)\) by estimating the two local quantities, paralleling the construction in Section \ref{sec:irt_estimation} with participant level averages replaced by \(\ell\) weighted cluster level averages. Define the inverse probability of treatment weight \(\xi_i^a=\ind(A_i=a)/\pi_i^a\). Under Assumption \ref{asm:crt_randomization}, \(\Ebb\left(\xi_i^a\mid n_i,\mathbf L_i,\bZ_{i1},\ldots,\bZ_{in_i}\right)=1\). The nuisance functions are the cluster analogues of those in Section \ref{sec:irt_estimation}. Let \(H_{ij}^a(t)\) denote the conditional terminal event survival function with conditional cumulative hazard \(\Lambda_{ij}^D(t;a)\), let \(\Lambda_{ij}^C(t;a)=-\log K_{ij}^a(t)\), and define the censoring martingale increment \(\dd M_{ij}^C(t;a)=\dd N_{ij}^C(t)-Y_{ij}^\dagger(t)\dd\Lambda_{ij}^C(t;a)\), the augmentation process \(U_{ij}^a(t)=1-\int_{(0,t)}\dd M_{ij}^C(u;a)/\{K_{ij}^a(u-)H_{ij}^a(u-)\}\), the conditional at risk probability \(r_{ij}^a(t)=H_{ij}^a(t-)\), and the conditional full data increments \(\dd q_{ij}^D(t;a)=H_{ij}^a(t-)\dd\Lambda_{ij}^D(t;a)\) and \(\dd q_{ij}^{R,\bw}(t;a)=H_{ij}^a(t-)\dd\Lambda_{ij}^{R,\bw}(t;a)\), where \(\dd\Lambda_{ij}^{R,\bw}(t;a)\) is the conditional weighted recurrent event rate among participants alive at \(t\). For a fixed target \(\ell\), the augmented Nelson--Aalen estimating equation for the \(\ell\) weighted marginal terminal event hazard is
\begin{equation}
\mathbb P_M\!\left[\frac{\omega_i^\ell}{n_i}\sum_{j=1}^{n_i}\left\{
  \begin{aligned}
  &\xi_i^a\{K_{ij}^a(t-)\}^{-1}\{\dd N_{ij}^D(t)-Y_{ij}(t)\dd\Lambda_{a,\ell}^D(t)\} \\
  &\quad+\{1-\xi_i^a U_{ij}^a(t)\}\{\dd q_{ij}^D(t;a)-r_{ij}^a(t)\dd\Lambda_{a,\ell}^D(t)\}
  \end{aligned}
  \right\}\right]=0,
\label{eq:crt_DR_eq_D}
\end{equation}
and the analogous equation for the \(\ell\) weighted marginal recurrent event rate is
\begin{equation}
\mathbb P_M\!\left[\frac{\omega_i^\ell}{n_i}\sum_{j=1}^{n_i}\left\{\begin{aligned}
  &\xi_i^a\{K_{ij}^a(t-)\}^{-1}
  \{\dd N_{ij}^{\bw}(t)-Y_{ij}(t)\dd\Lambda_{a,\ell}^{R,\bw}(t)\} \\
  &\quad+\{1-\xi_i^a U_{ij}^a(t)\}\{\dd q_{ij}^{R,\bw}(t;a)-r_{ij}^a(t)\dd\Lambda_{a,\ell}^{R,\bw}(t)\}
  \end{aligned}
  \right\}\right]=0.
\label{eq:crt_DR_eq_R}
\end{equation}
Within each cluster, the summands have the same structure as in the IRT setting, an inverse probability of censoring weighted estimating function augmented by conditional full data increments, and the target weight \(\omega_i^\ell/n_i\) matches the empirical average to the chosen population.

We now formalize the double robustness of the proposed estimators. Let \(K_{ij}^{a,\star}(t)\), \(H_{ij}^{a,\star}(t)\), \(\dd q_{ij}^{D,\star}(t;a)\), and \(\dd q_{ij}^{R,\bw,\star}(t;a)\) denote the uniform probability limits of \(\widehat K_{ij}^a(t)\), \(\widehat H_{ij}^a(t)\), \(\dd\widehat q_{ij}^D(t;a)\), and \(\dd\widehat q_{ij}^{R,\bw}(t;a)\), respectively. These limits coincide with the corresponding true conditional quantities defined earlier in this subsection under correct specification, and are otherwise pseudo true limits under misspecification. The censoring model is correctly specified for arm \(a\), denoted by \(\mathcal{C}_a\), if \(K_{ij}^{a,\star}(t)=K_{ij}^a(t)\) uniformly on \([0,\tau]\). The terminal event outcome model is correctly specified, denoted by \(\mathcal{O}_a^D\), if \(H_{ij}^{a,\star}(t)=H_{ij}^a(t)\) and \(\dd q_{ij}^{D,\star}(t;a)=\dd q_{ij}^D(t;a)\) uniformly on \([0,\tau]\). The recurrent event outcome model is correctly specified for the weighted full data increment, denoted by \(\mathcal{O}_a^R\), if \(\dd q_{ij}^{R,\bw,\star}(t;a)=\dd q_{ij}^{R,\bw}(t;a)\) uniformly on \([0,\tau]\).
 
\begin{theorem}[Componentwise double robustness]
\label{thm:crt_double_robustness}
Under Assumptions \ref{asm:crt_consistency}--\ref{asm:crt_censoring} and the regularity conditions in Appendix \ref{supp:sec:crt_variance_detailed}, the following properties hold for treatment arm \(a\) and target population \(\ell\). If \(\mathcal{C}_a\) or \(\mathcal{O}_a^D\) holds, then \(\widehat\nu_{a,\ell}^{\mathrm{DR}}(\tau)\overset{p}{\longrightarrow}\nu_{a,\ell}(\tau)\). If \(\mathcal{C}_a\) or \((\mathcal{O}_a^D\cap\mathcal{O}_a^R)\) holds, then \(\widehat\mu_{a,\ell}^{\mathrm{DR}}(\tau;\bw)\overset{p}{\longrightarrow}\mu_{a,\ell}(\tau;\bw)\). If, in addition, \(\nu_{a,\ell}(\tau)>0\), then \(\widehat\psi_{a,\ell}^{\mathrm{DR}}(\tau;\bw)\overset{p}{\longrightarrow}\psi_{a,\ell}(\tau;\bw)\). If the same condition holds for both treatment arms, then \(\widehat\Delta_{\ell}^{\mathrm{DR}}(\tau;\bw)\overset{p}{\longrightarrow}\Delta_{\ell}(\tau;\bw)\).
\end{theorem}
 
\begin{remark}[Robustness across target populations]
The double robustness statement is the same for the while-alive cluster-average estimand and the while-alive individual-average estimand, but the estimands are not the same. The weight \(\omega_i^\ell\) determines which population is averaged over in \eqref{eq:crt_mu} and \eqref{eq:crt_nu}. Thus, under informative cluster size, consistency for \(\ell=\mathrm{ind}\) and consistency for \(\ell=\mathrm{clus}\) refer to different causal quantities, even though the same fitted nuisance functions can be used in both estimators.
\end{remark}
 
The proof of Theorem \ref{thm:crt_double_robustness} follows by evaluating the probability limits of the two local estimators \eqref{eq:crt_LamD_DR} and \eqref{eq:crt_LamR_DR}. Under \(\mathcal{C}_a\), the inverse probability weighted cluster level estimating functions identify the target weighted full data increments. Under the relevant outcome model conditions, the augmentation terms identify those same increments even when the censoring working model converges to a pseudo true limit. The detailed proof is provided in Appendix \ref{supp:sec:crt_dr_detailed}, and the asymptotic distribution of the proposed estimators is established in Section \ref{sec:crt_asymptotics}.

In practice, the nuisance functions are unknown and are replaced by estimates from working regression models described at the end of this section. Let \(\widehat K_{ij}^a\), \(\widehat H_{ij}^a\), \(\widehat U_{ij}^a\), \(\widehat r_{ij}^a\), \(\widehat q_{ij}^D\), and \(\widehat q_{ij}^{R,\bw}\) denote the resulting plug in versions. Substituting these estimates into \eqref{eq:crt_DR_eq_D} and solving for \(\dd\Lambda_{a,\ell}^D(t)\) gives
\begin{equation}
\dd\widehat\Lambda_{a,\ell}^{D,\mathrm{DR}}(t)=\frac{\mathbb P_M\!\left[\frac{\omega_i^\ell}{n_i}\sum_{j=1}^{n_i}\left\{\xi_i^a\{\widehat K_{ij}^a(t-)\}^{-1}\dd N_{ij}^D(t)+\{1-\xi_i^a\widehat U_{ij}^a(t)\}\dd\widehat q_{ij}^D(t;a)\right\}\right]}{\mathbb P_M\!\left[\frac{\omega_i^\ell}{n_i}\sum_{j=1}^{n_i}\left\{\xi_i^a\{\widehat K_{ij}^a(t-)\}^{-1}Y_{ij}(t)+\{1-\xi_i^a\widehat U_{ij}^a(t)\}\widehat r_{ij}^a(t)\right\}\right]},
\label{eq:crt_LamD_DR}
\end{equation}
with the corresponding survival and restricted mean survival time estimators
\begin{align}
  \widehat S_{a,\ell}^{\mathrm{DR}}(t)&=\prod_{0<u\le t}\{1-\dd\widehat\Lambda_{a,\ell}^{D,\mathrm{DR}}(u)\},\notag\\
  \widehat\nu_{a,\ell}^{\mathrm{DR}}(\tau)&=\int_0^\tau\widehat S_{a,\ell}^{\mathrm{DR}}(t-)\dd t.
  \label{eq:crt_S_nu_DR}
\end{align}
Similarly, substituting the estimates into \eqref{eq:crt_DR_eq_R} and solving gives
\begin{equation}
\dd\widehat\Lambda_{a,\ell}^{R,\mathrm{DR},\bw}(t)=\frac{\mathbb P_M\!\left[
\frac{\omega_i^\ell}{n_i}\sum_{j=1}^{n_i}\left\{\xi_i^a\{\widehat K_{ij}^a(t-)\}^{-1}\dd N_{ij}^{\bw}(t)+\{1-\xi_i^a\widehat U_{ij}^a(t)\}\dd\widehat q_{ij}^{R,\bw}(t;a)\right\}\right]}{\mathbb P_M\!\left[\frac{\omega_i^\ell}{n_i}\sum_{j=1}^{n_i}\left\{\xi_i^a\{\widehat K_{ij}^a(t-)\}^{-1}Y_{ij}(t)+\{1-\xi_i^a\widehat U_{ij}^a(t)\}\widehat r_{ij}^a(t)\right\}\right]}.
\label{eq:crt_LamR_DR}
\end{equation}
For a given target \(\ell\), the denominators of \eqref{eq:crt_LamD_DR} and \eqref{eq:crt_LamR_DR} are identical, so the terminal event hazard and the recurrent event rate are estimated on the same \(\ell\) weighted alive risk set. Setting \(\ell=\mathrm{clus}\) yields the estimator of the cluster-average while-alive treatment effect, and setting \(\ell=\mathrm{ind}\) yields the estimator of the individual-average while-alive treatment effect. Moreover, the two estimators differ only through the choice of \(\omega_i^\ell\), so the same fitted nuisance functions can be used for both estimands. Combining \eqref{eq:crt_LamR_DR} with \eqref{eq:crt_S_nu_DR}, the estimated weighted recurrent event burden is \(\widehat\mu_{a,\ell}^{\mathrm{DR}}(\tau;\bw)=\int_0^\tau\widehat S_{a,\ell}^{\mathrm{DR}}(t-)\dd\widehat\Lambda_{a,\ell}^{R,\mathrm{DR},\bw}(t)\), and the arm-specific while-alive rate and treatment contrast are
\begin{align*}
  \widehat\psi_{a,\ell}^{\mathrm{DR}}(\tau;\bw)&=\frac{\widehat\mu_{a,\ell}^{\mathrm{DR}}(\tau;\bw)}{\widehat\nu_{a,\ell}^{\mathrm{DR}}(\tau)},\notag\\
  \widehat\Delta_{\ell}^{\mathrm{DR}}(\tau;\bw)&=\widehat\psi_{1,\ell}^{\mathrm{DR}}(\tau;\bw)-\widehat\psi_{0,\ell}^{\mathrm{DR}}(\tau;\bw).
\end{align*}

This construction preserves the central feature of the IRT estimator, namely augmentation of inverse probability weighted local Nelson--Aalen estimating equations, while replacing participant-level empirical averages by cluster-level empirical averages matched to the target population. The estimating equations are evaluated through memberwise contributions under working independence, but this working convention is only a device for constructing feasible local estimating equations and nuisance estimators. It does not impose independence among participants or recurrent event increments within a cluster. Within cluster dependence is accommodated by forming cluster-level nuisance influence contributions and treating clusters as the independent units for inference, so the resulting theory is based on cluster-level first order influence functions. The asymptotic normality of the proposed estimators is established in Section \ref{sec:crt_asymptotics}.

The nuisance estimates entering \eqref{eq:crt_LamD_DR} and \eqref{eq:crt_LamR_DR} are generated by the same working model classes used in the IRT estimator, applied to \(\mathcal V_{ij}\) and fitted with arm specific models among participants in clusters assigned to arm \(a\). For the censoring and terminal event distributions, we fit arm specific Cox working models under working independence,
\[
  \lambda_C^a(t\mid\mathcal V_{ij}) =\lambda_{C0}^a(t)\exp\{\balpha_C^{a\mathsf T}\bh_C(\mathcal V_{ij})\},\qquad
  \lambda_D^a(t\mid\mathcal V_{ij})=\lambda_{D0}^a(t)\exp\{\bbeta_D^{a\mathsf T}\bh_D(\mathcal V_{ij})\},
\]
with fitted quantities \(\widehat\Lambda_{ij}^C(t;a)\), \(\widehat K_{ij}^a(t)=\exp\{-\widehat\Lambda_{ij}^C(t;a)\}\), \(\widehat\Lambda_{ij}^D(t;a)\), and \(\widehat H_{ij}^a(t)=\exp\{-\widehat\Lambda_{ij}^D(t;a)\}\). For recurrent event type \(k\), we use the arm-specific LWYY marginal proportional rates working model
\[
  \Ebb\{\dd N_{ijk}^{a}(t)\mid\mathcal V_{ij}\}=H_{ij}^a(t-)\exp\{\bgamma_{Rk}^{a\mathsf T}\bh_R(\mathcal V_{ij})\} \dd\Lambda_{Rk0}^a(t),
\]
with fitted type \(k\) increments \(\dd\widehat\Lambda_{ij}^{R_k}(t;a)\), fitted weighted increment \(\dd\widehat\Lambda_{ij}^{R,\bw}(t;a)=\sum_{k=1}^Kw_k\dd\widehat\Lambda_{ij}^{R_k}(t;a)\), and \(\dd\widehat q_{ij}^{R,\bw}(t;a)=\widehat H_{ij}^a(t-)\dd\widehat\Lambda_{ij}^{R,\bw}(t;a)\). Correct specification of a working model based on \(\mathcal V_{ij}\) means that the corresponding conditional nuisance function given the full baseline information \(\mathcal F_i\) is captured by \(\mathcal V_{ij}\). As in the IRT setting, outcome side correctness for the numerator requires \(\dd\widehat q_{ij}^{R,\bw}(t;a)\) to target \(\Ebb\{\dd N_{ij}^{a,\bw}(t)\mid\mathcal V_{ij}\}\). If the recurrent events and the terminal event are dependent through unmeasured participant-level or cluster-level factors, this requirement can be met either through a suitable joint model or through marginalization over the latent dependence structure \citep{huang2004joint,liu2004shared}. Further details on nuisance fitting and influence function construction are provided in Appendix \ref{supp:sec:crt_variance_detailed}.

The cluster-level first order influence functions describe the asymptotic distribution of the proposed working independent local estimator, but they cannot, in general, be interpreted as semiparametric efficient influence functions in the unrestricted observed data CRT model. With member specific censoring and within cluster dependence, the exact canonical gradient depends on the cluster level coarsening structure. In particular, observations from other members after a given member is censored may contain auxiliary information about that member's remaining recurrent event burden or remaining survival time, so the Hilbert space projection defining the exact canonical gradient may involve pattern specific or higher order coarsening projection terms. The proposed estimator does not attempt to estimate these terms, because doing so would require specifying and estimating the joint within cluster coarsening law and high dimensional history specific remaining outcome regressions. Instead, the estimator uses the feasible first order memberwise augmentation in \eqref{eq:crt_DR_eq_D} and \eqref{eq:crt_DR_eq_R}, which preserves componentwise double robustness and root \(M\) asymptotic normality under the stated conditions but does not generally attain the semiparametric efficiency bound. The first order memberwise influence function coincides with the efficient influence function only under additional conditions under which the higher order coarsening projection terms vanish, for example when later observations from other cluster members provide no additional information for predicting a censored member's remaining outcomes.

\subsection{Asymptotic properties}
\label{sec:crt_asymptotics}
 
In this section, we establish the asymptotic distribution of the proposed CRT estimators, with clusters as the independent units. Participant level dependence within a cluster is allowed and is incorporated through cluster-level influence functions. Throughout, the time horizon \(\tau\), the treatment arm \(a\), the event type weight vector \(\bw\), and the target population \(\ell\in\{\mathrm{clus},\mathrm{ind}\}\) are fixed, and all limits are taken as \(M\to\infty\). The regularity conditions for the asymptotic analysis, together with the probability limits of the two local estimators under possible working model misspecification, are stated in Appendix \ref{supp:sec:crt_variance_detailed}.
 
The asymptotic distribution is expressed in terms of cluster-level influence functions. Let \(\Phi_{\Lambda^D,a,\ell,i}(t)\) and \(\Phi_{\Lambda^{R,\bw},a,\ell,i}(t)\) denote the cluster-level influence processes of \(\widehat\Lambda_{a,\ell}^{D,\mathrm{DR}}(t)\) and \(\widehat\Lambda_{a,\ell}^{R,\mathrm{DR},\bw}(t)\), which include the empirical cluster contributions and the first order effects of estimating the censoring and outcome working models, and let \(\Phi_{\nu,a,\ell,i}\) and \(\Phi_{\mu,a,\ell,i}\) denote the induced cluster-level influence functions of \(\widehat\nu_{a,\ell}^{\mathrm{DR}}(\tau)\) and \(\widehat\mu_{a,\ell}^{\mathrm{DR}}(\tau;\bw)\), obtained from the product integral delta method for the target weighted survival function and the Stieltjes integral representation of the target weighted burden. Explicit expressions, together with product integral versions allowing jumps in the terminal event cumulative hazard, are given in Appendix \ref{supp:sec:crt_variance_detailed}.
 
\begin{theorem}[Asymptotic normality]
\label{thm:crt_asymptotic_normality}
Suppose Assumptions \ref{asm:crt_consistency}--\ref{asm:crt_censoring} and the regularity conditions in Appendix \ref{supp:sec:crt_variance_detailed} hold. Suppose also that, for each treatment arm, \(\mathcal{C}_a\) or \((\mathcal{O}_a^D\cap\mathcal{O}_a^R)\) holds, and that \(\nu_{a,\ell}(\tau)>0\). Then, for each fixed target population \(\ell\),
\begin{align*}
  M^{1/2}\left[
    \begin{pmatrix}
      \widehat\mu_{a,\ell}^{\mathrm{DR}}(\tau;\bw) \\
      \widehat\nu_{a,\ell}^{\mathrm{DR}}(\tau)
    \end{pmatrix}
    -
    \begin{pmatrix}
      \mu_{a,\ell}(\tau;\bw) \\
      \nu_{a,\ell}(\tau)
    \end{pmatrix}
  \right]
  &= M^{-1/2}\sum_{i=1}^M
  \begin{pmatrix}
    \Phi_{\mu,a,\ell,i} \\
    \Phi_{\nu,a,\ell,i}
  \end{pmatrix}
  +o_p(1).
\end{align*}
Thus, \(  M^{1/2}\{\widehat\psi_{a,\ell}^{\mathrm{DR}}(\tau;\bw)-\psi_{a,\ell}(\tau;\bw)\}=M^{-1/2}\sum_{i=1}^M\Phi_{\psi,a,\ell,i}+o_p(1)\),  where \(\Phi_{\psi,a,\ell,i}=\frac{\Phi_{\mu,a,\ell,i}}{\nu_{a,\ell}(\tau)}-\frac{\mu_{a,\ell}(\tau;\bw)}{\{\nu_{a,\ell}(\tau)\}^2}\Phi_{\nu,a,\ell,i}\). For the treatment contrast, \(\Phi_{\Delta,\ell,i}=\Phi_{\psi,1,\ell,i}-\Phi_{\psi,0,\ell,i}\), and hence
\(M^{1/2}\{\widehat\Delta_{\ell}^{\mathrm{DR}}(\tau;\bw)-\Delta_{\ell}(\tau;\bw)\}=M^{-1/2}\sum_{i=1}^M\Phi_{\Delta,\ell,i}+o_p(1)\). If \(\sigma_{\Delta,\ell}^2=\Ebb(\Phi_{\Delta,\ell,i}^2)\) is finite and nonzero, then \(M^{1/2}\{\widehat\Delta_{\ell}^{\mathrm{DR}}(\tau;\bw)-\Delta_{\ell}(\tau;\bw)\}\overset{d}{\longrightarrow}\mathrm{N}(0,\sigma_{\Delta,\ell}^2).\)
\end{theorem}
 
Theorem \ref{thm:crt_asymptotic_normality} follows from cluster-level empirical process arguments, the asymptotic linearity of the nuisance estimators, the first order expansion of the two local ratio estimators, the differentiability of the product integral survival map, and the delta method for the ratio \((\mu,\nu)\mapsto\mu/\nu\) \citep{andersen1993statistical,tsiatis2006semiparametric,vandervaart1998asymptotic}. The proof is given in Appendix \ref{supp:sec:crt_variance_detailed}.
 
\begin{corollary}[Variance estimation]
\label{cor:crt_variance}
Let \(\widehat\Phi_{\Delta,\ell,i}\) be the estimated cluster level influence function obtained by replacing unknown quantities in \(\Phi_{\Delta,\ell,i}\) with their sample analogues and fitted nuisance quantities. Define \(\overline{\widehat\Phi}_{\Delta,\ell}=M^{-1}\sum_{i=1}^M\widehat\Phi_{\Delta,\ell,i}\) and \(\widehat\sigma_{\Delta,\ell}^2=\frac{1}{M-1}\sum_{i=1}^M\left(\widehat\Phi_{\Delta,\ell,i}-\overline{\widehat\Phi}_{\Delta,\ell}\right)^2\). Under the conditions of Theorem \ref{thm:crt_asymptotic_normality}, \(\widehat{\operatorname{se}}\{\widehat\Delta_{\ell}^{\mathrm{DR}}(\tau;\bw)\}=\left(\widehat\sigma_{\Delta,\ell}^2/M\right)^{1/2}\) is a consistent standard error estimator, and an asymptotic \(95\%\) Wald confidence interval is \(\widehat\Delta_{\ell}^{\mathrm{DR}}(\tau;\bw)\pm1.96\,\widehat{\operatorname{se}}\{\widehat\Delta_{\ell}^{\mathrm{DR}}(\tau;\bw)\}\). The same construction provides standard errors for \(\widehat\psi_{a,\ell}^{\mathrm{DR}}(\tau;\bw)\), \(\widehat\mu_{a,\ell}^{\mathrm{DR}}(\tau;\bw)\), and \(\widehat\nu_{a,\ell}^{\mathrm{DR}}(\tau)\).
\end{corollary}

\section{Simulation studies}
\label{sec:simulation}

We conducted simulation studies to assess the finite sample performance of the proposed estimators in individually-randomized trials and cluster-randomized trials. The operating characteristics were relative bias (RBias), calculated as \(100(\overline{\widehat\Delta}-\Delta)/\Delta\), Monte Carlo standard deviation (MCSD), average asymptotic standard error (ASE), and empirical coverage of nominal \(95\%\) confidence intervals (COV). The proposed doubly robust estimator was compared with an inverse probability of censoring weighted estimator and an outcome regression estimator. The inverse probability weighted estimator retained the weighted observed data contributions but omitted outcome augmentation, whereas the outcome regression estimator retained the fitted full data contributions without the inverse probability correction. Four nuisance specifications were considered: \(O1C1\), \(O1C0\), \(O0C1\), and \(O0C0\), where \(O1\) and \(O0\) indicate correct and misspecified outcome working models and \(C1\) and \(C0\) indicate correct and misspecified censoring working models.

All potential event histories were generated directly in continuous time. Weibull baseline functions were used only to generate exact terminal-event, censoring, and recurrent-event histories. Estimation remained semiparametric: the terminal-event and censoring nuisance functions were fitted by arm-specific Cox regression, and each recurrent-event type was fitted by an arm-specific LWYY marginal proportional-rate model. Breslow-type estimators were used for the baseline cumulative hazard and mean-rate functions \citep{cox1972regression,lin2000semiparametric}. Thus, correct specification concerned the covariate component of each working model; no parametric form was imposed on a fitted baseline function. The numerical grid was used only to evaluate the local Nelson--Aalen equations, product integrals, and Stieltjes integrals after the exact continuous-time event histories had been generated. Each setting used \(1000\) Monte Carlo replications. The asymptotic standard errors were calculated from the influence functions in Theorems \ref{thm:irt_asymptotic_normality} and \ref{thm:crt_asymptotic_normality}. The nuisance influence contributions comprised the Cox martingale terms for terminal-event and censoring regression, the complete LWYY mean-zero rate-residual terms for recurrent-event regression, and the Breslow baseline contributions for all nuisance models. The independent influence contribution was defined at the participant level for IRTs and at the cluster level for CRTs. Further implementation details are provided in Appendix \ref{supp:sec:simulation_settings}.

\subsection{Individually-randomized trials}
\label{subsec:simulation_irt}

Participants were randomized in equal numbers to intervention and control. The baseline covariate vector was \(\bZ_i=(Z_{i1},Z_{i2})^{\mathsf T}\), where \(Z_{i1}\sim\operatorname{Bernoulli}(0.5)\) and \(Z_{i2}\sim\operatorname{Uniform}(-1,1)\) were independent, and we defined \(\bh_i=(Z_{i1},Z_{i2},Z_{i1}Z_{i2})^{\mathsf T}\). For \(h\in\{D,C\}\), the conditional cumulative hazard under arm \(a\) was \(\Lambda_{h,i}^a(t)=\lambda_{h,a}t^{\rho_h}\exp\{\bbeta_{h,a}^{\mathsf T}\bh_i\}\). If \(E_{h,i}\sim\operatorname{Exp}(1)\), the corresponding exact time was generated as \(T_{h,i}^a=[E_{h,i}/\{\lambda_{h,a}\exp(\bbeta_{h,a}^{\mathsf T}\bh_i)\}]^{1/\rho_h}\), with \(T_{D,i}^a=D_i^a\) and \(T_{C,i}^a=C_i^a\). Terminal event and censoring times were conditionally independent given treatment and baseline covariates.

A shared participant frailty induced dependence between the two recurrent event types and serial dependence within each type. Specifically, \(V_i\sim\operatorname{Gamma}(2,0.5)\), where \(\Ebb(V_i)=1\) and \(\operatorname{Var}(V_i)=0.5\). Conditional on \(A_i=a\), \(\bZ_i\), and \(V_i\), recurrent event type \(k\) followed a nonhomogeneous Poisson process with intensity \(V_i\lambda_{R_k,a}\rho_{R_k}t^{\rho_{R_k}-1}\exp\{\bbeta_{R_k,a}^{\mathsf T}\bh_i\}\), for \(k=1,2\). The process was generated up to \(D_i^a\wedge\tau\) and was then observed only up to censoring. Equivalently, conditional on \(D_i^a\), \(V_i\), and the covariates, the number of type \(k\) events by \(D_i^a\wedge\tau\) was Poisson with mean \(V_i\lambda_{R_k,a}(D_i^a\wedge\tau)^{\rho_{R_k}}\exp\{\bbeta_{R_k,a}^{\mathsf T}\bh_i\}\), and conditional on that count, the event times were generated as \((D_i^a\wedge\tau)U^{1/\rho_{R_k}}\), where \(U\sim\operatorname{Uniform}(0,1)\). Because the same mean one frailty was used for both recurrent event types, the marginal conditional mean rates retained the LWYY proportional-rate structure used by the correct working models. The terminal event and recurrent event processes also shared prognostic measured covariates, and recurrent events were stopped by death. No unmeasured frailty was shared between death and recurrence, so the outcome-side working models labeled \(O1\) genuinely identified the conditional full data recurrent event increments.

The parameters specific to each process are summarized in Appendix Table \ref{supp:tab:sim_irt_dgp}. The censoring baseline scale was \(\lambda_C=\exp(c_{\tau,\gamma})\), where \(\gamma\in\{0.40,0.60\}\) was the target observed censoring proportion. The calibrated log-scales were \(c_{3,0.40}=-1.8149\), \(c_{3,0.60}=-1.1819\), \(c_{5,0.40}=-2.2582\), and \(c_{5,0.60}=-1.6129\). The event-type weight vector was \(\bw=(1,1)^{\mathsf T}\). Under \(O1\), the arm-specific Cox terminal-event model and the arm-specific LWYY recurrent-event models included all components of \(\bh_i\). Under \(O0\), each outcome model retained only \(Z_{i1}\). Under \(C1\), the arm-specific Cox censoring model included all components of \(\bh_i\), whereas under \(C0\) it retained only \(Z_{i1}\). The fitted baseline hazard and rate functions remained unrestricted in every nuisance scenario. We considered \(n\in\{1600,3200\}\), \(\tau\in\{3,5\}\), and target censoring proportions of \(40\%\) and \(60\%\). Since \(V_i\) was independent of \(D_i^a\) conditional on \(\bZ_i\) and had mean one, the true component estimands were computed from \(\nu_a(\tau)=\int_0^\tau\Ebb\{S_D^a(t\mid\bZ_i)\}\dd t\) and \(\mu_a(\tau;\bw)=\sum_{k=1}^2w_k\int_0^\tau\Ebb[S_D^a(t\mid\bZ_i)\lambda_{R_k}^a(t\mid\bZ_i)]\dd t\). The expectation over \(Z_{i1}\) was evaluated exactly and the expectations over \(Z_{i2}\) and time were evaluated by Gauss--Legendre quadrature. The true treatment contrasts were \(0.5419\) at \(\tau=3\) and \(0.5814\) at \(\tau=5\). Table \ref{tab:sim_irt_semiparametric} reports the results for \(n=1600\), and results for \(n=3200\) are provided in Appendix Table \ref{supp:tab:irt_n3200}.

\begin{table}[tbp]
\caption{Monte Carlo performance for the IRT while-alive treatment contrast with \(n=1600\). Cens. is the mean observed censoring percentage. \(O1\) and \(C1\) denote correct outcome-side and censoring-side working models, respectively. RBias is relative bias in percent, MCSD is the Monte Carlo standard deviation, ASE is the average asymptotic standard error, and Cov is empirical coverage in percent of the nominal \(95\%\) Wald interval.}
\label{tab:sim_irt_semiparametric}
\centering
\setlength{\tabcolsep}{2.2pt}
\renewcommand{\arraystretch}{1.04}
\begin{adjustbox}{max width=\textwidth}
\begin{tabular}{@{}cccc*{12}{r}@{}}
\toprule
& & & & \multicolumn{4}{c}{\textsc{dr}} & \multicolumn{4}{c}{\textsc{ipcw}} & \multicolumn{4}{c}{\textsc{or}} \\
\cmidrule(lr){5-8}\cmidrule(lr){9-12}\cmidrule(lr){13-16}
\(\tau\) & Cens. & \(O,C\) & \(\Delta\) & RBias & MCSD & ASE & Cov & RBias & MCSD & ASE & Cov & RBias & MCSD & ASE & Cov \\
\midrule
3 & 40.0 & O1C1 & 0.542 & +0.1 & 0.061 & 0.060 & 95.1 & +0.1 & 0.062 & 0.061 & 94.5 & +0.3 & 0.060 & 0.059 & 95.6 \\
 &  & O1C0 & 0.542 & +0.1 & 0.060 & 0.059 & 95.7 & -5.9 & 0.056 & 0.055 & 89.7 & +0.3 & 0.060 & 0.059 & 95.6 \\
 &  & O0C1 & 0.542 & +0.1 & 0.062 & 0.061 & 94.4 & +0.1 & 0.062 & 0.061 & 94.5 & -5.8 & 0.056 & 0.055 & 90.2 \\
 &  & O0C0 & 0.542 & -6.0 & 0.056 & 0.055 & 90.0 & -5.9 & 0.056 & 0.055 & 89.7 & -5.8 & 0.056 & 0.055 & 90.2 \\
\addlinespace[2pt]
3 & 60.0 & O1C1 & 0.542 & -0.9 & 0.069 & 0.069 & 95.7 & -0.6 & 0.071 & 0.071 & 95.6 & -0.3 & 0.067 & 0.067 & 95.4 \\
 &  & O1C0 & 0.542 & -0.7 & 0.067 & 0.067 & 95.8 & -11.5 & 0.057 & 0.057 & 78.1 & -0.3 & 0.067 & 0.067 & 95.4 \\
 &  & O0C1 & 0.542 & -0.5 & 0.071 & 0.071 & 96.0 & -0.6 & 0.071 & 0.071 & 95.6 & -11.2 & 0.057 & 0.057 & 79.0 \\
 &  & O0C0 & 0.542 & -11.5 & 0.057 & 0.057 & 78.2 & -11.5 & 0.057 & 0.057 & 78.1 & -11.2 & 0.057 & 0.057 & 79.0 \\
\addlinespace[2pt]
5 & 40.1 & O1C1 & 0.581 & -0.4 & 0.061 & 0.059 & 94.2 & -0.3 & 0.062 & 0.060 & 93.7 & -0.1 & 0.060 & 0.058 & 94.0 \\
 &  & O1C0 & 0.581 & -0.3 & 0.060 & 0.058 & 94.1 & -7.2 & 0.054 & 0.054 & 85.7 & -0.1 & 0.060 & 0.058 & 94.0 \\
 &  & O0C1 & 0.581 & -0.3 & 0.062 & 0.060 & 93.9 & -0.3 & 0.062 & 0.060 & 93.7 & -7.0 & 0.054 & 0.054 & 86.0 \\
 &  & O0C0 & 0.581 & -7.2 & 0.054 & 0.054 & 85.3 & -7.2 & 0.054 & 0.054 & 85.7 & -7.0 & 0.054 & 0.054 & 86.0 \\
\addlinespace[2pt]
5 & 60.0 & O1C1 & 0.581 & -0.4 & 0.074 & 0.070 & 92.3 & -0.3 & 0.076 & 0.071 & 92.3 & +0.2 & 0.070 & 0.066 & 93.4 \\
 &  & O1C0 & 0.581 & -0.2 & 0.070 & 0.066 & 93.0 & -12.4 & 0.058 & 0.055 & 71.8 & +0.2 & 0.070 & 0.066 & 93.4 \\
 &  & O0C1 & 0.581 & -0.2 & 0.076 & 0.072 & 92.4 & -0.3 & 0.076 & 0.071 & 92.3 & -12.0 & 0.059 & 0.056 & 72.6 \\
 &  & O0C0 & 0.581 & -12.4 & 0.058 & 0.055 & 71.4 & -12.4 & 0.058 & 0.055 & 71.8 & -12.0 & 0.059 & 0.056 & 72.6 \\
\bottomrule
\end{tabular}
\end{adjustbox}
\end{table}

Table \ref{tab:sim_irt_semiparametric} supports the componentwise robustness result. When at least one nuisance side was correct, the absolute relative bias of the doubly robust estimator did not exceed \(0.9\%\). The asymptotic standard error closely tracked the Monte Carlo standard deviation in most settings. For example, under \(O0C1\), \(\tau=3\), and approximately \(60\%\) censoring, the MCSD and ASE were \(0.071\) and \(0.071\), respectively, with \(96.0\%\) coverage. The largest discrepancy occurred under \(O1C1\), \(\tau=5\), and approximately \(60\%\) censoring, where the MCSD and ASE were \(0.074\) and \(0.070\), respectively, with \(92.3\%\) coverage. Across settings in which at least one nuisance side was correct, coverage ranged from \(92.3\%\) to \(96.0\%\). The singly robust estimators failed along their unprotected nuisance pathway. Under censoring-model misspecification, the absolute relative bias of IPCW ranged from \(5.9\%\) to \(12.4\%\), whereas under outcome-model misspecification, the absolute relative bias of outcome regression ranged from \(5.8\%\) to \(12.0\%\). When both nuisance sides were misspecified, the absolute relative bias of the doubly robust estimator increased from \(6.0\%\) under \(40\%\) censoring at \(\tau=3\) to \(12.4\%\) under \(60\%\) censoring at \(\tau=5\), with coverage decreasing from \(90.0\%\) to \(71.4\%\). Results for \(n=3200\) in Appendix Table \ref{supp:tab:irt_n3200} showed the same qualitative pattern, smaller Monte Carlo variability, and absolute doubly robust bias below \(0.4\%\) whenever at least one nuisance side was correct.

\subsection{Cluster-randomized trials}
\label{subsec:simulation_crt}

We next evaluated the proposed estimators in cluster-randomized trials with informative cluster size, continuous terminal event and censoring times, two recurrent event types, and dependence both within participants and within clusters. Clusters were assigned in equal numbers to intervention and control. Cluster size was generated from the discrete uniform distribution on \(\{20,\ldots,80\}\), and we defined \(n_i^\star=(n_i-50)/30\). A binary cluster-level covariate \(L_i\), a binary individual-level covariate \(Z_{ij1}\), and a continuous individual-level covariate \(Z_{ij2}\) were generated according to \(L_i\sim\operatorname{Bernoulli}[\operatorname{expit}\{-0.20+0.70n_i^\star\}]\), \(Z_{ij1}\mid(n_i,L_i)\sim\operatorname{Bernoulli}[\operatorname{expit}\{-0.15+0.40L_i+0.55n_i^\star\}]\), and \(Z_{ij2}=\max[-2,\min\{2,\,0.25n_i^\star+0.20L_i+0.25Z_{ij1}+0.75\varepsilon_{ij}\}]\), where \(\varepsilon_{ij}\sim\mathrm{N}(0,1)\). The covariate vector entering the data-generating models was \(\bh_{ij}=(L_i,Z_{ij1},Z_{ij2},n_i^\star,n_i^\star Z_{ij2},Z_{ij1}Z_{ij2})^\mathsf{T}\). Thus, cluster size affected both baseline covariate distributions and the event processes directly, and it modified the effect of the continuous individual-level covariate. Conditional on \(\bh_{ij}\), the marginal terminal event distribution under arm \(a\) was Weibull proportional hazards, with \(\Pbb(D_{ij}^a>t\mid\bh_{ij})=\exp[-\lambda_{D,a}t^{\rho_D}\exp\{\bbeta_{D,a}^\mathsf{T}\bh_{ij}\}]\). Censoring was generated independently of the outcome processes conditional on \(\bh_{ij}\), with the analogous survival function determined by \(\lambda_C\), \(\rho_C\), and \(\bbeta_{C,a}\). For recurrent-event type \(k\), exact event times were generated from a nonhomogeneous Poisson process with conditional intensity \(\lambda_{R_k,a}\rho_{R_k}t^{\rho_{R_k}-1}\exp\{\bbeta_{R_k,a}^\mathsf{T}\bh_{ij}\}V_iV_{ij}\), stopped at \(D_{ij}^a\wedge\tau\). Here \(V_i\sim\operatorname{Gamma}(5,0.2)\) was shared by all participants in cluster \(i\), and \(V_{ij}\sim\operatorname{Gamma}(1/0.35,0.35)\) was shared by the two recurrent-event types for participant \(j\). Both frailties had mean one. They induced within-cluster, within-participant, and cross-type recurrent-event dependence while preserving the marginal LWYY proportional-rate structure after integration over the frailties.

Within-cluster dependence of terminal-event times was introduced through a Gaussian copula. Specifically, \(U_{ij}^D=\Phi\{\sqrt{0.10}B_i+\sqrt{0.90}\epsilon_{ij}^D\}\), where \(B_i\) and \(\epsilon_{ij}^D\) were independent standard normal variables, and \(D_{ij}^a\) was obtained by applying the inverse Weibull proportional-hazards survival function to \(U_{ij}^D\). This construction preserved each participant's marginal Cox model while inducing latent correlation \(0.10\) between terminal-event times from the same cluster. The terminal-event copula, recurrent-event frailties, and censoring variables were mutually independent conditional on the observed baseline covariates. Exact continuous event times were generated before introducing a grid of \(200\) points used only to evaluate the local Nelson--Aalen equations, product integrals, and Stieltjes integrals. The baseline scales, shape parameters, and regression coefficients are summarized in Appendix Table \ref{supp:tab:crt_sim_dgp}.

The censoring baseline scale was calibrated separately for each horizon and target censoring proportion. For \((\tau,\gamma)=(3,0.40),(3,0.60),(5,0.40),(5,0.60)\), the corresponding values of \(\log(\lambda_C)\) were \(-1.8372\), \(-1.2210\), \(-2.2867\), and \(-1.6565\). The resulting mean observed censoring proportions were approximately \(40\%\) and \(60\%\). The nuisance models were fitted semiparametrically. Terminal-event and censoring nuisance functions were estimated by arm-specific Cox regression with Breslow baseline hazards, and the two recurrent-event nuisance functions were estimated by arm-specific LWYY marginal proportional-rate regression with Breslow baseline mean-rate functions. Under \(O1\), all six components of \(\bh_{ij}\) were included. Under \(O0\), only the binary cluster-level covariate \(L_i\) was retained, thereby omitting both individual-level covariates, cluster size, and the two interaction terms. The corresponding rule was used for the censoring model under \(C1\) and \(C0\). 

We considered \(M\in\{50,80\}\), \(\tau\in\{3,5\}\), censoring targets of \(40\%\) and \(60\%\), and \(1000\) Monte Carlo replications per setting. Confidence intervals used a \(t_{M-2}\) critical value. Deterministic true values were obtained by summing over the discrete cluster size and binary covariate distributions and applying Gaussian quadrature to the continuous covariate. The cluster-average treatment effects were \(-0.309\) at \(\tau=3\) and \(-0.289\) at \(\tau=5\); the corresponding individual-average effects were \(-0.428\) and \(-0.403\). Their separation confirms that cluster size was informative. Table \ref{tab:sim_crt_semiparametric} reports the cluster-average treatment effect for \(M=50\). Results for \(M=80\) and all individual-average effects are given in Appendix Tables \ref{supp:tab:crt_individual_M50}--\ref{supp:tab:crt_cluster_M80}.

\begin{table}[tbp]
\caption{Monte Carlo performance for the cluster-average while-alive treatment contrast with \(M=50\). Cens. is the mean observed censoring percentage. \(O1\) and \(C1\) denote correct outcome-side and censoring-side working models, respectively. RBias is relative bias in percent, MCSD is the Monte Carlo standard deviation, ASE is the average cluster influence-function standard error, and Cov is empirical coverage in percent of the nominal \(95\%\) interval based on a \(t_{M-2}\) critical value.}
\label{tab:sim_crt_semiparametric}
\centering
\setlength{\tabcolsep}{2.2pt}
\renewcommand{\arraystretch}{1.04}
\begin{adjustbox}{max width=\textwidth}
\begin{tabular}{@{}cccc*{12}{r}@{}}
\toprule
& & & & \multicolumn{4}{c}{\textsc{dr}}
& \multicolumn{4}{c}{\textsc{ipcw}}
& \multicolumn{4}{c}{\textsc{or}} \\
\cmidrule(lr){5-8}\cmidrule(lr){9-12}\cmidrule(lr){13-16}
\(\tau\) & Cens. & \(O,C\) & \(\Delta_{\mathrm{clus}}\)
& RBias & MCSD & ASE & Cov
& RBias & MCSD & ASE & Cov
& RBias & MCSD & ASE & Cov \\
\midrule
3 & 40.0 & O1C1 & -0.309 & +0.9 & 0.089 & 0.082 & 93.9 & +1.0 & 0.093 & 0.088 & 93.8 & +1.5 & 0.089 & 0.083 & 93.6 \\
 &  & O1C0 & -0.309 & +0.9 & 0.089 & 0.082 & 93.8 & -6.3 & 0.087 & 0.083 & 91.9 & +1.5 & 0.089 & 0.083 & 93.6 \\
 &  & O0C1 & -0.309 & +1.1 & 0.093 & 0.086 & 94.1 & +1.0 & 0.093 & 0.088 & 93.8 & +24.8 & 0.099 & 0.090 & 89.1 \\
 &  & O0C0 & -0.309 & -6.4 & 0.087 & 0.082 & 91.8 & -6.3 & 0.087 & 0.083 & 91.9 & +24.8 & 0.099 & 0.090 & 89.1 \\
\addlinespace[2pt]
3 & 60.0 & O1C1 & -0.309 & +1.3 & 0.089 & 0.084 & 94.2 & +1.7 & 0.094 & 0.090 & 93.7 & +2.1 & 0.089 & 0.085 & 94.3 \\
 &  & O1C0 & -0.309 & +1.5 & 0.089 & 0.084 & 94.2 & -11.0 & 0.084 & 0.081 & 88.9 & +2.1 & 0.089 & 0.085 & 94.3 \\
 &  & O0C1 & -0.309 & +1.6 & 0.093 & 0.089 & 93.6 & +1.7 & 0.094 & 0.090 & 93.7 & +19.0 & 0.094 & 0.088 & 92.1 \\
 &  & O0C0 & -0.309 & -11.5 & 0.083 & 0.080 & 90.1 & -11.0 & 0.084 & 0.081 & 88.9 & +19.0 & 0.094 & 0.088 & 92.1 \\
\addlinespace[2pt]
5 & 40.0 & O1C1 & -0.289 & -1.7 & 0.085 & 0.079 & 91.2 & -1.4 & 0.088 & 0.084 & 92.7 & -0.8 & 0.085 & 0.080 & 92.1 \\
 &  & O1C0 & -0.289 & -1.7 & 0.085 & 0.079 & 91.4 & -7.9 & 0.084 & 0.079 & 90.3 & -0.8 & 0.085 & 0.080 & 92.1 \\
 &  & O0C1 & -0.289 & -1.3 & 0.088 & 0.083 & 93.4 & -1.4 & 0.088 & 0.084 & 92.7 & +24.3 & 0.095 & 0.086 & 89.7 \\
 &  & O0C0 & -0.289 & -8.0 & 0.083 & 0.079 & 90.4 & -7.9 & 0.084 & 0.079 & 90.3 & +24.3 & 0.095 & 0.086 & 89.7 \\
\addlinespace[2pt]
5 & 60.0 & O1C1 & -0.289 & -0.9 & 0.089 & 0.081 & 93.5 & -0.2 & 0.091 & 0.086 & 93.5 & +0.4 & 0.089 & 0.082 & 93.0 \\
 &  & O1C0 & -0.289 & -0.6 & 0.088 & 0.081 & 93.5 & -11.8 & 0.081 & 0.078 & 88.9 & +0.4 & 0.089 & 0.082 & 93.0 \\
 &  & O0C1 & -0.289 & -0.8 & 0.091 & 0.085 & 93.3 & -0.2 & 0.091 & 0.086 & 93.5 & +19.2 & 0.092 & 0.085 & 93.2 \\
 &  & O0C0 & -0.289 & -12.8 & 0.081 & 0.077 & 87.8 & -11.8 & 0.081 & 0.078 & 88.9 & +19.2 & 0.092 & 0.085 & 93.2 \\
\bottomrule
\end{tabular}
\end{adjustbox}
\end{table}

Table \ref{tab:sim_crt_semiparametric} shows the expected componentwise robustness pattern. Whenever either nuisance side was correctly specified, the absolute relative bias of the doubly robust estimator was at most \(1.7\%\). The analytic variance formula reproduced the Monte Carlo variability well. Across these settings, the ratio of MCSD to ASE ranged from \(1.04\) to \(1.10\), the largest absolute difference between MCSD and ASE was \(0.008\), and empirical coverage ranged from \(91.2\%\) to \(94.2\%\). For example, under \(O1C1\), \(\tau=3\), and approximately \(40\%\) censoring, the MCSD and ASE were \(0.089\) and \(0.082\), respectively, with \(93.9\%\) coverage. Under \(O0C1\), \(\tau=5\), and approximately \(60\%\) censoring, the MCSD and ASE were \(0.091\) and \(0.085\), with \(93.3\%\) coverage. The singly robust estimators failed along their unprotected nuisance pathway. Under \(O1C0\), the absolute relative bias of IPCW ranged from \(6.3\%\) to \(11.8\%\), with coverage between \(88.9\%\) and \(91.9\%\). Under \(O0C1\), the absolute relative bias of outcome regression ranged from \(19.0\%\) to \(24.8\%\). Although some outcome-regression intervals retained moderate coverage because sampling variability remained large relative to the bias, the point estimates were clearly displaced from the true contrast. When both nuisance sides were misspecified, the absolute relative bias of the doubly robust estimator ranged from \(6.4\%\) to \(12.8\%\), and coverage decreased to between \(87.8\%\) and \(91.8\%\). The larger distortions under approximately \(60\%\) censoring illustrate how heavier censoring amplifies sensitivity to simultaneous nuisance misspecification. The additional results in Appendix Tables \ref{supp:tab:crt_cluster_M80}, \ref{supp:tab:crt_individual_M50}, and \ref{supp:tab:crt_individual_M80} showed the same qualitative behavior with \(M=80\) and for the individual-average estimand. The individual-average while-alive estimators exhibited the same componentwise robustness pattern while targeting the more negative individual-level contrasts induced by informative cluster size.

\section{Applications to data examples}
\label{sec:data_illustration}

We applied the proposed methods to one individually-randomized trial and one cluster-randomized trial to illustrate their usage under the two randomization schemes considered in this paper.

\subsection{HF-ACTION individually-randomized trial}
\label{subsec:hfaction}

The Heart Failure: A Controlled Trial Investigating Outcomes of Exercise Training (HF-ACTION) study was a multicenter, parallel-arm randomized trial designed to evaluate the efficacy and safety of aerobic exercise training in medically stable outpatients with chronic heart failure and reduced left ventricular ejection fraction. The original trial randomized 2331 participants at 82 centers in the United States, Canada, and France to usual care alone or usual care plus an exercise program consisting of 36 supervised sessions followed by home-based training. Its primary clinical outcome was the time to the first occurrence of all-cause death or all-cause hospitalization \citep{o2009efficacy}. In the protocol specified analysis, exercise training was associated with a nonsignificant reduction in this first-event composite, with a hazard ratio of 0.93 (95\% confidence interval, 0.84 to 1.02; \(p=0.13\)). A prespecified covariate-adjusted analysis yielded a hazard ratio of 0.89 (95\% confidence interval, 0.81 to 0.99; \(p=0.03\)) \citep{o2009efficacy}. A first event composite is clinically useful but does not use hospitalizations occurring after the first admission. This omission is consequential in chronic heart failure, where repeated admissions account for a substantial part of disease burden. Moreover, death permanently terminates the recurrent hospitalization process, whereas longer survival creates additional opportunity to experience hospitalization. We therefore reanalyzed HF-ACTION using the while-alive estimand, taking all-cause hospitalization as the recurrent loss event and all-cause death as the terminal event. The resulting estimand measures the expected number of all-cause hospitalizations per year alive and complements the original first-event hazard-ratio analysis rather than targeting the same treatment effect parameter.

The available event-history data contained 2130 randomized participants, comprising 1070 assigned to usual care and 1060 assigned to exercise training. The treatment groups were similar with respect to the available baseline characteristics. The mean cardiopulmonary exercise-test duration was 9.9 minutes (standard deviation 4.0) under usual care and 9.7 minutes (standard deviation 3.8) under exercise training, and the mean left ventricular ejection fraction was 25.2\% in both groups. Ischemic etiology was present in 51.0\% and 51.9\%, atrial fibrillation in 20.5\% and 21.1\%, and depression in 22.1\% and 20.3\% of the usual-care and exercise-training groups, respectively. By four years, the data contained 1659 all-cause hospitalizations over 2674.5 observed person-years in the usual-care group and 1610 hospitalizations over 2659.7 observed person-years in the exercise-training group, corresponding to crude hospitalization rates of 62.0 and 60.5 per 100 person-years. At least one hospitalization occurred in 668 usual-care participants (62.4\%) and 629 exercise-training participants (59.3\%); 514 (48.0\%) and 496 (46.8\%), respectively, experienced at least two hospitalizations. There were 183 deaths in the usual-care group and 167 in the exercise-training group, corresponding to crude mortality rates of 6.84 and 6.28 per 100 person-years. 

We set \(K=1\) and \(\bw=(1)\), so that \(N_i^{a,\bw}(t)\) was the cumulative number of all-cause hospitalizations by time \(t\). The primary horizon was four years, and secondary estimates were reported at 0.5, 1, 2, and 3 years. To show the evolution of the estimand, we additionally evaluated the treatment-specific rates and treatment contrast on a dense grid from 0.5 to 4 years. We used the doubly robust estimator in \eqref{eq:irt_psi_delta_DR}. The arm-specific terminal-event and censoring nuisance functions were fitted using Cox working models, and the recurrent hospitalization nuisance function was fitted using an LWYY marginal proportional rate working model. A common, prespecified adjustment set was used in all nuisance models: standardized cardiopulmonary exercise-test duration, standardized left ventricular ejection fraction, ischemic etiology, atrial fibrillation, depression, and indicators of missing cardiopulmonary exercise-test duration and left ventricular ejection fraction. The two continuous variables were median imputed before standardization. This adjustment set closely parallels the prognostic factors emphasized in the original HF-ACTION analysis, which included exercise-test duration, left ventricular ejection fraction, depression score, atrial fibrillation or flutter, and heart-failure etiology \citep{o2009efficacy}. The treatment probability was fixed at \(1/2\) according to the randomized design and was not estimated. Standard errors were calculated from the participant-level influence function in Theorem \ref{thm:irt_asymptotic_normality}, with the nuisance estimation contributions and transformations described in Section \ref{sec:irt_asymptotics}. Wald confidence intervals were formed according to Corollary \ref{cor:irt_variance}. Intervals along the dense curves are pointwise 95\% confidence intervals.

Figure \ref{fig:hfaction_while_alive} and Table \ref{supp:tab:hfaction_while_alive} in Appendix \ref{supp:sec:data_analysis} show that the estimated treatment contrast remained close to zero throughout follow-up. None of the five reported confidence intervals excluded zero. At one year, the estimated rate was slightly higher under exercise training, with a difference of 0.030 hospitalizations per year alive (95\% confidence interval, -0.060 to 0.120). At two, three, and four years, the point estimates favored exercise training, but the estimated differences remained modest and imprecise. At four years, the estimated while-alive hospitalization rates were 0.551 under usual care and 0.526 under exercise training, giving \(\widehat\Delta(4;\bw)=-0.025\) hospitalizations per year alive (95\% confidence interval, -0.070 to 0.020; \(p=0.27\)). Equivalently, the point estimate corresponds to approximately 2.5 fewer all-cause hospitalizations per 100 years alive under exercise training. The component estimates at four years aid interpretation of the ratio. The estimated mean cumulative hospitalization burden was 1.930 under usual care and 1.865 under exercise training, whereas the estimated restricted mean survival time was 3.504 and 3.549 years, respectively. Thus, the lower exercise-training rate reflected both a modestly smaller hospitalization numerator and a modestly larger time-alive denominator. The confidence interval nevertheless included no effect and clinically plausible effects in either direction.

\begin{figure}[tbp]
\centering
\textbf{(a)}\\[-2pt]
\includegraphics[width=0.5\textwidth]{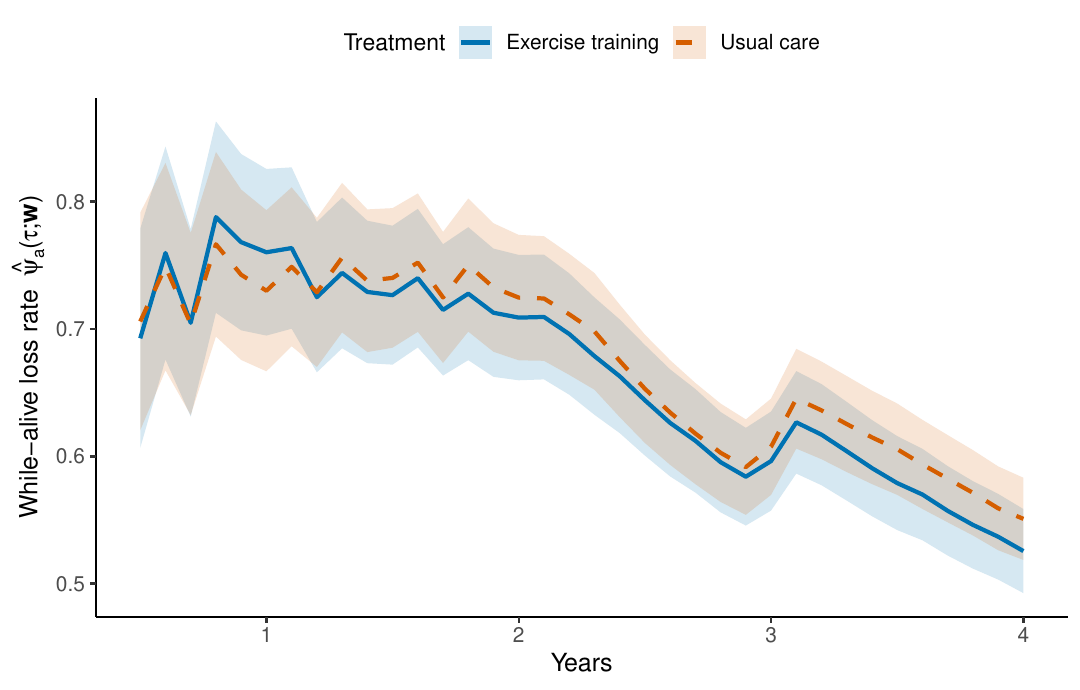}\\[4pt]
\textbf{(b)}\\[-2pt]
\includegraphics[width=0.5\textwidth]{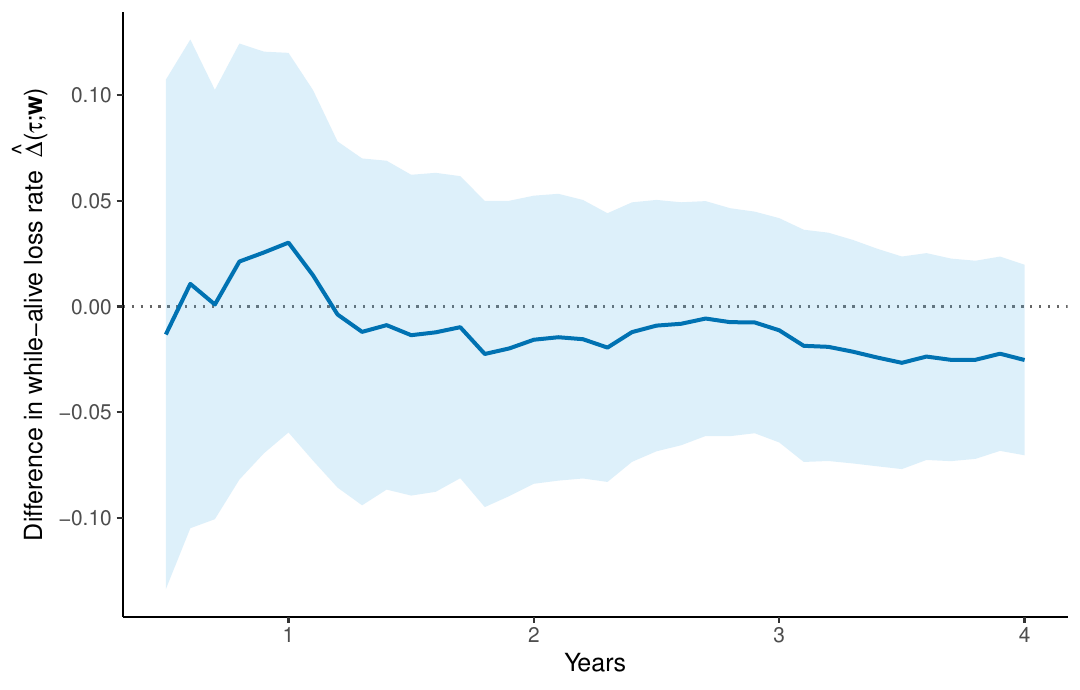}
\caption{Horizon-specific doubly robust estimates in the HF-ACTION analysis. Panel (a) shows the treatment-specific while-alive all-cause hospitalization rates, \(\widehat\psi_a(\tau;\bw)\). Panel (b) shows the exercise-training-minus-usual-care contrast, \(\widehat\Delta(\tau;\bw)\); negative values indicate a lower hospitalization burden per year alive under exercise training. Shaded regions are pointwise 95\% confidence intervals.}
\label{fig:hfaction_while_alive}
\end{figure}

These findings illustrate the distinction between the while-alive estimand and the original time-to-first-event estimand. The original covariate-adjusted analysis suggested a modest reduction in the hazard of first all-cause death or hospitalization, whereas the present analysis found no statistically clear reduction in the total all-cause hospitalization burden per year alive. These findings are not contradictory because the estimands differ in outcome aggregation, effect scale, and treatment of repeated hospitalizations. The original hazard ratio describes the instantaneous relative rate of a first composite event among participants still event free, whereas the while-alive contrast incorporates every observed hospitalization, treats death as a terminal event, and reports an absolute difference in hospitalization burden per unit time alive. These analyses suggest that any clinical benefit of exercise training was modest, with no strong evidence in the available study of a sustained reduction in recurrent all-cause hospitalization burden relative to survival time.

\subsection{STRIDE cluster-randomized trial}
\label{subsec:stride}

The Strategies to Reduce Injuries and Develop Confidence in Elders (STRIDE) study was a pragmatic, parallel-arm cluster-randomized trial designed to evaluate a patient-centered multifactorial strategy for preventing fall injuries in community-dwelling older adults at increased risk of falling. A total of 86 primary care practices within 10 health care systems were randomized, with 43 practices assigned to a nurse-delivered intervention and 43 to enhanced usual care. The intervention combined structured assessment of modifiable fall-risk factors with an individualized prevention plan developed by a trained falls care manager in collaboration with the participant and the participant's clinicians. The original trial enrolled 5451 participants who were at least 70 years of age. For the secondary outcome of first participant-reported fall injury, the corresponding rates were 25.6 and 28.6 per 100 person-years, with a hazard ratio of 0.90 (95\% confidence interval, 0.83 to 0.99; \(p=0.004\)) \citep{bhasin2020randomized}. The original time-to-first-event analyses answer whether the intervention delayed a participant's first fall injury, but they do not use injuries occurring after that first event. Recurrent injuries are clinically relevant in this older population, and death permanently terminates the period during which further injuries can occur. A while-alive analysis therefore provides a complementary summary by comparing the total recurrent fall-injury burden per unit time alive. The cluster-randomized design introduces an additional estimand distinction. The while-alive individual-average estimand gives equal weight to each participant and describes the effect for the average older adult, whereas a cluster-average while-alive first averages within each practice and then gives equal weight to practices, describing the effect for the average implementation unit. These targets can differ when practice size is associated with prognosis, recurrent injury burden, or heterogeneity of the treatment effect. This distinction is particularly relevant in STRIDE because the intervention was delivered through primary care practices and practice sizes varied substantially. The supplied event-history data contained all 5451 randomized participants, with 2649 participants in enhanced-usual-care practices and 2802 in intervention practices. Baseline characteristics were similar between treatment groups. The mean practice size was 61.6 participants (standard deviation 27.4; range 10 to 129) under enhanced usual care and 65.2 participants (standard deviation 37.6; range 10 to 199) under the intervention. The wide practice-size distribution reinforces the need to distinguish participant-weighted and practice-weighted targets. By three years, 2266 recurrent fall injuries were recorded over 5972.1 person-years in the enhanced-usual-care group and 2235 over 6262.2 person-years in the intervention group, giving crude recurrent-event rates of 37.9 and 35.7 per 100 person-years. At least one recorded fall injury occurred in 1224 enhanced-usual-care participants (46.2\%) and 1202 intervention participants (42.9\%); 537 (20.3\%) and 521 (18.6\%), respectively, experienced at least two injuries. There were 106 deaths under enhanced usual care and 117 under the intervention, corresponding to crude mortality rates of 1.77 and 1.87 per 100 person-years. Observed follow-up ended before three years for 78.9\% and 78.4\% of participants, respectively, making adjustment for right censoring important at the later horizons. 

We set \(K=1\) and \(\bw=(1)\), so that \(N_{ij}^{a,\bw}(t)\) was the cumulative number of recorded fall injuries by time \(t\). The primary horizon was three years, and the estimates were reported at 0.5, 1, 1.5, and 2 years. We also evaluated the arm-specific rates and treatment contrasts on a dense grid from 0.5 to 3 years. For each horizon and \(\ell\in\{\mathrm{ind},\mathrm{clus}\}\), we estimated \(\psi_{a,\ell}(\tau;\bw)\) and the intervention minus enhanced-usual-care contrast \(\Delta_{\ell}(\tau;\bw)\) using the doubly robust procedure in Section \ref{sec:crt_estimand}. The individual-level analysis used \(\omega_i^{\mathrm{ind}}=n_i\), so larger practices contributed in proportion to their number of participants. The cluster-level analysis used \(\omega_i^{\mathrm{clus}}=1\), so every practice contributed equally regardless of size. The same fitted nuisance functions were used for both targets, and only the target-defining cluster weights differed. The terminal-event and censoring nuisance functions were fitted using arm-specific marginal Cox working models. The recurrent fall-injury nuisance function was fitted using an arm-specific marginal LWYY proportional rate working model. A common prespecified adjustment set was used throughout: health-system site indicators, standardized age, female sex, White race, Hispanic ethnicity, standardized number of chronic conditions, and standardized practice size. Health-system indicators reflected the stratified trial structure, while practice size was included because it may be prognostic and is directly relevant to informative cluster size. The cluster treatment probability was fixed at \(1/2\) by design. Standard errors were calculated from the cluster-level influence function in Theorem \ref{thm:crt_asymptotic_normality}, including nuisance estimation contributions and the transformations described in Section \ref{sec:crt_asymptotics}. Confidence intervals used a \(t_{84}\) critical value, corresponding to \(M-2\) degrees of freedom. Intervals along the dense curves are pointwise 95\% confidence intervals.

\begin{table}[tbp]
\caption{Doubly robust estimates of the while-alive recurrent fall-injury rate in the STRIDE analysis. Rates are expected recorded fall injuries per year alive. The contrast is intervention minus enhanced usual care, so a negative value indicates a lower recurrent fall-injury burden per year alive under the intervention.}
\label{tab:stride_while_alive}
\centering
\setlength{\tabcolsep}{4.5pt}
\renewcommand{\arraystretch}{1.12}
\begin{adjustbox}{max width=\textwidth}
\begin{tabular}{ccccc}
\toprule
Horizon, years
& Enhanced usual care, \(\widehat\psi_0\) (95\% CI)
& Intervention, \(\widehat\psi_1\) (95\% CI)
& Difference, \(\widehat\Delta\) (95\% CI)
& \(p\)-value \\
\midrule
\multicolumn{5}{l}{\textit{While-alive individual-average estimand, \(\ell=\mathrm{ind}\)}}\\
0.5 & 0.327 (0.300, 0.353) & 0.285 (0.248, 0.322) & -0.042 (-0.089, 0.005) & 0.079 \\
1.0 & 0.355 (0.333, 0.377) & 0.311 (0.283, 0.339) & -0.044 (-0.078, -0.009) & 0.014 \\
1.5 & 0.372 (0.354, 0.391) & 0.331 (0.300, 0.361) & -0.041 (-0.077, -0.006) & 0.024 \\
2.0 & 0.385 (0.368, 0.401) & 0.349 (0.321, 0.376) & -0.036 (-0.067, -0.005) & 0.022 \\
3.0 & 0.379 (0.356, 0.402) & 0.371 (0.344, 0.399) & -0.007 (-0.043, 0.028) & 0.674 \\
\addlinespace[4pt]
\multicolumn{5}{l}{\textit{While-alive cluster-average estimand, \(\ell=\mathrm{clus}\)}}\\
0.5 & 0.312 (0.281, 0.344) & 0.293 (0.250, 0.337) & -0.019 (-0.072, 0.034) & 0.483 \\
1.0 & 0.339 (0.313, 0.365) & 0.317 (0.285, 0.348) & -0.022 (-0.061, 0.016) & 0.252 \\
1.5 & 0.357 (0.334, 0.380) & 0.330 (0.297, 0.364) & -0.027 (-0.065, 0.011) & 0.167 \\
2.0 & 0.372 (0.349, 0.396) & 0.348 (0.318, 0.377) & -0.024 (-0.059, 0.010) & 0.160 \\
3.0 & 0.362 (0.337, 0.387) & 0.371 (0.342, 0.400) &  0.009 (-0.027, 0.045) & 0.625 \\
\bottomrule
\end{tabular}
\end{adjustbox}
\par\vspace{2pt}
{\footnotesize CI, confidence interval. All intervals are pointwise intervals based on the cluster-level influence-function variance estimator and a \(t_{84}\) critical value.}
\end{table}

Table \ref{tab:stride_while_alive} and Figure \ref{fig:stride_while_alive} show a time-varying pattern. For the individual-average while-alive estimand, the intervention was associated with a lower recurrent fall-injury rate through approximately two years. At one year, the estimated rates were 0.355 under enhanced usual care and 0.311 under the intervention, giving \(\widehat\Delta_{\mathrm{ind}}(1;\bw)=-0.044\) injuries per year alive (95\% confidence interval, -0.078 to -0.009; \(p=0.014\)). Similar differences were observed at 1.5 and 2 years. Expressed on a more familiar scale, the two-year estimate corresponds to approximately 3.6 fewer recorded fall injuries per 100 years alive under the intervention. By three years, however, the individual-average while-alive difference had attenuated to -0.007 injuries per year alive (95\% confidence interval, -0.043 to 0.028; \(p=0.67\)). The pointwise significance at several intermediate horizons should therefore be interpreted as evidence about the shape of the horizon-specific effect curve rather than as a multiplicity-adjusted confirmatory finding.

The cluster-level estimates were closer to zero throughout follow-up, and none of their confidence intervals excluded zero. At 1.5 years, the estimated average-practice contrast was -0.027 injuries per year alive (95\% confidence interval, -0.065 to 0.011), compared with -0.041 for the average-participant contrast. At three years, the cluster-level estimate was slightly positive, \(\widehat\Delta_{\mathrm{clus}}(3;\bw)=0.009\) (95\% confidence interval, -0.027 to 0.045), whereas the individual-average treatment effect remained slightly negative. This divergence is substantively important. The participant-weighted analysis assigns more influence to larger practices, while the cluster-weighted analysis gives the same influence to every practice. The more favorable individual-average treatment effects are therefore consistent with a stronger intervention effect, a different baseline burden, or a different case mix in larger practices. The component estimates at three years clarify the source of the two contrasts. For the individual-level target, the estimated mean cumulative fall-injury burden was 1.110 under enhanced usual care and 1.088 under the intervention, while the estimated restricted mean survival time was 2.929 years in both groups. Thus, the small negative individual-level while-alive estimates arose almost entirely from a modestly lower recurrent-event numerator. For the cluster-level target, the estimated mean burden was 1.061 under enhanced usual care and 1.084 under the intervention, and the corresponding restricted mean survival times were 2.930 and 2.924 years. The slightly positive cluster-level rate contrast therefore reflected both a somewhat larger practice-average event burden and a slightly shorter practice-average time alive under the intervention. None of these three-year differences was statistically precise.

\begin{figure}[tbp]
\centering
\textbf{(a)}\\[-2pt]
\includegraphics[width=0.6\textwidth]{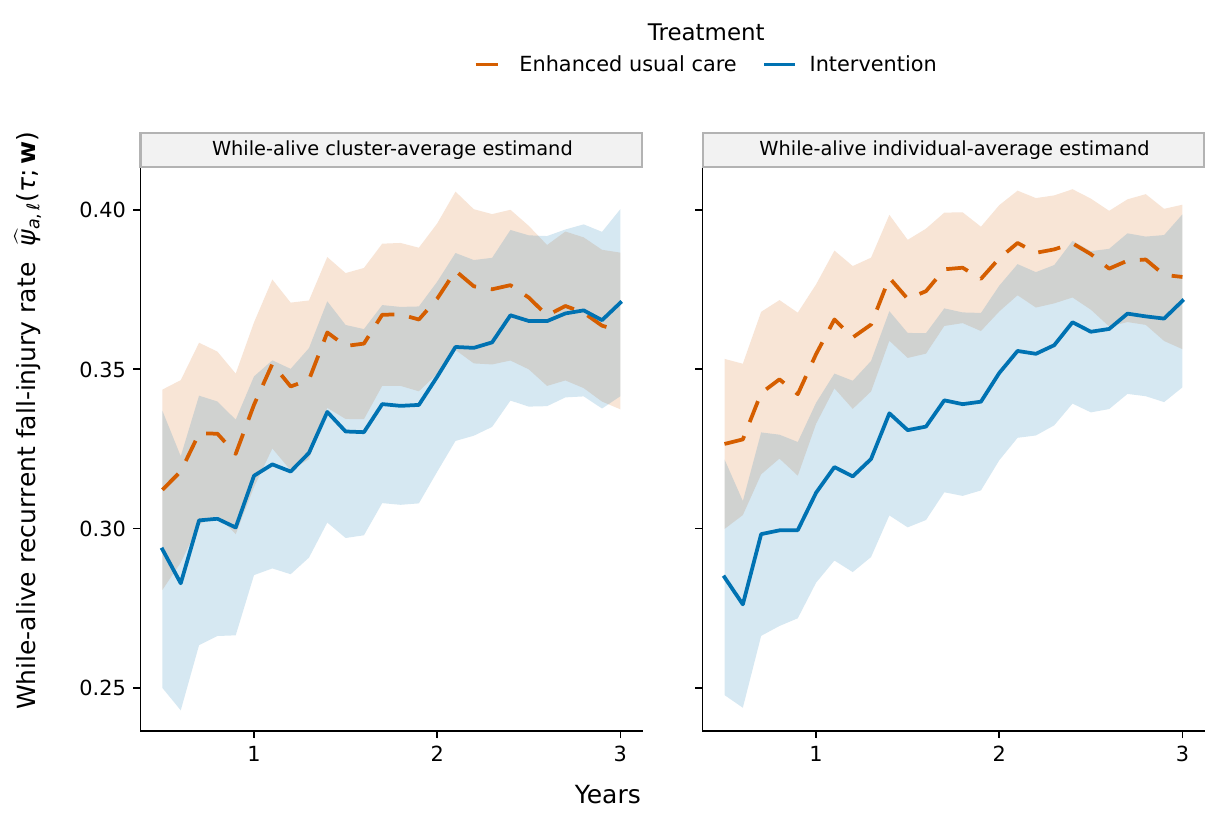}\\[4pt]
\textbf{(b)}\\[-2pt]
\includegraphics[width=0.6\textwidth]{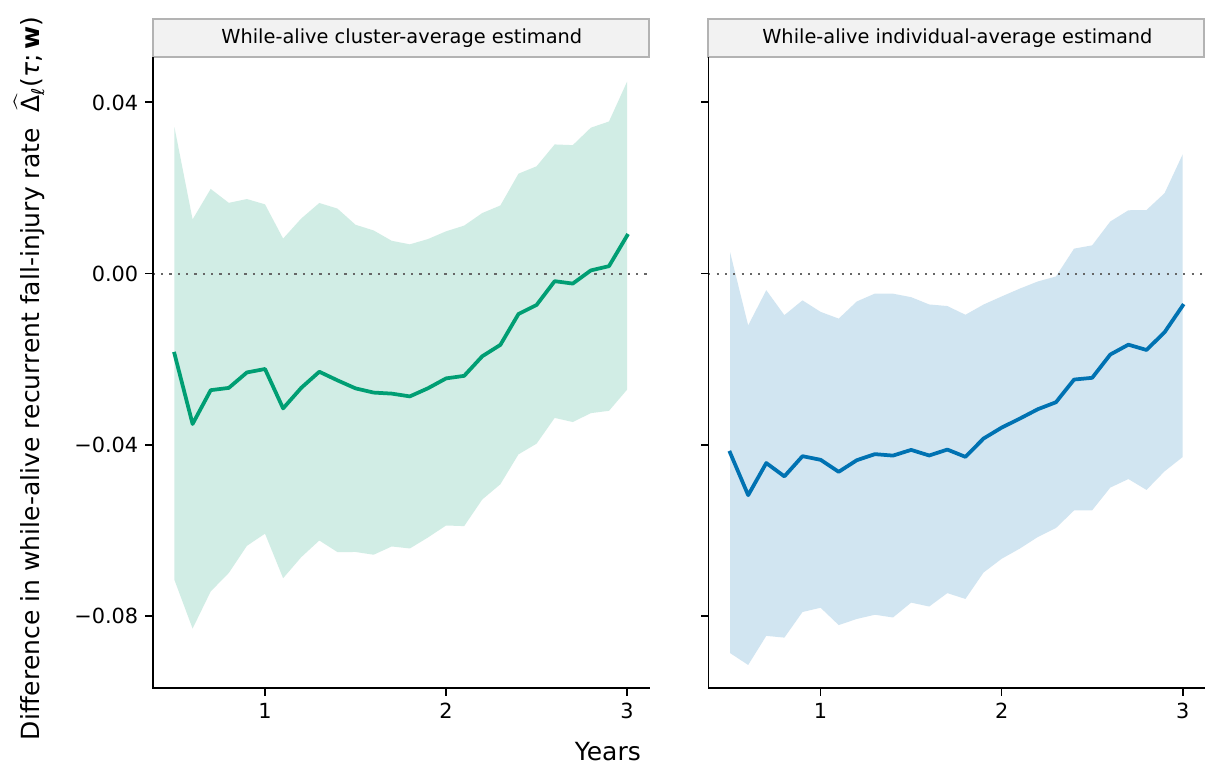}
\caption{Horizon-specific estimates in the STRIDE analysis. Panel (a) shows the arm-specific while-alive recurrent fall-injury rates, \(\widehat\psi_{a,\ell}(\tau;\bw)\), for the cluster-level and individual-level targets. Panel (b) shows the intervention-minus-enhanced-usual-care contrasts, \(\widehat\Delta_{\ell}(\tau;\bw)\); negative values indicate a lower recurrent fall-injury burden per year alive under the intervention. Shaded regions are pointwise 95\% confidence intervals.}
\label{fig:stride_while_alive}
\end{figure}

The STRIDE results complement the original publication in three respects. First, the original primary analysis found a modest reduction in the hazard of a first participant-reported fall injury \citep{bhasin2020randomized}. The early individual-average while-alive estimates are directionally consistent with this finding, but they suggest that the reduction in total recurrent fall-injury burden per year alive diminished as follow-up approached three years. Second, the present analysis uses every recorded recurrent injury rather than only the first event and treats death as a terminal event rather than ordinary censoring. Third, it separates the average-participant and average-practice targets. The original participant-level first-event analysis accommodated clustering in inference but did not report a while-alive cluster-average contrast. These analyses suggest a possible early reduction in recurrent fall-injury burden for the average participant, but no sustained three-year benefit and no clear benefit for the average practice. Because the outcome definition and effect scale differ from the original primary endpoint, these findings should be viewed as complementary rather than contradictory.

\section{Discussion}
\label{sec:discussion}

In this article, we develop doubly robust estimators of the exposure-weighted while-alive recurrent-event rate for both individually-randomized and cluster-randomized trials. This estimand is defined as the ratio of the expected clinically weighted recurrent-event burden to the restricted mean survival time and therefore quantifies recurrent-event burden per unit of population time alive. Its numerator and denominator remain separately interpretable, allowing investigators to report the treatment effects on recurrent-event burden and survival alongside their ratio. Our estimator uses a local Nelson-Aalen representation, with augmented estimating equations for the terminal-event hazard and the recurrent-event rate among those alive. By constructing survival through the product integral and accumulating recurrent-event burden through Stieltjes integration, the estimator respects the natural structure and shape constraints of these component processes rather than treating the final ratio as a single regression functional. The exposure-weighted while-alive estimand complements, rather than replaces, other summaries for recurrent-event trials with death. Time-to-first-event analyses remain useful when the first event is the dominant clinical outcome but discard subsequent recurrences, whereas cumulative recurrent-event means use all events but can be difficult to interpret when treatment also affects survival. Patient-weighted while-alive estimands average subject-specific event rates and therefore answer a different question from the exposure-weighted estimand considered here \citep{mao2022whilealive,wei2023properties,ragni2026patient}. The latter quantifies the expected recurrent-event burden accumulated by the population per expected unit of time alive, making it particularly relevant when clinical burden is naturally expressed over person-time alive, as with hospitalizations, exacerbations, injuries, and other nonfatal recurrent events. Our estimator is closely related to the future-mean one-step estimator of \citet{baer2025causal}. Under a compatible nuisance representation, the two approaches are first-order equivalent. Our local construction nevertheless offers an attractive and computationally simpler implementation. It uses standard working survival models for the terminal-event and recurrent-event processes, does not require specification of a full joint parametric model for their event histories, and produces estimates across all time horizons from a common set of local hazard and rate increments. In contrast, direct implementation of the future-mean approach requires modeling the remaining event burden over future intervals at relevant censoring times and horizons. These differences reflect complementary estimation strategies rather than competing scientific targets, and may guide the choice of implementation in a given application.

Our contribution is particularly relevant for CRTs, in which the cluster is the unit of randomization and inference but recurrent and terminal events are experienced by individuals. The recently developed ``\emph{CRT-Estimands Framework},'' a consensus-based extension of the ICH E9(R1) addendum, emphasizes that an estimand for a CRT must explicitly specify both how individuals and clusters contribute to the target population and how intercurrent events are handled \citep{kahan2026crtestimands,kahan2026crtJAMA,kahan2023estimands,kahan2024demystifying}. Death is a consequential intercurrent event for recurrent nonfatal outcomes because it permanently terminates the period during which further events can occur. The while-alive strategy, one of the five strategies articulated in the CRT-Estimands Framework for addressing death, defines treatment effects through the outcome burden accumulated during survival rather than treating death as ordinary censoring or implicitly combining mortality and morbidity. Despite this conceptual guidance, there has been little corresponding methodological development for marginal while-alive estimands in CRTs. To our knowledge, the only directly related CRT work is \citet{fang2025whilealive}, which developed a direct regression approach for while-alive summaries. That approach targets regression parameters rather than explicitly defined marginal causal estimands and was not designed to distinguish individual-average from cluster-average effects or to accommodate informative cluster size. These distinctions are fundamental because an analysis that weights every participant equally and one that weights every randomized cluster equally generally target different treatment effects when cluster size is associated with prognosis, recurrent-event burden, survival, or effect heterogeneity. The present work therefore provides an important methodological implementation of the CRT-Estimands Framework. It defines paired individual-average and cluster-average marginal while-alive estimands, makes their target population weights explicit, accommodates informative cluster size, and provides doubly robust estimation under covariate-dependent censoring with cluster-level influence-function inference. In this way, the proposed methods bridge the gap between specifying a meaningful strategy for death and actually estimating the resulting treatment effect in CRTs, while preventing the target estimand from being determined implicitly by the choice of regression model.

There are several limitations that we plan to address in future research. First, the censoring assumption is conditional on measured baseline covariates. If dropout or loss to follow-up depends on unmeasured health status or other post-baseline factors not adequately captured by these covariates, the proposed estimators may remain biased. Second, the event-type weights \(\bw\) are treated as prespecified clinical constants. Although this is appropriate for an estimand-based analysis, different choices of \(\bw\) define different causal questions, and sensitivity analyses across clinically plausible weighting schemes may therefore be informative when no single set of weights is broadly accepted. Finally, for computational simplicity and practical implementation, we have relied on familiar semiparametric working models, including arm-specific Cox models for the terminal and censoring processes and multiplicative-rate models for the recurrent-event process. Componentwise double robustness protects against misspecification of either the censoring side or the relevant outcome side, but not against simultaneous misspecification of both. The proposed framework could be extended using cross-fitted, debiased machine-learning estimators of the nuisance functions, thereby, in large samples, permitting more flexible covariate adjustment while preserving valid inference under suitable product-rate and regularity conditions for recurrent-event and censoring processes \citep{hines2022demystifying}.

\section*{Software and supplementary material}
The proposed methods are implemented in the \texttt{DRCRTwa} R package, available at \url{https://github.com/fancy575/DRCRTwa}. All proofs, asymptotic expansions and variance derivations, the efficient influence function, data-generating parameters, and additional simulation and data-analysis results are included in the appended Supplementary Appendix.

\section*{Data availability statement}
The HF-ACTION data are available through the National Heart, Lung, and Blood Institute Biologic Specimen and Data Repository Information Coordinating Center (BioLINCC, \url{https://biolincc.nhlbi.nih.gov}) under an approved data use agreement. The STRIDE data are available through the National Institute on Aging (\url{https://agingresearchbiobank.nia.nih.gov}) under an approved data-use agreement.

\section*{Acknowledgements}
All statements in this report, including its findings and conclusions, are solely those of the authors and do not necessarily represent the views of the National Institutes of Health.

\section*{Funding}
This work is supported by the United States National Institutes of Health, National Heart, Lung, and Blood Institute (grant number 1R01HL178513).

\clearpage
\printbibliography[title={References}]

\clearpage
\appendix
\setcounter{section}{0}
\setcounter{subsection}{0}
\setcounter{subsubsection}{0}
\setcounter{equation}{0}
\setcounter{table}{0}
\setcounter{figure}{0}
\setcounter{theorem}{0}
\setcounter{lemma}{0}
\setcounter{proposition}{0}
\setcounter{corollary}{0}
\setcounter{definition}{0}
\setcounter{remark}{0}
\setcounter{assumption}{0}
\numberwithin{assumption}{section}

\renewcommand{\thesection}{A\arabic{section}}
\renewcommand{\thesubsection}{\thesection.\arabic{subsection}}
\renewcommand{\thesubsubsection}{\thesubsection.\arabic{subsubsection}}
\renewcommand{\theequation}{S\arabic{equation}}
\renewcommand{\thetable}{S\arabic{table}}
\renewcommand{\thefigure}{S\arabic{figure}}

\renewcommand{\theHsection}{appendix.\arabic{section}}
\renewcommand{\theHsubsection}{appendix.\arabic{section}.\arabic{subsection}}
\renewcommand{\theHsubsubsection}{appendix.\arabic{section}.\arabic{subsection}.\arabic{subsubsection}}
\renewcommand{\theHequation}{appendix.\arabic{equation}}
\renewcommand{\theHtable}{appendix.\arabic{table}}
\renewcommand{\theHfigure}{appendix.\arabic{figure}}
\renewcommand{\theHtheorem}{appendix.\arabic{theorem}}
\renewcommand{\theHlemma}{appendix.\arabic{lemma}}
\renewcommand{\theHproposition}{appendix.\arabic{proposition}}
\renewcommand{\theHcorollary}{appendix.\arabic{corollary}}
\renewcommand{\theHdefinition}{appendix.\arabic{definition}}
\renewcommand{\theHremark}{appendix.\arabic{remark}}
\renewcommand{\theHassumption}{appendix.\arabic{section}.\arabic{assumption}}

\section*{Supplementary Appendix}
\addcontentsline{toc}{section}{Supplementary Appendix}
\raggedbottom

Section~\ref{supp:sec:full_observed_identities} establishes the full data and observed data identities used by the augmented Nelson--Aalen construction.  Section~\ref{supp:sec:irt_dr_detailed} gives the expectation-equals-target double-robustness proof for IRTs.  Section~\ref{supp:sec:irt_asymptotic_detailed} derives the full first-order expansion and the analytic variance, explicitly retaining the terms due to estimating the censoring, terminal-event, and recurrent-event nuisance models.  Section~\ref{supp:sec:eif} gives the efficient influence function derivation.  Section~\ref{supp:sec:baer_detailed_equivalence} proves the first-order equivalence with the future-mean estimator of \citet{baer2025causal} under a compatible nuisance representation.  Sections~\ref{supp:sec:crt_dr_detailed} and~\ref{supp:sec:crt_variance_detailed} provide the corresponding CRT proofs and variance formulas.  Section~\ref{supp:sec:simulation_tables_detailed} reports additional simulation results in Tables~\ref{supp:tab:irt_n1600}--\ref{supp:tab:crt_cluster_M80}.

\section{Full data and observed data identities}
\label{supp:sec:full_observed_identities}

\subsection{Full data local representation}

Fix treatment arm \(a\in\{0,1\}\), horizon \(\tau\), and event-type weight vector \(\bw\).  Let \(Y^a(t)=\ind(D^a\ge t)\), \(N^{D,a}(t)=\ind(D^a\le t)\), and \(N^{a,\bw}(t)=\sum_{k=1}^K w_kN_k^a(t)\), with recurrent-event processes stopped at \(D^a\).  Define \(R_a(t)=\Ebb\{Y^a(t)\}\).  The local terminal-event and recurrent-event increments are
\[
  \dd\Lambda_a^D(t)=\frac{\Ebb\{\dd N^{D,a}(t)\}}{R_a(t-)},
  \qquad
  \dd\Lambda_a^{R,\bw}(t)=\frac{\Ebb\{\dd N^{a,\bw}(t)\}}{R_a(t-)}.
\]
The left limit in the denominator is immaterial in continuous time and is kept to make the notation valid under ties and point masses.  The next proposition gives the local representation used throughout the supplement.

\begin{proposition}[Full data Nelson--Aalen representation]
\label{supp:prop:full_data_na}
Suppose \(\Ebb\{N^{a,\bw}(\tau)\}<\infty\), \(\Ebb(D^a\wedge\tau)<\infty\), and \(\inf_{0\le t\le\tau}R_a(t-)>0\).  Then
\[
  S_a(t)=\Pbb(D^a\ge t)=\prod_{0<u\le t}\{1-\dd\Lambda_a^D(u)\},
  \qquad
  \nu_a(\tau)=\int_0^\tau S_a(t)\dd t,
  \qquad
  \mu_a(\tau;\bw)=\int_0^\tau S_a(t-)\dd\Lambda_a^{R,\bw}(t).
\]
\end{proposition}

\begin{proof}
Because \(Y^a(t)\) changes only when the terminal event occurs,
\begin{align*}
  \dd R_a(t)
  &=\Ebb\{\dd Y^a(t)\} \\
  &=\Ebb\{-\dd N^{D,a}(t)\} \\
  &=-\Ebb\{\dd N^{D,a}(t)\} \\
  &=-R_a(t-)\frac{\Ebb\{\dd N^{D,a}(t)\}}{R_a(t-)} \\
  &=-R_a(t-)\dd\Lambda_a^D(t).
\end{align*}
The cadlag solution of \(\dd R_a(t)=-R_a(t-)\dd\Lambda_a^D(t)\) with \(R_a(0)=1\) is the product integral \(R_a(t)=\prod_{0<u\le t}\{1-\dd\Lambda_a^D(u)\}\), so \(R_a(t)=S_a(t)\).  For the restricted mean survival time, Tonelli's theorem gives
\begin{align*}
  \nu_a(\tau)
  &=\Ebb(D^a\wedge\tau) \\
  &=\Ebb\left\{\int_0^\tau \ind(D^a\ge t)\dd t\right\} \\
  &=\int_0^\tau \Ebb\{Y^a(t)\}\dd t \\
  &=\int_0^\tau S_a(t)\dd t.
\end{align*}
For the recurrent-event burden, using the definition of \(\dd\Lambda_a^{R,\bw}\), finite first moment, and Tonelli's theorem,
\begin{align*}
  \int_0^\tau S_a(t-)\dd\Lambda_a^{R,\bw}(t)
  &=\int_0^\tau R_a(t-)\frac{\Ebb\{\dd N^{a,\bw}(t)\}}{R_a(t-)} \\
  &=\int_0^\tau \Ebb\{\dd N^{a,\bw}(t)\} \\
  &=\Ebb\left\{\int_0^\tau \dd N^{a,\bw}(t)\right\} \\
  &=\Ebb\{N^{a,\bw}(\tau)-N^{a,\bw}(0)\} \\
  &=\Ebb\{N^{a,\bw}(\tau)\}.
\end{align*}
We use the convention \(N^{a,\bw}(0)=0\), as in the main text.  This proves the proposition.
\end{proof}

\subsection{Observed-data inverse-censoring identities}

Let \(\xi^a=\ind(A=a)/\pi_a\).  Given \(\bZ\), define the true conditional censoring survival \(K_0^a(t-\mid\bZ)=\Pbb(C^a\ge t\mid\bZ)\), terminal survival \(H_0^a(t-\mid\bZ)=\Pbb(D^a\ge t\mid\bZ)\), conditional event increments
\[
  \dd q_0^D(t;a\mid\bZ)=\Ebb\{\dd N^{D,a}(t)\mid\bZ\},
  \qquad
  \dd q_0^{R,\bw}(t;a\mid\bZ)=\Ebb\{\dd N^{a,\bw}(t)\mid\bZ\},
\]
and \(r_0^a(t\mid\bZ)=H_0^a(t-\mid\bZ)\).  Let \(K^\star\), \(H^\star\), \(\dd q^{D,\star}\), \(\dd q^{R,\bw,\star}\), and \(r^\star=H^\star(t-)\) be probability limits of the fitted working models.  The superscript \(\star\) denotes a probability limit and does not imply correctness.

\begin{lemma}[Conditional inverse-censoring identity]
\label{supp:lem:ipcw_identity_detailed}
Under Assumption \ref{asm:irt_consistency}, \ref{asm:irt_randomization}, and \ref{asm:irt_censoring}, for any finite variation full data increment \(\dd G^a(t)\) stopped by death,
\[
  \Ebb\left[\xi^a\{K^\star(t-\mid\bZ)\}^{-1}\ind(C^a\ge t)\dd G^a(t)\mid\bZ\right]
  =\frac{K_0^a(t-\mid\bZ)}{K^\star(t-\mid\bZ)}\Ebb\{\dd G^a(t)\mid\bZ\}.
\]
The same equality with \(\dd G^a(t)\) replaced by \(Y^a(t)\) is
\[
  \Ebb\left[\xi^a\{K^\star(t-\mid\bZ)\}^{-1}\ind(C^a\ge t)Y^a(t)\mid\bZ\right]
  =\frac{K_0^a(t-\mid\bZ)}{K^\star(t-\mid\bZ)}H_0^a(t-\mid\bZ).
\]
\end{lemma}

\begin{proof}
By randomization, \(\Ebb(\xi^a\mid\bZ)=1\), and by consistency the arm-specific observed history is the corresponding potential history among subjects with \(A=a\).  By conditional independent censoring, \(\ind(C^a\ge t)\) is conditionally independent of the outcome-history increment \(\dd G^a(t)\) given \(\bZ\).  Therefore, using consecutive conditioning,
\begin{align*}
&\quad \quad \Ebb\left[\xi^a\{K^\star(t-\mid\bZ)\}^{-1}\ind(C^a\ge t)\dd G^a(t)\mid\bZ\right] \\
&\quad=\{K^\star(t-\mid\bZ)\}^{-1}
   \Ebb\left[\Ebb\{\xi^a\ind(C^a\ge t)\dd G^a(t)\mid\bZ,A\}\mid\bZ\right] \\
&\quad=\{K^\star(t-\mid\bZ)\}^{-1}
   \Ebb(\xi^a\mid\bZ)
   \Ebb\{\ind(C^a\ge t)\dd G^a(t)\mid\bZ\} \\
&\quad=\{K^\star(t-\mid\bZ)\}^{-1}
   \Ebb(\xi^a\mid\bZ)
   \Ebb\{\ind(C^a\ge t)\mid\bZ\}
   \Ebb\{\dd G^a(t)\mid\bZ\} \\
&\quad=\frac{K_0^a(t-\mid\bZ)}{K^\star(t-\mid\bZ)}\Ebb\{\dd G^a(t)\mid\bZ\}.
\end{align*}
Replacing \(\dd G^a(t)\) by \(Y^a(t)\) gives the second equality because \(\Ebb\{Y^a(t)\mid\bZ\}=H_0^a(t-\mid\bZ)\).
\end{proof}

\subsection{Censoring-martingale calibration}

Let \(\dd\Lambda_C^\star(t\mid\bZ)\) be the probability-limit censoring cumulative hazard, and let \(K^\star(t\mid\bZ)=\prod_{0<u\le t}\{1-\dd\Lambda_C^\star(u\mid\bZ)\}\).  Write
\[
  \dd M_C^\star(t)=\dd N^C(t)-Y^\dagger(t)\dd\Lambda_C^\star(t\mid\bZ),
  \qquad
  U^\star(t)=1-\int_{(0,t)}\frac{\dd M_C^\star(u)}{K^\star(u-\mid\bZ)H^\star(u-\mid\bZ)}.
\]
The following calibration is the algebraic source of double robustness.

\begin{lemma}[Calibration of \(U^\star\)]
\label{supp:lem:calibration_detailed}
Under Assumption \ref{asm:irt_consistency}, \ref{asm:irt_randomization}, and \ref{asm:irt_censoring}, at each fixed \(t\in[0,\tau]\): (i) if \(K^\star=K_0^a\), then \(\Ebb\{\xi^aU^\star(t)\mid\bZ\}=1\); (ii) if \(H^\star=H_0^a\), then
\[
  \Ebb\{\xi^aU^\star(t)\mid\bZ\}=K_0^a(t-\mid\bZ)/K^\star(t-\mid\bZ).
\]
\end{lemma}

\begin{proof}
If \(K^\star=K_0^a\), then \(\dd M_C^\star\) is the true censoring martingale increment.  The integrand \(\{K_0^aH^\star\}^{-1}\) is predictable and bounded by positivity.  Hence
\begin{align*}
  \Ebb\{\xi^aU^\star(t)\mid\bZ\}
  &=1-\Ebb\left[\xi^a\int_{(0,t)}\frac{\dd M_C^\star(u)}{K_0^a(u-\mid\bZ)H^\star(u-\mid\bZ)}\mid\bZ\right] \\
  &=1-\int_{(0,t)}\frac{\Ebb\{\xi^a\dd M_C^\star(u)\mid\cF_{u-},\bZ\}}{K_0^a(u-\mid\bZ)H^\star(u-\mid\bZ)} \\
  &=1.
\end{align*}
For the second part, assume \(H^\star=H_0^a\).  Conditional on \(A=a\) and \(\bZ\), the true compensator of \(N^C\) is \(Y^\dagger(u)\dd\Lambda_{C0}^a(u\mid\bZ)\).  Therefore
\begin{align*}
  \Ebb\{\dd M_C^\star(u)\mid A=a,\bZ\}
  &=\Ebb\{\dd N^C(u)-Y^\dagger(u)\dd\Lambda_C^\star(u\mid\bZ)\mid A=a,\bZ\} \\
  &=\Ebb\{Y^\dagger(u)\mid A=a,\bZ\}\{\dd\Lambda_{C0}^a(u\mid\bZ)-\dd\Lambda_C^\star(u\mid\bZ)\} \\
  &=H_0^a(u-\mid\bZ)K_0^a(u-\mid\bZ)\{\dd\Lambda_{C0}^a(u\mid\bZ)-\dd\Lambda_C^\star(u\mid\bZ)\}.
\end{align*}
Substituting this expression into the definition of \(U^\star\), and using \(\Ebb(\xi^a\mid\bZ)=1\), gives
\begin{align*}
  \Ebb\{\xi^aU^\star(t)\mid\bZ\}
  &=1-\int_{(0,t)}\frac{H_0^a(u-\mid\bZ)K_0^a(u-\mid\bZ)}{K^\star(u-\mid\bZ)H_0^a(u-\mid\bZ)}
     \{\dd\Lambda_{C0}^a(u\mid\bZ)-\dd\Lambda_C^\star(u\mid\bZ)\} \\
  &=1-\int_{(0,t)}\frac{K_0^a(u-\mid\bZ)}{K^\star(u-\mid\bZ)}
     \{\dd\Lambda_{C0}^a(u\mid\bZ)-\dd\Lambda_C^\star(u\mid\bZ)\}.
\end{align*}
Let \(R_K(t\mid\bZ)=K_0^a(t-\mid\bZ)/K^\star(t-\mid\bZ)\).  By the Duhamel identity for a ratio of product integrals,
\begin{align*}
  R_K(t\mid\bZ)
  &=1+\int_{(0,t)}R_K(u-\mid\bZ)\{\dd\Lambda_C^\star(u\mid\bZ)-\dd\Lambda_{C0}^a(u\mid\bZ)\} \\
  &=1-\int_{(0,t)}R_K(u-\mid\bZ)\{\dd\Lambda_{C0}^a(u\mid\bZ)-\dd\Lambda_C^\star(u\mid\bZ)\}.
\end{align*}
The last display is the preceding conditional expectation, proving the claim.
\end{proof}

\section{Double-robustness proof in IRTs}
\label{supp:sec:irt_dr_detailed}

\subsection{Population limits of the local estimating equations}

For \(h\in\{D,R\}\), write \(\dd N^h(t)=\dd N^D(t)\) when \(h=D\) and \(\dd N^h(t)=\dd N^{\bw}(t)\) when \(h=R\).  The corresponding fitted full data increment is \(\dd q^{h,\star}\).  The probability limits of the local numerator and denominator of the augmented Nelson--Aalen estimator are
\[
\begin{aligned}
  \dd\mathcal A_h^\star(t)
  &=\Ebb\left[\xi^a\{K^\star(t-\mid\bZ)\}^{-1}\dd N^h(t)+\{1-\xi^aU^\star(t)\}\dd q^{h,\star}(t;a\mid\bZ)\right],\\
  \mathcal B^\star(t)
  &=\Ebb\left[\xi^a\{K^\star(t-\mid\bZ)\}^{-1}Y(t)+\{1-\xi^aU^\star(t)\}r^\star(t\mid\bZ)\right].
\end{aligned}
\]
The limiting local estimator is \(\dd\Lambda_h^\star(t)=\dd\mathcal A_h^\star(t)/\mathcal B^\star(t)\).  We prove below that \(\dd\Lambda_D^\star(t)=\dd\Lambda_a^D(t)\) and \(\dd\Lambda_R^\star(t)=\dd\Lambda_a^{R,\bw}(t)\) under the corresponding double-robustness conditions.

\begin{proposition}[Terminal-event component]
\label{supp:prop:irt_terminal_dr_detailed}
Under Assumption \ref{asm:irt_consistency}, \ref{asm:irt_randomization}, \ref{asm:irt_censoring} and \ref{asm:irt_asymptotic_regularity}, if either the censoring model is correct, \(\cC_a\), or the terminal-event outcome model is correct, \(\cO_a^D\), then \(\dd\Lambda_D^\star(t)=\dd\Lambda_a^D(t)\) for all \(t\le\tau\).  Hence \(\widehat S_a^{\mathrm{DR}}\) and \(\widehat\nu_a^{\mathrm{DR}}(\tau)\) are consistent.
\end{proposition}

\begin{proof}
First suppose \(\cC_a\) holds, so \(K^\star=K_0^a\).  Conditional on \(\bZ\), Lemma~\ref{supp:lem:ipcw_identity_detailed} and Lemma~\ref{supp:lem:calibration_detailed} give
\begin{align*}
&\Ebb\left[\xi^a\{K_0^a(t-\mid\bZ)\}^{-1}\dd N^D(t)+\{1-\xi^aU^\star(t)\}\dd q^{D,\star}(t;a\mid\bZ)\mid\bZ\right] \\
&\quad=\Ebb\left[\xi^a\{K_0^a(t-\mid\bZ)\}^{-1}\dd N^D(t)\mid\bZ\right]
       +\Ebb\left[\{1-\xi^aU^\star(t)\}\dd q^{D,\star}(t;a\mid\bZ)\mid\bZ\right] \\
&\quad=\dd q_0^D(t;a\mid\bZ)
       +\{1-\Ebb(\xi^aU^\star(t)\mid\bZ)\}\dd q^{D,\star}(t;a\mid\bZ) \\
&\quad=\dd q_0^D(t;a\mid\bZ)+\{1-1\}\dd q^{D,\star}(t;a\mid\bZ) \\
&\quad=\dd q_0^D(t;a\mid\bZ).
\end{align*}
The denominator is treated identically:
\begin{align*}
&\Ebb\left[\xi^a\{K_0^a(t-\mid\bZ)\}^{-1}Y(t)+\{1-\xi^aU^\star(t)\}r^\star(t\mid\bZ)\mid\bZ\right] \\
&\quad=r_0^a(t\mid\bZ)+\{1-\Ebb(\xi^aU^\star(t)\mid\bZ)\}r^\star(t\mid\bZ) \\
&\quad=r_0^a(t\mid\bZ).
\end{align*}
Taking expectations over \(\bZ\) yields
\begin{align*}
  \dd\mathcal A_D^\star(t)
  &=\Ebb\{\dd q_0^D(t;a\mid\bZ)\}
    =\Ebb\{\Ebb(\dd N^{D,a}(t)\mid\bZ)\}
    =\Ebb\{\dd N^{D,a}(t)\},\\
  \mathcal B^\star(t)
  &=\Ebb\{r_0^a(t\mid\bZ)\}
    =\Ebb\{\Ebb(Y^a(t)\mid\bZ)\}
    =\Ebb\{Y^a(t)\},\\
  \dd\Lambda_D^\star(t)
  &=\frac{\dd\mathcal A_D^\star(t)}{\mathcal B^\star(t)}
    =\frac{\Ebb\{\dd N^{D,a}(t)\}}{\Ebb\{Y^a(t)\}}
    =\dd\Lambda_a^D(t).
\end{align*}

Now suppose \(\cO_a^D\) holds.  Then \(r^\star(t\mid\bZ)=r_0^a(t\mid\bZ)\) and \(\dd q^{D,\star}(t;a\mid\bZ)=\dd q_0^D(t;a\mid\bZ)\), whereas \(K^\star\) may be incorrect.  Let \(R_K(t\mid\bZ)=K_0^a(t-\mid\bZ)/K^\star(t-\mid\bZ)\).  By Lemmas~\ref{supp:lem:ipcw_identity_detailed} and~\ref{supp:lem:calibration_detailed},
\begin{align*}
&\Ebb\left[\xi^a\{K^\star(t-\mid\bZ)\}^{-1}\dd N^D(t)+\{1-\xi^aU^\star(t)\}\dd q_0^D(t;a\mid\bZ)\mid\bZ\right] \\
&\quad=R_K(t\mid\bZ)\dd q_0^D(t;a\mid\bZ)
       +\{1-\Ebb(\xi^aU^\star(t)\mid\bZ)\}\dd q_0^D(t;a\mid\bZ) \\
&\quad=R_K(t\mid\bZ)\dd q_0^D(t;a\mid\bZ)+\{1-R_K(t\mid\bZ)\}\dd q_0^D(t;a\mid\bZ) \\
&\quad=\dd q_0^D(t;a\mid\bZ).
\end{align*}
For the denominator,
\begin{align*}
&\Ebb\left[\xi^a\{K^\star(t-\mid\bZ)\}^{-1}Y(t)+\{1-\xi^aU^\star(t)\}r_0^a(t\mid\bZ)\mid\bZ\right] \\
&\quad=R_K(t\mid\bZ)r_0^a(t\mid\bZ)+\{1-R_K(t\mid\bZ)\}r_0^a(t\mid\bZ) \\
&\quad=r_0^a(t\mid\bZ).
\end{align*}
The same unconditional equalities and ratio calculation displayed above therefore give \(\dd\Lambda_D^\star(t)=\dd\Lambda_a^D(t)\).  The product-integral map \(\Lambda^D\mapsto S\) and the integral map \(S\mapsto\int_0^\tau S(t)\dd t\) are continuous under bounded variation and positivity.  Thus \(\widehat S_a^{\mathrm{DR}}\to_p S_a\) uniformly at continuity points and \(\widehat\nu_a^{\mathrm{DR}}(\tau)\to_p\nu_a(\tau)\).
\end{proof}

\begin{proposition}[Weighted recurrent-event component]
\label{supp:prop:irt_recurrent_dr_detailed}
Under Assumption \ref{asm:irt_consistency}, \ref{asm:irt_randomization}, \ref{asm:irt_censoring} and \ref{asm:irt_asymptotic_regularity}, if either \(\cC_a\) holds or \(\cO_a^D\cap\cO_a^R\) holds, then \(\dd\Lambda_R^\star(t)=\dd\Lambda_a^{R,\bw}(t)\) for all \(t\le\tau\).  Together with Proposition~\ref{supp:prop:irt_terminal_dr_detailed}, this implies \(\widehat\mu_a^{\mathrm{DR}}(\tau;\bw)\to_p\mu_a(\tau;\bw)\).
\end{proposition}

\begin{proof}
The denominator of the recurrent-event local estimator is the same \(\mathcal B^\star(t)\) as in the terminal-event component.  If \(\cC_a\) holds, then conditional on \(\bZ\),
\begin{align*}
&\Ebb\left[\xi^a\{K_0^a(t-\mid\bZ)\}^{-1}\dd N^{\bw}(t)+\{1-\xi^aU^\star(t)\}\dd q^{R,\bw,\star}(t;a\mid\bZ)\mid\bZ\right] \\
&\quad=\Ebb\left[\xi^a\{K_0^a(t-\mid\bZ)\}^{-1}\dd N^{\bw}(t)\mid\bZ\right]
       +\{1-\Ebb(\xi^aU^\star(t)\mid\bZ)\}\dd q^{R,\bw,\star}(t;a\mid\bZ) \\
&\quad=\dd q_0^{R,\bw}(t;a\mid\bZ)+\{1-1\}\dd q^{R,\bw,\star}(t;a\mid\bZ) \\
&\quad=\dd q_0^{R,\bw}(t;a\mid\bZ).
\end{align*}
The denominator has conditional expectation \(r_0^a(t\mid\bZ)\) by Proposition~\ref{supp:prop:irt_terminal_dr_detailed}.  Therefore
\begin{align*}
  \dd\Lambda_R^\star(t)
  &=\frac{\Ebb\{\dd q_0^{R,\bw}(t;a\mid\bZ)\}}{\Ebb\{r_0^a(t\mid\bZ)\}} \\
  &=\frac{\Ebb\{\dd N^{a,\bw}(t)\}}{\Ebb\{Y^a(t)\}} \\
  &=\dd\Lambda_a^{R,\bw}(t).
\end{align*}
If instead \(\cO_a^D\cap\cO_a^R\) holds, then \(r^\star=r_0^a\), \(\dd q^{D,\star}=\dd q_0^D\), and \(\dd q^{R,\bw,\star}=\dd q_0^{R,\bw}\).  With \(R_K(t\mid\bZ)=K_0^a(t-\mid\bZ)/K^\star(t-\mid\bZ)\),
\begin{align*}
&\Ebb\left[\xi^a\{K^\star(t-\mid\bZ)\}^{-1}\dd N^{\bw}(t)+\{1-\xi^aU^\star(t)\}\dd q_0^{R,\bw}(t;a\mid\bZ)\mid\bZ\right] \\
&\quad=R_K(t\mid\bZ)\dd q_0^{R,\bw}(t;a\mid\bZ)
       +\{1-R_K(t\mid\bZ)\}\dd q_0^{R,\bw}(t;a\mid\bZ) \\
&\quad=\dd q_0^{R,\bw}(t;a\mid\bZ),\\
&\Ebb\left[\xi^a\{K^\star(t-\mid\bZ)\}^{-1}Y(t)+\{1-\xi^aU^\star(t)\}r_0^a(t\mid\bZ)\mid\bZ\right] \\
&\quad=R_K(t\mid\bZ)r_0^a(t\mid\bZ)+\{1-R_K(t\mid\bZ)\}r_0^a(t\mid\bZ) \\
&\quad=r_0^a(t\mid\bZ).
\end{align*}
Taking expectations and forming the ratio again gives \(\dd\Lambda_R^\star(t)=\dd\Lambda_a^{R,\bw}(t)\).  By Proposition~\ref{supp:prop:irt_terminal_dr_detailed}, the survival factor in the Stieltjes integral is also consistent.  Hence
\begin{align*}
  \widehat\mu_a^{\mathrm{DR}}(\tau;\bw)
  &=\int_0^\tau\widehat S_a^{\mathrm{DR}}(t-)\dd\widehat\Lambda_a^{R,\mathrm{DR},\bw}(t)
    \overset{p}{\longrightarrow}
    \int_0^\tau S_a(t-)\dd\Lambda_a^{R,\bw}(t)
    =\mu_a(\tau;\bw).
\end{align*}
\end{proof}

\begin{theorem}[Componentwise double robustness for the while-alive rate]
\label{supp:thm:irt_dr_detailed}
Under Assumption \ref{asm:irt_consistency}, \ref{asm:irt_randomization}, \ref{asm:irt_censoring} and \ref{asm:irt_asymptotic_regularity}, \(\widehat\nu_a^{\mathrm{DR}}(\tau)\) is consistent if \(\cC_a\) or \(\cO_a^D\) holds; \(\widehat\mu_a^{\mathrm{DR}}(\tau;\bw)\) is consistent if \(\cC_a\) or \(\cO_a^D\cap\cO_a^R\) holds; and \(\widehat\psi_a^{\mathrm{DR}}(\tau;\bw)\) is consistent whenever both components are consistent and \(\nu_a(\tau)>0\).  The treatment contrast is consistent if the corresponding arm-specific conditions hold in both arms.
\end{theorem}

\begin{proof}
The consistency statements for \(\widehat\nu_a^{\mathrm{DR}}\) and \(\widehat\mu_a^{\mathrm{DR}}\) are exactly Propositions~\ref{supp:prop:irt_terminal_dr_detailed} and~\ref{supp:prop:irt_recurrent_dr_detailed}.  Since \(\nu_a(\tau)>0\), the map \((x,y)\mapsto x/y\) is continuous at \((\mu_a(\tau;\bw),\nu_a(\tau))\).  Therefore
\begin{align*}
  \widehat\psi_a^{\mathrm{DR}}(\tau;\bw)
  &=\frac{\widehat\mu_a^{\mathrm{DR}}(\tau;\bw)}{\widehat\nu_a^{\mathrm{DR}}(\tau)} \\
  &\overset{p}{\longrightarrow}
    \frac{\mu_a(\tau;\bw)}{\nu_a(\tau)} \\
  &=\psi_a(\tau;\bw).
\end{align*}
For \(\Delta(\tau;\bw)=\psi_1(\tau;\bw)-\psi_0(\tau;\bw)\), linearity gives
\[
  \widehat\Delta^{\mathrm{DR}}(\tau;\bw)
  =\widehat\psi_1^{\mathrm{DR}}(\tau;\bw)-\widehat\psi_0^{\mathrm{DR}}(\tau;\bw)
  \overset{p}{\longrightarrow}
  \psi_1(\tau;\bw)-\psi_0(\tau;\bw)=\Delta(\tau;\bw).
\]
\end{proof}

\section{Asymptotic expansion and analytic variance in IRTs}
\label{supp:sec:irt_asymptotic_detailed}

This section derives the influence functions used in the analytic variance.  The formulas keep the nuisance-estimation terms instead of treating \(K\), \(H\), \(q^D\), \(q^{R,\bw}\), and \(U\) as fixed.  All expressions are arm-specific; to reduce visual clutter, the arm superscript is suppressed when the meaning is unambiguous.  Section~\ref{supp:sec:local_ratio_derivative_detailed} gives the local ratio linearization, and Section~\ref{supp:sec:functional_delta_detailed} gives the product-integral, Stieltjes-integral, and ratio delta-method steps.

\subsection{Estimating equations for nuisance working models}

Throughout this section, $K_i^{a,\star}$, $H_i^{a,\star}$, $\dd q_i^{D,\star}$, and $\dd q_i^{R,\bw,\star}$ denote the uniform probability limits of the fitted nuisance functions, as defined in Section 2.2 of the main text. Let $U_i^{a,\star}(t)$ denote the uniform probability limit of $\widehat U_i^a(t)$, define $r_i^{a,\star}(t)=H_i^{a,\star}(t-)$, and define the limiting augmented increments
\begin{align*}
  \dd A_i^{D,a,\star}(t)&=\xi_i^a\{K_i^{a,\star}(t-)\}^{-1}\dd N_i^D(t)+\{1-\xi_i^a U_i^{a,\star}(t)\}\dd q_i^{D,\star}(t;a), \\
  \dd A_i^{R,a,\star}(t)&=\xi_i^a\{K_i^{a,\star}(t-)\}^{-1}\dd N_i^{\bw}(t)+\{1-\xi_i^a U_i^{a,\star}(t)\}\dd q_i^{R,\bw,\star}(t;a),
\end{align*}
and the limiting augmented risk set $B_i^{a,\star}(t)=\xi_i^a\{K_i^{a,\star}(t-)\}^{-1}Y_i(t)+\{1-\xi_i^a U_i^{a,\star}(t)\}r_i^{a,\star}(t)$, with $B_a^\star(t)=\Ebb\{B_i^{a,\star}(t)\}$. The probability limits of the two local estimators are then
\begin{align}
  \dd\Lambda_a^{D,\star}(t)&=\frac{\Ebb\{\dd A_i^{D,a,\star}(t)\}}{B_a^\star(t)},
  \label{eq:irt_LamD_star}\\
  \dd\Lambda_a^{R,\bw,\star}(t)&=\frac{\Ebb\{\dd A_i^{R,a,\star}(t)\}}{B_a^\star(t)}.
  \label{eq:irt_LamR_star}
\end{align}

\begin{assumption}[Asymptotic regularity]
\label{asm:irt_asymptotic_regularity}
For each $a\in\{0,1\}$, the following conditions hold on $[0,\tau]$.
\begin{enumerate}[label=(\roman*),ref=\ref{asm:irt_asymptotic_regularity}(\roman*)]
\item\label{asm:reg_iid} The observed data $\bO_1,\ldots,\bO_n$ are independent and identically distributed.
\item\label{asm:reg_positivity} The treatment probability $\pi_a$, the true censoring survival function $K_i^a(t)$, the limiting censoring survival function $K_i^{a,\star}(t)$, and the limiting augmented risk set $B_a^\star(t)$ are uniformly bounded away from zero.
\item\label{asm:reg_variation} The terminal event cumulative hazards, censoring cumulative hazards, and recurrent event cumulative mean functions have uniformly bounded total variation on $[0,\tau]$.
\item\label{asm:reg_moment} The weighted recurrent event count satisfies $\Ebb\{N_i^{a,\bw}(\tau)^2\}<\infty$.
\item\label{asm:reg_nuisance} The fitted censoring, terminal event, and recurrent event working models converge uniformly to their probability limits and admit asymptotically linear representations.
\item\label{asm:reg_hadamard} The product integral map defining the marginal survival function and the Stieltjes integral maps defining $\mu_a(\tau;\bw)$ and $\nu_a(\tau)$ are Hadamard differentiable at their limiting values.
\end{enumerate}
\end{assumption}
 
Assumption \ref{asm:irt_asymptotic_regularity} collects standard regularity conditions for the working models and for smooth functionals of cumulative hazard estimators \citep{andersen1982cox,lin2000semiparametric,andersen1993statistical,tsiatis2006semiparametric,vandervaart1998asymptotic}. Assumption \ref{asm:reg_iid} specifies the independent and identically distributed sampling framework. Assumption \ref{asm:reg_positivity} imposes positivity conditions, ensuring stable inverse probability weights and well-defined augmentation terms. Assumption \ref{asm:reg_variation} requires bounded variation of the cumulative hazard and mean functions, guaranteeing that the associated product integral and Stieltjes integral functionals are well behaved. Assumption \ref{asm:reg_moment} is a finite moment condition ensuring finite asymptotic variance of the weighted recurrent event burden. Assumption \ref{asm:reg_nuisance} collects the standard empirical process conditions for Cox type censoring and terminal event working models and the LWYY marginal proportional rates working model for recurrent events, ensuring uniformly consistent and asymptotically linear nuisance estimators even under possible working model misspecification. Assumption \ref{asm:reg_hadamard} requires Hadamard differentiability of the product integral and Stieltjes integral maps so that the functional delta method applies to the proposed estimators. Full details, including the nuisance influence function expansions under possible misspecification, are given in Appendix \ref{supp:sec:irt_asymptotic_detailed}.

For censoring, the arm-specific Cox working model in the main text is \(\lambda_C^a(t\mid\bZ_i)=\lambda_{C0}^a(t)\exp\{\balpha_C^{a\tr}\bh_C(\bZ_i)\}\).  The score equation and Breslow estimator are
\[
\begin{aligned}
  0
  &=\sum_{i=1}^n\int_0^\tau \xi_i^a\{\bh_C(\bZ_i)-\bar\bh_C(t;\widehat{\balpha}_C^a)\}\dd N_i^C(t),\\
  \dd\widehat\Lambda_{C0}^a(t)
  &=\frac{\sum_{i=1}^n \xi_i^a\dd N_i^C(t)}{\sum_{i=1}^n\xi_i^aY_i^\dagger(t)\exp\{\widehat{\balpha}_C^{a\tr}\bh_C(\bZ_i)\}},\\
  \bar\bh_C(t;\balpha)
  &=\frac{\sum_{i=1}^n\xi_i^aY_i^\dagger(t)\exp\{\balpha\tr\bh_C(\bZ_i)\}\bh_C(\bZ_i)}{\sum_{i=1}^n\xi_i^aY_i^\dagger(t)\exp\{\balpha\tr\bh_C(\bZ_i)\}}.
\end{aligned}
\]
For the terminal event, the same display holds after replacing \((C,Y^\dagger,N^C,\balpha_C,\bh_C)\) by \((D,Y,N^D,\bbeta_D,\bh_D)\).  For recurrent-event type \(k\), the LWYY working model in the main text uses
\[
  \Ebb\{\dd N_{ik}^a(t)\mid\bZ_i\}=H_i^a(t-\mid\bZ_i)\exp\{\bgamma_{Rk}^{a\tr}\bh_R(\bZ_i)\}\dd\Lambda_{Rk0}^a(t).
\]
The estimating equation and Breslow-type baseline estimator are
\[
\begin{aligned}
  0
  &=\sum_{i=1}^n\int_0^\tau \xi_i^a\{\bh_R(\bZ_i)-\bar\bh_{Rk}(t;\widehat{\bgamma}_{Rk}^a)\}\dd N_{ik}(t),\\
  \dd\widehat\Lambda_{Rk0}^a(t)
  &=\frac{\sum_{i=1}^n\xi_i^a\dd N_{ik}(t)}{\sum_{i=1}^n\xi_i^aY_i(t)\exp\{\widehat{\bgamma}_{Rk}^{a\tr}\bh_R(\bZ_i)\}},\\
  \bar\bh_{Rk}(t;\bgamma)
  &=\frac{\sum_{i=1}^n\xi_i^aY_i(t)\exp\{\bgamma\tr\bh_R(\bZ_i)\}\bh_R(\bZ_i)}{\sum_{i=1}^n\xi_i^aY_i(t)\exp\{\bgamma\tr\bh_R(\bZ_i)\}}.
\end{aligned}
\]
The LWYY residual is a mean-zero estimating-equation residual; it is a martingale if a multiplicative intensity model also holds, but the large-sample expansion only requires the mean-zero score condition.

\subsection{First-order expansion of the nuisance estimators}

The first order effects of estimating the working models enter the influence processes as follows. For the censoring and terminal event models, they enter through the LWYY estimating equation and the Breslow baseline rate estimator, whose residuals are mean zero but do not form martingales. For continuous terminal event cumulative hazards, the induced influence functions for the restricted mean survival time and the recurrent event burden are
\begin{align*}
  \phi_{\nu,a,i}
  &=-\int_0^\tau\left\{\int_u^\tau S_a(t)\dd t\right\}\dd\phi_{\Lambda^D,a,i}(u),\\
  \phi_{\mu,a,i}
  &=\int_0^\tau S_a(t-)\dd\phi_{\Lambda^{R,\bw},a,i}(t)\\
  &\quad-\int_0^\tau\left\{\int_u^\tau S_a(t-)\dd\Lambda_a^{R,\bw}(t)\right\}\dd\phi_{\Lambda^D,a,i}(u).
\end{align*}
The first term in $\phi_{\mu,a,i}$ is the direct contribution from estimating the marginal weighted recurrent event rate, and the second term is the indirect contribution through the estimated marginal survival function.

Let \(\bullet\in\{C,D,R_1,\ldots,R_K\}\).  For \(\bullet=C\), define \(\widetilde Y_{C,i}(t)=Y_i^\dagger(t)\); for \(\bullet=D\) and \(\bullet=R_k\), define \(\widetilde Y_{\bullet,i}(t)=Y_i(t)\).  Let \(\theta_\bullet^\star\), \(\Lambda_{\bullet0}^\star\), and \(\bh_\bullet\) denote the probability-limit coefficient, baseline cumulative hazard or mean-rate, and covariate transform.  Define
\[
\begin{aligned}
  s_\bullet^{(r)}(t)
  &=\Ebb\left[\xi_i^a\widetilde Y_{\bullet,i}(t)\exp\{\theta_\bullet^{\star\tr}\bh_\bullet(\bZ_i)\}\bh_\bullet(\bZ_i)^{\otimes r}\right],\quad r=0,1,2,\\
  \bar\bh_\bullet(t)&=s_\bullet^{(1)}(t)/s_\bullet^{(0)}(t),\\
  V_\bullet(t)&=s_\bullet^{(2)}(t)/s_\bullet^{(0)}(t)-\bar\bh_\bullet(t)^{\otimes2},\\
  A_\bullet&=\int_0^\tau V_\bullet(t)s_\bullet^{(0)}(t)\dd\Lambda_{\bullet0}^\star(t).
\end{aligned}
\]
Let
\[
  \dd M_{\bullet i}^\star(t)=\dd N_{\bullet i}(t)-\widetilde Y_{\bullet,i}(t)\exp\{\theta_\bullet^{\star\tr}\bh_\bullet(\bZ_i)\}\dd\Lambda_{\bullet0}^\star(t).
\]
The individual score contribution and baseline contribution are
\[
\begin{aligned}
  U_{\bullet i}
  &=\int_0^\tau\xi_i^a\{\bh_\bullet(\bZ_i)-\bar\bh_\bullet(t)\}\dd M_{\bullet i}^\star(t),\\
  \phi_{\bullet0,i}(t)
  &=\int_0^t\frac{\xi_i^a\dd M_{\bullet i}^\star(u)}{s_\bullet^{(0)}(u)}
    -c_\bullet(t)\tr A_\bullet^{-1}U_{\bullet i},\\
  c_\bullet(t)&=\int_0^t\bar\bh_\bullet(u)\dd\Lambda_{\bullet0}^\star(u).
\end{aligned}
\]
A Taylor expansion of the score equation and the Breslow estimator gives
\[
\begin{aligned}
  \sqrt n(\widehat\theta_\bullet-\theta_\bullet^\star)
  &=A_\bullet^{-1}n^{-1/2}\sum_{i=1}^nU_{\bullet i}+o_p(1),\\
  \sqrt n\{\widehat\Lambda_{\bullet0}(t)-\Lambda_{\bullet0}^\star(t)\}
  &=n^{-1/2}\sum_{i=1}^n\phi_{\bullet0,i}(t)+o_p(1).
\end{aligned}
\]
For a target subject \(j\), define the induced cumulative hazard or rate \(\Lambda_{\bullet j}^\star(t)=\exp\{\theta_\bullet^{\star\tr}\bh_\bullet(\bZ_j)\}\Lambda_{\bullet0}^\star(t)\).  The effect of observation \(i\) on subject \(j\)'s fitted cumulative hazard or rate is
\[
\begin{aligned}
  \Gamma_{\bullet ji}(t)
  &=\exp\{\theta_\bullet^{\star\tr}\bh_\bullet(\bZ_j)\}\phi_{\bullet0,i}(t)\\
  &\quad+\Lambda_{\bullet0}^\star(t)\exp\{\theta_\bullet^{\star\tr}\bh_\bullet(\bZ_j)\}\bh_\bullet(\bZ_j)\tr A_\bullet^{-1}U_{\bullet i}.
\end{aligned}
\]
Consequently,
\[
\begin{aligned}
  \sqrt n\{\widehat K_j(t)-K_j^\star(t)\}
  &=-K_j^\star(t)n^{-1/2}\sum_{i=1}^n\Gamma_{Cji}(t)+o_p(1),\\
  \sqrt n\{\widehat H_j(t)-H_j^\star(t)\}
  &=-H_j^\star(t)n^{-1/2}\sum_{i=1}^n\Gamma_{Dji}(t)+o_p(1),\\
  \sqrt n\{\dd\widehat\Lambda_j^{R,\bw}(t)-\dd\Lambda_j^{R,\bw,\star}(t)\}
  &=n^{-1/2}\sum_{i=1}^n\dd\Gamma_{R,ji}^{\bw}(t)+o_p(1),
\end{aligned}
\]
where \(\dd\Gamma_{R,ji}^{\bw}(t)=\sum_{k=1}^K w_k\dd\Gamma_{R_k,ji}(t)\).

\subsection{Linearization of \texorpdfstring{\(q\), \(U\)}{q, U}, and the local ratio}
\label{supp:sec:local_ratio_derivative_detailed}

The fitted full data increments are \(\dd q_j^{D,\star}(t)=H_j^\star(t-)\dd\Lambda_j^{D,\star}(t)\) and \(\dd q_j^{R,\bw,\star}(t)=H_j^\star(t-)\dd\Lambda_j^{R,\bw,\star}(t)\).  The product rule gives the perturbations caused by observation \(i\):
\[
\begin{aligned}
  \dd\dot q_{D,ji}(t)
  &=H_j^\star(t-)\dd\Gamma_{Dji}(t)-H_j^\star(t-)\Gamma_{Dji}(t-)\dd\Lambda_j^{D,\star}(t),\\
  \dd\dot q_{R,ji}^{\bw}(t)
  &=H_j^\star(t-)\dd\Gamma_{R,ji}^{\bw}(t)-H_j^\star(t-)\Gamma_{Dji}(t-)\dd\Lambda_j^{R,\bw,\star}(t).
\end{aligned}
\]
To derive the perturbation of \(U_j\), write \(W_j(u)=\{K_j^\star(u-)H_j^\star(u-)\}^{-1}\).  Since \(\dot K_j(u-)=-K_j^\star(u-)\Gamma_{Cji}(u-)\) and \(\dot H_j(u-)=-H_j^\star(u-)\Gamma_{Dji}(u-)\),
\[
  \dot W_{ji}(u)=W_j(u)\{\Gamma_{Cji}(u-)+\Gamma_{Dji}(u-)\}.
\]
Also \(\dd\dot M_{C,ji}(u)=-Y_j^\dagger(u)\dd\Gamma_{Cji}(u)\).  Therefore
\begin{align*}
  \dot U_{ji}(t)
  &=-\int_{(0,t)}\dot W_{ji}(u)\dd M_{C,j}^\star(u)-\int_{(0,t)}W_j(u)\dd\dot M_{C,ji}(u) \\
  &=-\int_{(0,t)}\frac{\Gamma_{Cji}(u-)+\Gamma_{Dji}(u-)}{K_j^\star(u-)H_j^\star(u-)}\dd M_{C,j}^\star(u)
    +\int_{(0,t)}\frac{Y_j^\dagger(u)\dd\Gamma_{Cji}(u)}{K_j^\star(u-)H_j^\star(u-)}.
\end{align*}
For \(h\in\{D,R\}\), define the probability limit local numerator and denominator terms
\[
\begin{aligned}
  \dd A_{h,j}(t)
  &=\xi_j^a\{K_j^\star(t-)\}^{-1}\dd N_j^h(t)+\{1-\xi_j^aU_j^\star(t)\}\dd q_j^{h,\star}(t),\\
  B_j(t)
  &=\xi_j^a\{K_j^\star(t-)\}^{-1}Y_j(t)+\{1-\xi_j^aU_j^\star(t)\}r_j^\star(t).
\end{aligned}
\]
Let \(a_h(t)=\Ebb\{\dd A_{h,j}(t)\}\), \(b(t)=\Ebb\{B_j(t)\}\), and \(\dd\Lambda_h^\star(t)=a_h(t)/b(t)\).  Differentiating the three pieces in \(\dd A_{h,j}\) gives
\[
\begin{aligned}
  \dot A_{h,t}[i]
  &=\Ebb_j\Bigl[\xi_j^a\{K_j^\star(t-)\}^{-1}\Gamma_{Cji}(t-)\dd N_j^h(t)
       -\xi_j^a\dot U_{ji}(t)\dd q_j^{h,\star}(t)\\
  &\hspace{9em}+\{1-\xi_j^aU_j^\star(t)\}\dd\dot q_{h,ji}(t)\Bigr],\\
  \dot B_t[i]
  &=\Ebb_j\Bigl[\xi_j^a\{K_j^\star(t-)\}^{-1}\Gamma_{Cji}(t-)Y_j(t)
       -\xi_j^a\dot U_{ji}(t)r_j^\star(t)\\
  &\hspace{9em}-\{1-\xi_j^aU_j^\star(t)\}H_j^\star(t-)\Gamma_{Dji}(t-)\Bigr],
\end{aligned}
\]
where \(\Ebb_j\) means expectation over an independent copy \(j\), while \(i\) is the observation whose nuisance influence contribution is being propagated.  Thus the sample numerator and denominator satisfy
\[
\begin{aligned}
  \sqrt n\{\widehat a_h(t)-a_h(t)\}
  &=n^{-1/2}\sum_{i=1}^n\left[\dd A_{h,i}(t)-a_h(t)+\dot A_{h,t}[i]\right]+o_p(1),\\
  \sqrt n\{\widehat b(t)-b(t)\}
  &=n^{-1/2}\sum_{i=1}^n\left[B_i(t)-b(t)+\dot B_t[i]\right]+o_p(1).
\end{aligned}
\]
Set \(\zeta_{h,i}(t)=\dd A_{h,i}(t)-a_h(t)+\dot A_{h,t}[i]\) and \(\zeta_{B,i}(t)=B_i(t)-b(t)+\dot B_t[i]\).  Consecutive use of the quotient rule gives
\begin{align*}
  \sqrt n\{\dd\widehat\Lambda_h(t)-\dd\Lambda_h^\star(t)\}
  &=\sqrt n\left\{\frac{\widehat a_h(t)}{\widehat b(t)}-\frac{a_h(t)}{b(t)}\right\} \\
  &=\frac{\sqrt n\{\widehat a_h(t)-a_h(t)\}}{b(t)}-\frac{a_h(t)}{b(t)^2}\sqrt n\{\widehat b(t)-b(t)\}+o_p(1) \\
  &=n^{-1/2}\sum_{i=1}^n\frac{\zeta_{h,i}(t)-\dd\Lambda_h^\star(t)\zeta_{B,i}(t)}{b(t)}+o_p(1).
\end{align*}
Therefore the local influence process, including all first-order nuisance-estimation contributions, is
\[
  \phi_{\Lambda^h,i}(t)=\frac{\zeta_{h,i}(t)-\dd\Lambda_h^\star(t)\zeta_{B,i}(t)}{b(t)},
  \qquad h\in\{D,R\}.
\]
The empirical part \(\dd A_{h,i}-a_h\) and \(B_i-b\) is the variation that would remain if nuisances were known; the derivative terms \(\dot A_{h,t}[i]\) and \(\dot B_t[i]\) are precisely the additional variation from estimating \(K\), \(H\), \(q^D\), \(q^{R,\bw}\), and \(U\).

\subsection{From local increments to \texorpdfstring{\(\mu\), \(\nu\), \(\psi\)}{mu, nu, psi}, and the contrast}
\label{supp:sec:functional_delta_detailed}

Under the double-robust conditions in Theorem~\ref{supp:thm:irt_dr_detailed}, \(\dd\Lambda_D^\star=\dd\Lambda_a^D\) and \(\dd\Lambda_R^\star=\dd\Lambda_a^{R,\bw}\).  The product-integral derivative is
\[
  \phi_{S,i}(t)
  =-S_a(t)\int_{(0,t]}\frac{\dd\phi_{\Lambda^D,i}(u)}{1-\dd\Lambda_a^D(u)},
\]
with the simplification \(-S_a(t)\int_0^t\dd\phi_{\Lambda^D,i}(u)\) when \(\Lambda_a^D\) is continuous.  Integrating the survival influence function yields
\[
  \phi_{\nu,i}^a(\tau)=\int_0^\tau\phi_{S,i}(t)\dd t.
\]
For the recurrent event burden, the Stieltjes product rule gives
\[
  \phi_{\mu,i}^{a,\bw}(\tau)
  =\int_0^\tau S_a(t-)\dd\phi_{\Lambda^R,i}(t)+\int_0^\tau\phi_{S,i}(t-)\dd\Lambda_a^{R,\bw}(t).
\]
The first term is the direct effect of estimating the recurrent-event local rate, while the second term is the indirect effect of estimating survival.  The ratio map \((\mu,\nu)\mapsto\mu/\nu\) gives
\[
  \phi_{\psi,i}^{a,\bw}(\tau)=\frac{\phi_{\mu,i}^{a,\bw}(\tau)}{\nu_a(\tau)}-\frac{\mu_a(\tau;\bw)}{\nu_a(\tau)^2}\phi_{\nu,i}^a(\tau),
  \qquad
  \phi_{\Delta,i}^{\bw}(\tau)=\phi_{\psi,i}^{1,\bw}(\tau)-\phi_{\psi,i}^{0,\bw}(\tau).
\]
Consequently,
\[
  \sqrt n\{\widehat\Delta^{\mathrm{DR}}(\tau;\bw)-\Delta(\tau;\bw)\}
  =n^{-1/2}\sum_{i=1}^n\phi_{\Delta,i}^{\bw}(\tau)+o_p(1).
\]
The analytic influence function variance estimator is
\[
  \widehat V_{\mathrm{IF}}\{\widehat\Delta^{\mathrm{DR}}(\tau;\bw)\}
  =\frac{1}{n(n-1)}\sum_{i=1}^n\left\{\widehat\phi_{\Delta,i}^{\bw}(\tau)-n^{-1}\sum_{r=1}^n\widehat\phi_{\Delta,r}^{\bw}(\tau)\right\}^2.
\]
Every term in \(\widehat\phi_{\Delta,i}^{\bw}\) is obtained by replacing the population quantities in the preceding displays with their fitted or empirical analogues.  This variance therefore includes the variation from the augmented local estimating equations and the variation transmitted through the fitted Cox and LWYY nuisance estimators.



\section{Efficient influence function and efficiency gap in IRTs}
\label{supp:sec:eif}

This section distinguishes two influence function objects.  The first is the semiparametric efficient influence function for the unrestricted observed-data model with known treatment randomization and covariate-dependent right censoring.  The second is the first-order influence function induced by the proposed local Nelson--Aalen estimator when the local nuisance functions are correctly specified.  These two objects coincide only under an additional history-sufficiency condition.  Therefore, the influence function used for variance estimation in the main text should be interpreted as the influence function of the proposed local estimator; it is not, in general, the efficient influence function.

For a fixed treatment arm \(a\), let
\[
  \xi_i^a=\frac{\ind(A_i=a)}{\pi_a},
  \qquad
  K_{0i}^a(t)=\Pbb(C_i^a\ge t\mid\bZ_i),
\]
and let \(\dd M_i^C(t;a)=\dd N_i^C(t)-Y_i^\dagger(t)\dd\Lambda_{C0}^a(t\mid\bZ_i)\) denote the true censoring martingale increment in arm \(a\).  Define the two full data cumulative outcomes
\[
  G_{\mu,i}^a(t)=N_i^{a,\bw}(t),
  \qquad
  G_{\nu,i}^a(t)=\int_0^tY_i^a(s)\dd s.
\]
Thus
\[
  \theta_{\mu,a}=\Ebb\{G_{\mu,i}^a(\tau)\}=\mu_a(\tau;\bw),
  \qquad
  \theta_{\nu,a}=\Ebb\{G_{\nu,i}^a(\tau)\}=\nu_a(\tau).
\]
For \(\ell\in\{\mu,\nu\}\), write
\[
  m_{\ell,a}(\bZ_i)=\Ebb\{G_{\ell,i}^a(\tau)\mid\bZ_i\}.
\]
The observed increments are
\[
  \dd G_{\mu,i}^{\mathrm{obs}}(t)=\dd N_i^{\bw}(t),
  \qquad
  \dd G_{\nu,i}^{\mathrm{obs}}(t)=Y_i(t)\dd t.
\]
Here \(Y_i(t)=\ind(X_i\ge t)\) and \(N_i^{\bw}(t)\) are observed only up to censoring.

If the full data \(\{G_{\ell,i}^a(\tau):a=0,1\}\) were observed, the efficient influence function for \(\theta_{\ell,a}\) would be
\[
  D_{\ell,a}^{\mathrm{full}}=m_{\ell,a}(\bZ_i)-\theta_{\ell,a}+\xi_i^a\{G_{\ell,i}^a(\tau)-m_{\ell,a}(\bZ_i)\}.
\]
To verify this expression, let \(S\) be the full data score along an arbitrary regular parametric submodel.  Decompose the score into the baseline-covariate score and the conditional outcome-history score.  Since the randomization probability is known by design, the treatment law contributes no score.  Then
\begin{align*}
  \dot\theta_{\ell,a}
  &=\frac{\partial}{\partial\varepsilon}\Ebb_\varepsilon\{G_{\ell,i}^a(\tau)\}\bigg|_{\varepsilon=0} \\
  &=\Ebb\!\left[\{m_{\ell,a}(\bZ_i)-\theta_{\ell,a}\}S_{\bZ}(\bZ_i)\right]+\Ebb\!\left[ \Ebb\{G_{\ell,i}^a(\tau)-m_{\ell,a}(\bZ_i)\mid A_i=a,\bZ_i\}S_{G\mid A,\bZ}\right] \\
  &=\Ebb\!\left[\{m_{\ell,a}(\bZ_i)-\theta_{\ell,a}\}S_{\bZ}(\bZ_i)\right]+\Ebb\!\left[\frac{\ind(A_i=a)}{\pi_a}\{G_{\ell,i}^a(\tau)-m_{\ell,a}(\bZ_i)\}S_{G\mid A,\bZ}\right] \\
  &=\Ebb\{D_{\ell,a}^{\mathrm{full}}S\}.
\end{align*}
Moreover,
\[
  \Ebb\{D_{\ell,a}^{\mathrm{full}}\}=\Ebb\{m_{\ell,a}(\bZ_i)-\theta_{\ell,a}\}+\Ebb\!\left[\Ebb\{\xi_i^a(G_{\ell,i}^a(\tau)-m_{\ell,a}(\bZ_i))\mid\bZ_i\}\right]=0.
\]
We next project the full data gradient through right censoring.  Let
\[
  \mathcal F_i^a(u-)=\sigma\!\left[\bZ_i,\{N_i^{D,a}(s),N_i^a(s):0\le s<u\}\right]
\]
be the full event history immediately before \(u\).  Define the remaining outcome from just before \(u\) to \(\tau\) by
\[
  R_{\ell,i}^a(u,\tau)=G_{\ell,i}^a(\tau)-G_{\ell,i}^a(u-),
\]
where
\[
  R_{\mu,i}^a(u,\tau)=N_i^{a,\bw}(\tau)-N_i^{a,\bw}(u-),
  \qquad
  R_{\nu,i}^a(u,\tau)=\int_u^\tau Y_i^a(t)\dd t.
\]
The history specific remaining outcome regression is
\[
  \rho_{\ell,a}^{\mathrm{EIF}}\{u,\tau,\mathcal F_i^a(u-)\}=\Ebb\!\left[ R_{\ell,i}^a(u,\tau)\mid\mathcal F_i^a(u-),Y_i^a(u)=1\right].
\]
By convention, \(\rho_{\ell,a}^{\mathrm{EIF}}\{u,\tau,\mathcal F_i^a(u-)\}=0\) after death.

The fundamental censoring identity is
\[
  \int_{(0,\tau]}\frac{\dd G_{\ell,i}^{\mathrm{obs}}(t)}{K_{0i}^a(t-)}=G_{\ell,i}^a(\tau)-\int_{(0,\tau]}\frac{R_{\ell,i}^a(u,\tau)}{K_{0i}^a(u-)}\dd M_i^C(u;a)
  \qquad \text{on } A_i=a.
\]
Indeed, with \(R_i^C(t)=\ind(C_i^a\ge t)\), the product-integral identity gives
\[
  \dd\!\left\{\frac{R_i^C(t)}{K_{0i}^a(t)}\right\}=-\frac{\dd M_i^C(t;a)}{K_{0i}^a(t-)}.
\]
Since the observed increment equals \(R_i^C(t-)\dd G_{\ell,i}^a(t)\),
\begin{align*}
  \int_{(0,\tau]}\frac{\dd G_{\ell,i}^{\mathrm{obs}}(t)}{K_{0i}^a(t-)}&=\int_{(0,\tau]}\frac{R_i^C(t-)}{K_{0i}^a(t-)}\dd G_{\ell,i}^a(t) \\
  &=\int_{(0,\tau]}\dd G_{\ell,i}^a(t)-\int_{(0,\tau]}\left\{\int_{(u,\tau]}\dd G_{\ell,i}^a(t)\right\}\frac{\dd M_i^C(u;a)}{K_{0i}^a(u-)} \\
  &=G_{\ell,i}^a(\tau)-\int_{(0,\tau]}\frac{R_{\ell,i}^a(u,\tau)}{K_{0i}^a(u-)}\dd M_i^C(u;a).
\end{align*}

The observed-data efficient influence function for \(\theta_{\ell,a}\), \(\ell\in\{\mu,\nu\}\), is
\[
  D_{\ell,a}^{\mathrm{EIF}}=m_{\ell,a}(\bZ_i)-\theta_{\ell,a}+\xi_i^a\left[\int_{(0,\tau]}\frac{\dd G_{\ell,i}^{\mathrm{obs}}(t)}{K_{0i}^a(t-)}-m_{\ell,a}(\bZ_i)+\int_{(0,\tau]}\frac{\rho_{\ell,a}^{\mathrm{EIF}}\{u,\tau,\mathcal F_i^a(u-)\}}{K_{0i}^a(u-)}\dd M_i^C(u;a)\right].
\]
Equivalently, by substituting the preceding censoring identity,
\begin{align*}
  D_{\ell,a}^{\mathrm{EIF}}&=m_{\ell,a}(\bZ_i)-\theta_{\ell,a}+\xi_i^a\left[G_{\ell,i}^a(\tau)-\int_{(0,\tau]}\frac{R_{\ell,i}^a(u,\tau)}{K_{0i}^a(u-)}\dd M_i^C(u;a)-m_{\ell,a}(\bZ_i)\right.\\
  &\hspace{35mm}\left.+\int_{(0,\tau]}\frac{\rho_{\ell,a}^{\mathrm{EIF}}\{u,\tau,\mathcal F_i^a(u-)\}}{K_{0i}^a(u-)}\dd M_i^C(u;a)\right] \\
  &=D_{\ell,a}^{\mathrm{full}}-\xi_i^a\int_{(0,\tau]}\frac{R_{\ell,i}^a(u,\tau)-\rho_{\ell,a}^{\mathrm{EIF}}\{u,\tau,\mathcal F_i^a(u-)\}}{K_{0i}^a(u-)}\dd M_i^C(u;a).
\end{align*}
This display shows both why the observed-data influence function represents the same outcome-law pathwise derivative as the full data influence function and why it is orthogonal to the censoring nuisance tangent space.  For any censoring score of the form
\[
  S_C(c)=\int_{(0,\tau]}c(u,\bZ_i)\dd M_i^C(u;a),
\]
martingale covariance and iterated expectation give
\begin{align*}
  \Ebb\{D_{\ell,a}^{\mathrm{EIF}}S_C(c)\}&=-\Ebb\!\left[\xi_i^a\int_{(0,\tau]}\frac{R_{\ell,i}^a(u,\tau)-\rho_{\ell,a}^{\mathrm{EIF}}\{u,\tau,\mathcal F_i^a(u-)\}}{K_{0i}^a(u-)}c(u,\bZ_i)\dd\langle M_i^C\rangle(u;a)\right] \\
  &=-\Ebb\!\left[ \xi_i^a\int_{(0,\tau]}\frac{ \Ebb[R_{\ell,i}^a(u,\tau)-\rho_{\ell,a}^{\mathrm{EIF}}\{u,\tau,\mathcal F_i^a(u-)\}\mid\mathcal F_i^a(u-),Y_i^a(u)=1]}{K_{0i}^a(u-)}c(u,\bZ_i)\dd\langle M_i^C\rangle(u;a)\right] \\
  &=0.
\end{align*}
Thus \(D_{\ell,a}^{\mathrm{EIF}}\) is the canonical gradient for \(\theta_{\ell,a}\) in the unrestricted observed-data model.

For the numerator and denominator specifically,
\[
  D_{\mu,a}^{\mathrm{EIF}}=m_{\mu,a}(\bZ_i)-\mu_a(\tau;\bw)+\xi_i^a\left[\int_{(0,\tau]}\frac{\dd N_i^{\bw}(t)}{K_{0i}^a(t-)}-m_{\mu,a}(\bZ_i)+\int_{(0,\tau]}\frac{\rho_{\mu,a}^{\mathrm{EIF}}\{u,\tau,\mathcal F_i^a(u-)\}}{K_{0i}^a(u-)}\dd M_i^C(u;a)\right],
\]
where
\[
  m_{\mu,a}(\bZ_i)=\Ebb\{N_i^{a,\bw}(\tau)\mid\bZ_i\},
\]
and
\[
  \rho_{\mu,a}^{\mathrm{EIF}}\{u,\tau,\mathcal F_i^a(u-)\}=\Ebb\{N_i^{a,\bw}(\tau)-N_i^{a,\bw}(u-)\mid\mathcal F_i^a(u-),Y_i^a(u)=1\}.
\]
Similarly,
\[
  D_{\nu,a}^{\mathrm{EIF}}=m_{\nu,a}(\bZ_i)-\nu_a(\tau)+\xi_i^a\left[\int_0^\tau\frac{Y_i(t)}{K_{0i}^a(t-)}\dd t-m_{\nu,a}(\bZ_i)+\int_{(0,\tau]}\frac{\rho_{\nu,a}^{\mathrm{EIF}}\{u,\tau,\mathcal F_i^a(u-)\}}{K_{0i}^a(u-)}\dd M_i^C(u;a)\right],
\]
where
\[
  m_{\nu,a}(\bZ_i)=\Ebb(D_i^a\wedge\tau\mid\bZ_i),
\]
and
\[
  \rho_{\nu,a}^{\mathrm{EIF}}\{u,\tau,\mathcal F_i^a(u-)\}=\Ebb\left\{\int_u^\tau Y_i^a(t)\dd t\mid\mathcal F_i^a(u-),Y_i^a(u)=1\right\}.
\]

The while-alive estimand is the smooth ratio \(\psi_a(\tau;\bw)=\mu_a(\tau;\bw)/\nu_a(\tau)\).  Therefore, when \(\nu_a(\tau)>0\),
\[
  D_{\psi,a}^{\mathrm{EIF}}=\frac{D_{\mu,a}^{\mathrm{EIF}}-\psi_a(\tau;\bw)D_{\nu,a}^{\mathrm{EIF}}}{\nu_a(\tau)}.
\]
The efficient influence function for the treatment contrast is
\[
  D_{\Delta}^{\mathrm{EIF}}=D_{\psi,1}^{\mathrm{EIF}}-D_{\psi,0}^{\mathrm{EIF}}.
\]

We now compare this EIF with the first-order influence function induced by the proposed local Nelson--Aalen estimator.  The proposed estimator uses conditional terminal survival and conditional recurrent-event mean increments given baseline covariates, not the entire recurrent-event history before censoring.  Let
\[
  H_{0i}^a(t\mid\bZ_i)=\Pbb(D_i^a\ge t\mid\bZ_i),
\]
and
\[
  \dd q_{0i}^{R,\bw}(t;a)=\Ebb\{\dd N_i^{a,\bw}(t)\mid\bZ_i\}.
\]
The baseline-and-alive remaining-outcome regressions induced by the local representation are
\[
  \rho_{\nu,a}^{\mathrm{loc}}(u,\tau,\bZ_i)=\frac{\int_u^\tau H_{0i}^a(t\mid\bZ_i)\dd t}{H_{0i}^a(u-\mid\bZ_i)}
\]
and
\[
  \rho_{\mu,a}^{\mathrm{loc}}(u,\tau,\bZ_i)=\frac{\int_{(u,\tau]}\dd q_{0i}^{R,\bw}(t;a)}{H_{0i}^a(u-\mid\bZ_i)}.
\]
At the intersection where the censoring, terminal-event, and recurrent-event nuisance functions are all correctly specified, the first-order influence function associated with the local estimator has the same observed-data form as above but with \(\rho_{\ell,a}^{\mathrm{EIF}}\{u,\tau,\mathcal F_i^a(u-)\}\) replaced by \(\rho_{\ell,a}^{\mathrm{loc}}(u,\tau,\bZ_i)\):
\[
  D_{\ell,a}^{\mathrm{loc}}=m_{\ell,a}(\bZ_i)-\theta_{\ell,a}+\xi_i^a\left[\int_{(0,\tau]}\frac{\dd G_{\ell,i}^{\mathrm{obs}}(t)}{K_{0i}^a(t-)} -m_{\ell,a}(\bZ_i)+\int_{(0,\tau]}\frac{\rho_{\ell,a}^{\mathrm{loc}}(u,\tau,\bZ_i)}{K_{0i}^a(u-)}\dd M_i^C(u;a)\right].
\]
Consequently,
\[
  D_{\ell,a}^{\mathrm{loc}}-D_{\ell,a}^{\mathrm{EIF}}
  =\xi_i^a\int_{(0,\tau]}\frac{\rho_{\ell,a}^{\mathrm{loc}}(u,\tau,\bZ_i)-\rho_{\ell,a}^{\mathrm{EIF}}\{u,\tau,\mathcal F_i^a(u-)\}}{K_{0i}^a(u-)}
  \dd M_i^C(u;a).
\]
This difference is generally nonzero.  It vanishes only if the full recurrent-event history carries no additional prognostic information, beyond baseline covariates and being alive at \(u\), about the remaining outcome:
\[
  \rho_{\mu,a}^{\mathrm{EIF}}\{u,\tau,\mathcal F_i^a(u-)\}=\rho_{\mu,a}^{\mathrm{loc}}(u,\tau,\bZ_i)
\]
and
\[
  \rho_{\nu,a}^{\mathrm{EIF}}\{u,\tau,\mathcal F_i^a(u-)\}=\rho_{\nu,a}^{\mathrm{loc}}(u,\tau,\bZ_i)
\]
for almost every \(u\in(0,\tau]\).  This is a history sufficiency condition.  It can fail, for example, when prior recurrent events reveal latent frailty that also predicts future recurrent events or death.

The variance gap follows from orthogonality of the canonical gradient and the censoring-martingale residual above:
\[
  \Var(D_{\ell,a}^{\mathrm{loc}})-\Var(D_{\ell,a}^{\mathrm{EIF}})=\Var(D_{\ell,a}^{\mathrm{loc}}-D_{\ell,a}^{\mathrm{EIF}}).
\]
Using the martingale isometry,
\[
  \Var(D_{\ell,a}^{\mathrm{loc}})-\Var(D_{\ell,a}^{\mathrm{EIF}})=\Ebb\!\left[(\xi_i^a)^2\int_{(0,\tau]}\frac{\left[\rho_{\ell,a}^{\mathrm{loc}}(u,\tau,\bZ_i)-\rho_{\ell,a}^{\mathrm{EIF}}\{u,\tau,\mathcal F_i^a(u-)\}\right]^2}{\{K_{0i}^a(u-)\}^2}\dd\langle M_i^C\rangle(u;a)\right]\ge 0.
\]
For the ratio,
\[
  D_{\psi,a}^{\mathrm{loc}}=\frac{D_{\mu,a}^{\mathrm{loc}}-\psi_a(\tau;\bw)D_{\nu,a}^{\mathrm{loc}}}{\nu_a(\tau)}.
\]
Thus
\begin{align*}
  D_{\psi,a}^{\mathrm{loc}}-D_{\psi,a}^{\mathrm{EIF}}
  &=\frac{\xi_i^a}{\nu_a(\tau)}\int_{(0,\tau]}\frac{1}{K_{0i}^a(u-)}
  \Bigl[\rho_{\mu,a}^{\mathrm{loc}}(u,\tau,\bZ_i)
  -\rho_{\mu,a}^{\mathrm{EIF}}\{u,\tau,\mathcal F_i^a(u-)\}\\
  &\hspace{37mm}
  -\psi_a(\tau;\bw)\Bigl\{\rho_{\nu,a}^{\mathrm{loc}}(u,\tau,\bZ_i)
  -\rho_{\nu,a}^{\mathrm{EIF}}\{u,\tau,\mathcal F_i^a(u-)\}\Bigr\}\Bigr]
  \dd M_i^C(u;a).
\end{align*}
Componentwise history sufficiency is sufficient for efficiency of the proposed local estimator.  A weaker ratio-specific condition is
\[
  \rho_{\mu,a}^{\mathrm{loc}}(u,\tau,\bZ_i)-\rho_{\mu,a}^{\mathrm{EIF}}\{u,\tau,\mathcal F_i^a(u-)\}=\psi_a(\tau;\bw)\left[\rho_{\nu,a}^{\mathrm{loc}}(u,\tau,\bZ_i)-\rho_{\nu,a}^{\mathrm{EIF}}\{u,\tau,\mathcal F_i^a(u-)\}\right],\]
but this condition is not implied by the assumptions in the main text.

Therefore, the proposed local Nelson--Aalen estimator is doubly robust and asymptotically normal under the conditions stated in the main text, but it is not generally semiparametrically efficient in the unrestricted recurrent-event model.  Its influence function equals the efficient influence function only under the history-sufficiency condition above.  The variance estimator in the main text estimates the asymptotic variance of the proposed local estimator; it should be described as the semiparametric efficiency bound only when the history-sufficiency condition is imposed.
\section{First-order equivalence with the future-mean estimator in \texorpdfstring{\citet{baer2025causal}}{Baer et al. (2025)}}
\label{supp:sec:baer_detailed_equivalence}

\subsection{Restricted local representation}

The proposed estimator is built from local increments \(\dd q_a^{R,\bw}(t\mid\bZ)\) and \(H_a(t\mid\bZ)\), whereas the estimator of \citet{baer2025causal} uses the conditional future mean \(F_a(u,t\mid\bZ)=\Ebb\{\ind(D^a>u)N^{a,\bw}(t)\mid\bZ\}\).  The two representations are first-order comparable when the future mean is induced by the same local law.  Specifically, assume that for \(0<v\le u\le t\le\tau\),
\[
  \Ebb\{\ind(D^a>u)\dd N^{a,\bw}(v)\mid\bZ\}
  =\frac{H_a(u\mid\bZ)}{H_a(v-\mid\bZ)}\dd q_a^{R,\bw}(v\mid\bZ),
\]
with \(0/0=0\).  Define
\[
  \Gamma_a(u,t\mid\bZ)
  =H_a(u\mid\bZ)\int_0^u\frac{\dd q_a^{R,\bw}(v\mid\bZ)}{H_a(v-\mid\bZ)}+\int_u^t\dd q_a^{R,\bw}(v\mid\bZ).
\]
Then \(F_a(u,t\mid\bZ)=\Gamma_a(u,t\mid\bZ)\).  The proof is a direct calculation:
\begin{align*}
  F_a(u,t\mid\bZ)
  &=\Ebb\left\{\ind(D^a>u)\int_{(0,t]}\dd N^{a,\bw}(v)\mid\bZ\right\} \\
  &=\int_{(0,u]}\Ebb\{\ind(D^a>u)\dd N^{a,\bw}(v)\mid\bZ\}+
    \int_{(u,t]}\Ebb\{\ind(D^a>u)\dd N^{a,\bw}(v)\mid\bZ\} \\
  &=H_a(u\mid\bZ)\int_0^u\frac{\dd q_a^{R,\bw}(v\mid\bZ)}{H_a(v-\mid\bZ)}+
    \int_u^t\Ebb\{\dd N^{a,\bw}(v)\mid\bZ\} \\
  &=H_a(u\mid\bZ)\int_0^u\frac{\dd q_a^{R,\bw}(v\mid\bZ)}{H_a(v-\mid\bZ)}+
    \int_u^t\dd q_a^{R,\bw}(v\mid\bZ) \\
  &=\Gamma_a(u,t\mid\bZ),
\end{align*}
where the third equality uses the fact that \(\dd N^{a,\bw}(v)\ne0\) for \(v>u\) implies \(D^a\ge v>u\), so \(\ind(D^a>u)\dd N^{a,\bw}(v)=\dd N^{a,\bw}(v)\).

If \(\dot H_a\) and \(\dd\dot q_a^{R,\bw}\) are first-order perturbations, then the directional derivative of \(\Gamma_a\) is
\[
\begin{aligned}
  \dot\Gamma_a(u,t\mid\bZ)
  &=\dot H_a(u\mid\bZ)\int_0^u\frac{\dd q_a^{R,\bw}(v\mid\bZ)}{H_a(v-\mid\bZ)}
  -H_a(u\mid\bZ)\int_0^u\frac{\dot H_a(v-\mid\bZ)}{H_a(v-\mid\bZ)^2}\dd q_a^{R,\bw}(v\mid\bZ)\\
  &\quad+H_a(u\mid\bZ)\int_0^u\frac{\dd\dot q_a^{R,\bw}(v\mid\bZ)}{H_a(v-\mid\bZ)}+
  \int_u^t\dd\dot q_a^{R,\bw}(v\mid\bZ).
\end{aligned}
\]
The derivative is obtained as a literal first-order expansion.  Let \(H_{a,\epsilon}=H_a+\epsilon\dot H_a\) and \(\dd q_{a,\epsilon}^{R,\bw}=\dd q_a^{R,\bw}+\epsilon\dd\dot q_a^{R,\bw}\).  With
\[
  J_{a,\epsilon}(u\mid\bZ)=\int_0^u\frac{\dd q_{a,\epsilon}^{R,\bw}(v\mid\bZ)}{H_{a,\epsilon}(v-\mid\bZ)},
\]
one has, consecutively,
\[
\begin{aligned}
  \left.\frac{\partial}{\partial\epsilon}J_{a,\epsilon}(u\mid\bZ)\right|_{\epsilon=0}
  &=\left.\frac{\partial}{\partial\epsilon}\int_0^u
      \frac{\dd q_a^{R,\bw}(v\mid\bZ)+\epsilon\dd\dot q_a^{R,\bw}(v\mid\bZ)}{H_a(v-\mid\bZ)+\epsilon\dot H_a(v-\mid\bZ)}\right|_{\epsilon=0} \\
  &=\int_0^u\left\{\frac{\dd\dot q_a^{R,\bw}(v\mid\bZ)}{H_a(v-\mid\bZ)}
      -\frac{\dot H_a(v-\mid\bZ)}{H_a(v-\mid\bZ)^2}\dd q_a^{R,\bw}(v\mid\bZ)\right\},\\
  \left.\frac{\partial}{\partial\epsilon}\{H_{a,\epsilon}(u\mid\bZ)J_{a,\epsilon}(u\mid\bZ)\}\right|_{\epsilon=0}
  &=\dot H_a(u\mid\bZ)J_a(u\mid\bZ)+H_a(u\mid\bZ)
    \left.\frac{\partial}{\partial\epsilon}J_{a,\epsilon}(u\mid\bZ)\right|_{\epsilon=0},\\
  \left.\frac{\partial}{\partial\epsilon}\int_u^t\dd q_{a,\epsilon}^{R,\bw}(v\mid\bZ)\right|_{\epsilon=0}
  &=\int_u^t\dd\dot q_a^{R,\bw}(v\mid\bZ).
\end{aligned}
\]
Adding the last two displays gives the stated \(\dot\Gamma_a(u,t\mid\bZ)\).

\subsection{Denominator equivalence}

Let \(\widehat S_a^{\mathrm B}(t)\) denote the one-step survival estimator of \citet{baer2025causal} evaluated with the same fitted \(\widehat H_a\) and \(\widehat K_a\) used in the proposed estimator.  Its pointwise estimating function is
\[
\begin{aligned}
  \Phi_{\eta,a}(t;\bO;H,K)
  &=\xi^a\frac{\delta\ind(X>t)}{K(X-\mid\bZ)}-(\xi^a-1)H(t\mid\bZ)\\
  &\quad+\xi^a\int_{(0,\infty)}\frac{H(t\vee u\mid\bZ)}{H(u\mid\bZ)K(u\mid\bZ)}\dd M_C^{a,H,K}(u),
\end{aligned}
\]
where \(\dd M_C^{a,H,K}(u)=\dd N^C(u)-Y^\dagger(u)\dd\Lambda_C^{a,K}(u\mid\bZ)\).  A von Mises expansion at \((H_a,K_a)\) yields
\begin{align*}
  \widehat S_a^{\mathrm B}(t)-S_a(t)
  &=(\Pn-P_0)\{\Phi_{\eta,a}(t;\bO;H_a,K_a)\} \\
  &\quad+P_0\{\Phi_{\eta,a}(t;\bO;\widehat H_a,\widehat K_a)-\Phi_{\eta,a}(t;\bO;H_a,K_a)\}+o_p(n^{-1/2}).
\end{align*}
The first term is the empirical-process term.  The second term is the common linear image of \(\widehat H_a-H_a\) and \(\widehat K_a-K_a\).  Since the proposed survival estimator and the Baer survival estimator use the same \(\widehat H_a\) and \(\widehat K_a\), they have the same nuisance perturbation.  Denoting the common pointwise influence function by \(\varphi_{\eta,a}(t;\bO)\),
\begin{align*}
  \sqrt n\{\widehat S_a^{\mathrm B}(t)-S_a(t)\}
  &=\Gn\varphi_{\eta,a}(t;\bO)+o_p(1),\\
  \sqrt n\{\widehat S_a^{\mathrm{DR}}(t)-S_a(t)\}
  &=\Gn\varphi_{\eta,a}(t;\bO)+o_p(1).
\end{align*}
Integrating over \([0,\tau]\),
\begin{align*}
  \sqrt n\{\widehat\nu_a^{\mathrm B}(\tau)-\nu_a(\tau)\}
  &=\int_0^\tau \sqrt n\{\widehat S_a^{\mathrm B}(t)-S_a(t)\}\dd t \\
  &=\int_0^\tau \Gn\varphi_{\eta,a}(t;\bO)\dd t+o_p(1) \\
  &=\Gn\left\{\int_0^\tau\varphi_{\eta,a}(t;\bO)\dd t\right\}+o_p(1),
\end{align*}
and the proposed estimator has the same expansion.  Therefore \(\widehat\nu_a^{\mathrm{DR}}(\tau)-\widehat\nu_a^{\mathrm B}(\tau)=o_p(n^{-1/2})\).

\subsection{Numerator equivalence}

The Baer numerator one-step map can be written as
\[
\begin{aligned}
  \Psi_{\mu,a}(F,H,K;\bO)
  &=\phi_{\mu,a}^{\mathrm{obs}}(\tau;\bw;\bO)-(
  \xi^a-1)F(0,\tau\mid\bZ)\\
  &\quad+\xi^a\int_{(0,\tau]}\frac{F(u,\tau\mid\bZ)}{H(u\mid\bZ)K(u\mid\bZ)}\dd M_C^{a,K}(u),
\end{aligned}
\]
where \(\phi_{\mu,a}^{\mathrm{obs}}\) is the IPCW observed-count term and \(\dd M_C^{a,K}(u)=\dd N^C(u)-Y^\dagger(u)\dd\Lambda_C^{a,K}(u\mid\bZ)\).  Under the restricted local representation, \(F=\Gamma(H,q)\); hence the proposed local estimator is the same first-order map applied to \(\widetilde\Psi_{\mu,a}(H,K,q;\bO)=\Psi_{\mu,a}\{\Gamma(H,q),H,K;\bO\}\).

For an arbitrary perturbation \((\dot F,\dot H,\dot K)\), the derivative of \(\Psi_{\mu,a}\) is
\[
\begin{aligned}
  \dot\Psi_{\mu,a}[\dot F,\dot H,\dot K](\bO)
  &=-\xi^a\frac{\delta N^{\bw}(\tau)}{K_a(X-\mid\bZ)^2}\dot K_a(X-\mid\bZ)-(
  \xi^a-1)\dot F_a(0,\tau\mid\bZ)\\
  &\quad+\xi^a\int_{(0,\tau]}\frac{\dot F_a(u,\tau\mid\bZ)}{H_a(u\mid\bZ)K_a(u\mid\bZ)}\dd M_C^a(u)\\
  &\quad-\xi^a\int_{(0,\tau]}\frac{F_a(u,\tau\mid\bZ)\dot H_a(u\mid\bZ)}{H_a(u\mid\bZ)^2K_a(u\mid\bZ)}\dd M_C^a(u)\\
  &\quad-\xi^a\int_{(0,\tau]}\frac{F_a(u,\tau\mid\bZ)\dot K_a(u\mid\bZ)}{H_a(u\mid\bZ)K_a(u\mid\bZ)^2}\dd M_C^a(u)\\
  &\quad+\xi^a\int_{(0,\tau]}\frac{F_a(u,\tau\mid\bZ)}{H_a(u\mid\bZ)K_a(u\mid\bZ)}\dd\dot M_C^a(u),
\end{aligned}
\]
with \(\dd\dot M_C^a(u)=-Y^\dagger(u)\dd\dot\Lambda_C^a(u\mid\bZ)\).  Because \(F_a=\Gamma_a(H_a,q_a)\), the perturbation \(\dot F_a\) equals \(\dot\Gamma_a[\dot H_a,\dd\dot q_a^{R,\bw}]\).  Therefore the derivative of the composed local map is exactly the chain-rule derivative
\[
  \dot{\widetilde\Psi}_{\mu,a}[\dot H_a,
  \dot K_a,\dd\dot q_a^{R,\bw}](\bO)
  =\dot\Psi_{\mu,a}[\dot\Gamma_a,
  \dot H_a,
  \dot K_a](\bO).
\]
Now expand the estimator of \citet{baer2025causal}:
\begin{align*}
  \widehat\mu_a^{\mathrm B}(\tau;\bw)-\mu_a(\tau;\bw)
  &=(\Pn-P_0)\{\Psi_{\mu,a}(F_a,H_a,K_a;\bO)\}\\
  &\quad+P_0\{\Psi_{\mu,a}(\widehat F_a,\widehat H_a,\widehat K_a;\bO)-\Psi_{\mu,a}(F_a,H_a,K_a;\bO)\}+o_p(n^{-1/2})\\
  &=n^{-1/2}\Gn\{\Psi_{\mu,a}(F_a,H_a,K_a;\bO)-\mu_a(\tau;\bw)\}\\
  &\quad+P_0\{\dot\Psi_{\mu,a}[\widehat F_a-F_a,\widehat H_a-H_a,\widehat K_a-K_a](\bO)\}+o_p(n^{-1/2}).
\end{align*}
Substituting \(\widehat F_a=\widehat\Gamma_a\) and the derivative of \(\Gamma_a\) gives exactly the same expansion as the proposed local estimator based on \(\widetilde\Psi_{\mu,a}\):
\begin{align*}
  \sqrt n\{\widehat\mu_a^{\mathrm B}(\tau;\bw)-\mu_a(\tau;\bw)\}
  &=\Gn\varphi_{\mu,a}(\tau;\bw;\bO)+o_p(1),\\
  \sqrt n\{\widehat\mu_a^{\mathrm{DR}}(\tau;\bw)-\mu_a(\tau;\bw)\}
  &=\Gn\varphi_{\mu,a}(\tau;\bw;\bO)+o_p(1).
\end{align*}
Thus \(\widehat\mu_a^{\mathrm{DR}}(\tau;\bw)-\widehat\mu_a^{\mathrm B}(\tau;\bw)=o_p(n^{-1/2})\).  Combining numerator and denominator equivalence with the smooth ratio map gives
\[
  \widehat\psi_a^{\mathrm{DR}}(\tau;\bw)-\widehat\psi_a^{\mathrm B}(\tau;\bw)=o_p(n^{-1/2}),
  \qquad
  \widehat\Delta^{\mathrm{DR}}(\tau;\bw)-\widehat\Delta^{\mathrm B}(\tau;\bw)=o_p(n^{-1/2}).
\]
This is a first-order equivalence.  It does not imply finite-sample equality when \(F\) is fitted directly by a two-time regression unrelated to the local law, or when the local product-integral update is numerically constrained while the future-mean estimator is not.

\section{Double robustness in CRTs}
\label{supp:sec:crt_dr_detailed}

\subsection{Weighted full data representation}

Let \(\ell\in\{\mathrm{clus},\mathrm{ind}\}\), with \(\omega_i^{\mathrm{clus}}=1\) and \(\omega_i^{\mathrm{ind}}=n_i\).  The target-specific local increments are
\[
\begin{aligned}
  \dd\Lambda_{a,\ell}^D(t)
  &=\frac{\Ebb\left[(\omega_i^\ell/n_i)\sum_{j=1}^{n_i}\dd N_{ij}^{D,a}(t)\right]}{\Ebb\left[(\omega_i^\ell/n_i)\sum_{j=1}^{n_i}Y_{ij}^a(t)\right]},\\
  \dd\Lambda_{a,\ell}^{R,\bw}(t)
  &=\frac{\Ebb\left[(\omega_i^\ell/n_i)\sum_{j=1}^{n_i}\dd N_{ij}^{a,\bw}(t)\right]}{\Ebb\left[(\omega_i^\ell/n_i)\sum_{j=1}^{n_i}Y_{ij}^a(t)\right]}.
\end{aligned}
\]
With \(S_{a,\ell}(t)=\Ebb[(\omega_i^\ell/n_i)\sum_jY_{ij}^a(t)]/\Ebb(\omega_i^\ell)\), the same calculation as in Proposition~\ref{supp:prop:full_data_na} gives \(S_{a,\ell}(t)=\prod_{0<u\le t}\{1-\dd\Lambda_{a,\ell}^D(u)\}\), \(\nu_{a,\ell}(\tau)=\int_0^\tau S_{a,\ell}(t)\dd t\), and \(\mu_{a,\ell}(\tau;\bw)=\int_0^\tau S_{a,\ell}(t-)\dd\Lambda_{a,\ell}^{R,\bw}(t)\).

\subsection{Proof of consistency}

Let \(\xi_i^a=\ind(A_i=a)/\pi_a\).  For each member \((i,j)\), define \(K_{ij}^\star\), \(H_{ij}^\star\), \(U_{ij}^\star\), \(r_{ij}^\star\), and \(\dd q_{ij}^{h,\star}\) exactly as in the IRT proof, but conditionally on the CRT baseline information \(\cV_{ij}\).  The CRT population local terms are
\[
\begin{aligned}
  \dd\mathcal A_{h,\ell}^\star(t)
  &=\Ebb\left[\frac{\omega_i^\ell}{n_i}\sum_{j=1}^{n_i}\left\{\xi_i^a\{K_{ij}^\star(t-)\}^{-1}\dd N_{ij}^h(t)+\{1-\xi_i^aU_{ij}^\star(t)\}\dd q_{ij}^{h,\star}(t)\right\}\right],\\
  \mathcal B_{\ell}^\star(t)
  &=\Ebb\left[\frac{\omega_i^\ell}{n_i}\sum_{j=1}^{n_i}\left\{\xi_i^a\{K_{ij}^\star(t-)\}^{-1}Y_{ij}(t)+\{1-\xi_i^aU_{ij}^\star(t)\}r_{ij}^\star(t)\right\}\right].
\end{aligned}
\]
Condition on the full cluster baseline information \(\cV_i\).  If \(\cC_a\) holds, then for each member \((i,j)\), the IRT identity gives
\begin{align*}
&\Ebb\left[\xi_i^a\{K_{ij,0}(t-)\}^{-1}\dd N_{ij}^h(t)+\{1-\xi_i^aU_{ij}^\star(t)\}\dd q_{ij}^{h,\star}(t)\mid\cV_i\right] \\
&\quad=\dd q_{ij,0}^h(t)+\{1-\Ebb(\xi_i^aU_{ij}^\star(t)\mid\cV_i)\}\dd q_{ij}^{h,\star}(t) \\
&\quad=\dd q_{ij,0}^h(t),\\
&\Ebb\left[\xi_i^a\{K_{ij,0}(t-)\}^{-1}Y_{ij}(t)+\{1-\xi_i^aU_{ij}^\star(t)\}r_{ij}^\star(t)\mid\cV_i\right] \\
&\quad=r_{ij,0}(t).
\end{align*}
Multiplying by \(\omega_i^\ell/n_i\), summing over \(j\), and taking expectation over clusters gives the numerator and denominator of \(\dd\Lambda_{a,\ell}^h(t)\).  If \(\cO_a^D\) holds for \(h=D\), or if \(\cO_a^D\cap\cO_a^R\) holds for \(h=R\), then, with \(R_{K,ij}(t)=K_{ij,0}(t-)/K_{ij}^\star(t-)\),
\begin{align*}
&\Ebb\left[\xi_i^a\{K_{ij}^\star(t-)\}^{-1}\dd N_{ij}^h(t)+\{1-\xi_i^aU_{ij}^\star(t)\}\dd q_{ij,0}^h(t)\mid\cV_i\right] \\
&\quad=R_{K,ij}(t)\dd q_{ij,0}^h(t)+\{1-R_{K,ij}(t)\}\dd q_{ij,0}^h(t) \\
&\quad=\dd q_{ij,0}^h(t),\\
&\Ebb\left[\xi_i^a\{K_{ij}^\star(t-)\}^{-1}Y_{ij}(t)+\{1-\xi_i^aU_{ij}^\star(t)\}r_{ij,0}(t)\mid\cV_i\right] \\
&\quad=R_{K,ij}(t)r_{ij,0}(t)+\{1-R_{K,ij}(t)\}r_{ij,0}(t) \\
&\quad=r_{ij,0}(t).
\end{align*}
Again the cluster-weighted sum recovers the target numerator and denominator.  Therefore \(\widehat\nu_{a,\ell}^{\mathrm{DR}}(\tau)\) is consistent if \(\cC_a\) or \(\cO_a^D\) holds; \(\widehat\mu_{a,\ell}^{\mathrm{DR}}(\tau;\bw)\) is consistent if \(\cC_a\) or \(\cO_a^D\cap\cO_a^R\) holds; and the corresponding while-alive rate and contrast are consistent by the same ratio and linearity argument as in Theorem~\ref{supp:thm:irt_dr_detailed}.

\section{CRT asymptotic expansion and variance}
\label{supp:sec:crt_variance_detailed}
In addition to the probability limits introduced in Section \ref{sec:crt_estimation}, let \(U_{ij}^{a,\star}(t)\) denote the uniform probability limit of \(\widehat U_{ij}^a(t)\), and define \(r_{ij}^{a,\star}(t)=H_{ij}^{a,\star}(t-)\). To characterize the probability limits of the two local estimators, define the limiting cluster-level augmented increments
\begin{align*}
  \dd A_{i,\ell}^{D,a,\star}(t)&=\frac{\omega_i^\ell}{n_i}\sum_{j=1}^{n_i}\left[\xi_i^a\{K_{ij}^{a,\star}(t-)\}^{-1}\dd N_{ij}^D(t)+\{1-\xi_i^a U_{ij}^{a,\star}(t)\}\dd q_{ij}^{D,\star}(t;a)\right], \\
  \dd A_{i,\ell}^{R,a,\star}(t)&=\frac{\omega_i^\ell}{n_i}\sum_{j=1}^{n_i}\left[\xi_i^a\{K_{ij}^{a,\star}(t-)\}^{-1}\dd N_{ij}^{\bw}(t)+\{1-\xi_i^a U_{ij}^{a,\star}(t)\}\dd q_{ij}^{R,\bw,\star}(t;a)\right],
\end{align*}
and the limiting weighted risk set
\[
  B_{i,\ell}^{a,\star}(t)=\frac{\omega_i^\ell}{n_i}\sum_{j=1}^{n_i}\left[\xi_i^a\{K_{ij}^{a,\star}(t-)\}^{-1}Y_{ij}(t)+\{1-\xi_i^a U_{ij}^{a,\star}(t)\}r_{ij}^{a,\star}(t)\right],
\]
with \(B_{a,\ell}^\star(t)=\Ebb\{B_{i,\ell}^{a,\star}(t)\}\). The probability limits of the two CRT local estimators are then
\begin{align}
  \dd\Lambda_{a,\ell}^{D,\star}(t)&=\frac{\Ebb\{\dd A_{i,\ell}^{D,a,\star}(t)\}}{B_{a,\ell}^\star(t)},\label{eq:crt_LamD_star}\\
  \dd\Lambda_{a,\ell}^{R,\bw,\star}(t)&=\frac{\Ebb\{\dd A_{i,\ell}^{R,a,\star}(t)\}}{B_{a,\ell}^\star(t)}.
  \label{eq:crt_LamR_star}
\end{align}
These expressions are the cluster level analogues of the limiting local quantities for the IRT estimator given in \ref{supp:sec:irt_asymptotic_detailed}. The target weight \(\omega_i^\ell\) enters both the numerator and the risk set, so the individual-average and cluster-average estimators generally have different probability limits when cluster size is informative. Throughout this section, $K_{ij}^{a,\star}$, $H_{ij}^{a,\star}$, $\dd q_{ij}^{D,\star}$, and $\dd q_{ij}^{R,\bw,\star}$ denote the uniform probability limits of the fitted nuisance functions, as defined in the main text. 

We impose the following regularity conditions for the CRT estimators.
 
\begin{assumption}[Asymptotic regularity for CRTs]
\label{asm:crt_asymptotic_regularity}
For each \(a\in\{0,1\}\) and \(\ell\in\{\mathrm{clus},\mathrm{ind}\}\), the following conditions hold on \([0,\tau]\).
\begin{enumerate}[label=(\roman*),ref=\ref{asm:crt_asymptotic_regularity}(\roman*)]
\item\label{asm:crt_reg_iid} The clusters \(\bO_1,\ldots,\bO_M\) are independent and identically distributed draws from a superpopulation of clusters.
\item\label{asm:crt_reg_positivity} The design probabilities \(\pi_i^a\), the true censoring survival function \(K_{ij}^a(t)\), the limiting censoring survival function \(K_{ij}^{a,\star}(t)\), and the limiting weighted risk set \(B_{a,\ell}^\star(t)\) are uniformly bounded away from zero.
\item\label{asm:crt_reg_variation} The terminal event cumulative hazards, censoring cumulative hazards, and recurrent event cumulative mean functions have uniformly bounded total variation on \([0,\tau]\).
\item\label{asm:crt_reg_moment} The target weights and weighted cluster totals have finite second moments, in particular, \(\Ebb\{(\omega_i^\ell)^2\}<\infty\) and \(\Ebb\left[\left\{\frac{\omega_i^\ell}{n_i}\sum_{j=1}^{n_i}N_{ij}^{a,\bw}(\tau)\right\}^2\right]<\infty\).
\item\label{asm:crt_reg_nuisance} The fitted censoring, terminal event, and recurrent event working models converge uniformly to their probability limits and admit asymptotically linear representations with cluster level influence functions.
\item\label{asm:crt_reg_hadamard} The product integral map defining \(S_{a,\ell}(t)\) and the Stieltjes integral maps defining \(\mu_{a,\ell}(\tau;\bw)\) and \(\nu_{a,\ell}(\tau)\) are Hadamard differentiable at their limiting values.
\end{enumerate}
\end{assumption}
 
Assumption \ref{asm:crt_asymptotic_regularity} allows arbitrary dependence among participants within a cluster, provided that the corresponding cluster level totals have sufficient moments. This is the standard asymptotic regime for CRTs with many independent clusters \citep{balzer2019hierarchical,balzer2023twostage,wang2024modelrobust,li2025standardization,fang2026estimands}. 

The independent sampling unit in a CRT is the cluster.  Let \(\PM f=M^{-1}\sum_{i=1}^Mf(\bO_i)\).  For \(h\in\{D,R\}\), define
\[
\begin{aligned}
  \dd A_{h,\ell,i}(t)
  &=\frac{\omega_i^\ell}{n_i}\sum_{j=1}^{n_i}\left\{\xi_i^a\{K_{ij}^\star(t-)\}^{-1}\dd N_{ij}^h(t)+\{1-\xi_i^aU_{ij}^\star(t)\}\dd q_{ij}^{h,\star}(t)\right\},\\
  B_{\ell,i}(t)
  &=\frac{\omega_i^\ell}{n_i}\sum_{j=1}^{n_i}\left\{\xi_i^a\{K_{ij}^\star(t-)\}^{-1}Y_{ij}(t)+\{1-\xi_i^aU_{ij}^\star(t)\}r_{ij}^\star(t)\right\}.
\end{aligned}
\]
Let \(a_{h,\ell}(t)=\Ebb\{\dd A_{h,\ell,i}(t)\}\), \(b_\ell(t)=\Ebb\{B_{\ell,i}(t)\}\), and \(\dd\Lambda_{h,\ell}^\star(t)=a_{h,\ell}(t)/b_\ell(t)\).  If the nuisance estimators are fitted from cluster-level estimating equations, their expansion has the form \(\sqrt M(\widehat\eta-\eta^\star)=M^{-1/2}\sum_i\Phi_{\eta,i}+o_p(1)\). For the recurrent event nuisance models, the contribution \(\Phi_{\eta,i}\) is the cluster sum of the complete LWYY mean zero rate residual scores, which include both observed event and compensator contributions and are not interpreted as recurrent event martingales. For the terminal event and censoring models, the contributions are the corresponding cluster sums of Cox martingale scores together with the Breslow baseline contributions. These nuisance contributions are propagated through the augmented local estimators and the while-alive contrast as follows.  Write the cluster-level first-order perturbations induced by cluster \(i\) as \(\dot K_{rj,i}(t)\), \(\dot U_{rj,i}(t)\), \(\dot r_{rj,i}(t)\), and \(\dd\dot q_{rj,i}^h(t)\) for the fitted censoring survival, censoring augmentation, terminal-event survival risk, and full data increment of member \(j\) in cluster \(r\).  For a fixed evaluation cluster \(r\), the perturbation of the memberwise numerator is
\[
\begin{aligned}
  \dd\dot A_{h,rj,t}[i]
  &=-\xi_r^a\{K_{rj}^\star(t-)\}^{-2}\dot K_{rj,i}(t-)\dd N_{rj}^h(t) \\
  &\quad-\xi_r^a\dot U_{rj,i}(t)\dd q_{rj}^{h,\star}(t)+\{1-\xi_r^aU_{rj}^\star(t)\}\dd\dot q_{rj,i}^h(t),
\end{aligned}
\]
and the perturbation of the memberwise risk denominator is
\[
\begin{aligned}
  \dot B_{rj,t}[i]
  &=-\xi_r^a\{K_{rj}^\star(t-)\}^{-2}\dot K_{rj,i}(t-)Y_{rj}(t) \\
  &\quad-\xi_r^a\dot U_{rj,i}(t)r_{rj}^\star(t)+\{1-\xi_r^aU_{rj}^\star(t)\}\dot r_{rj,i}(t).
\end{aligned}
\]
The corresponding cluster-weighted derivatives are therefore
\[
  \dot A_{h,\ell,t}[i]
  =\Ebb_r\left[\frac{\omega_r^\ell}{n_r}\sum_{j=1}^{n_r}\dd\dot A_{h,rj,t}[i]\right],
  \qquad
  \dot B_{\ell,t}[i]
  =\Ebb_r\left[\frac{\omega_r^\ell}{n_r}\sum_{j=1}^{n_r}\dot B_{rj,t}[i]\right],
\]
where \(\Ebb_r\) denotes expectation over an independent evaluation cluster \(r\).  Thus
\[
  \zeta_{h,\ell,i}(t)=\dd A_{h,\ell,i}(t)-a_{h,\ell}(t)+\dot A_{h,\ell,t}[i],
  \qquad
  \zeta_{B,\ell,i}(t)=B_{\ell,i}(t)-b_\ell(t)+\dot B_{\ell,t}[i].
\]
Consecutive quotient expansion gives
\[
\begin{aligned}
  \sqrt M\{\dd\widehat\Lambda_{h,\ell}(t)-\dd\Lambda_{h,\ell}^\star(t)\}
  &=\sqrt M\left\{\frac{\widehat a_{h,\ell}(t)}{\widehat b_\ell(t)}-\frac{a_{h,\ell}(t)}{b_\ell(t)}\right\} \\
  &=\frac{\sqrt M\{\widehat a_{h,\ell}(t)-a_{h,\ell}(t)\}}{b_\ell(t)}
    -\frac{a_{h,\ell}(t)}{b_\ell(t)^2}\sqrt M\{\widehat b_\ell(t)-b_\ell(t)\}+o_p(1) \\
  &=M^{-1/2}\sum_{i=1}^M
    \frac{\zeta_{h,\ell,i}(t)-\dd\Lambda_{h,\ell}^\star(t)\zeta_{B,\ell,i}(t)}{b_\ell(t)}+o_p(1).
\end{aligned}
\]
The local cluster-level influence process is consequently
\[
  \phi_{\Lambda^h,\ell,i}(t)=\frac{\zeta_{h,\ell,i}(t)-\dd\Lambda_{h,\ell}^\star(t)\zeta_{B,\ell,i}(t)}{b_\ell(t)},
  \qquad h\in\{D,R\}.
\]
Let \(S_{a,\ell}^\star(t)=\prod_{0<u\le t}\{1-\dd\Lambda_{D,\ell}^\star(u)\}\).  The product-integral perturbation from cluster \(i\) is
\[
  \phi_{S,a,\ell,i}(t)
  =-S_{a,\ell}^\star(t)
   \int_{(0,t]}\frac{\dd\phi_{\Lambda^D,\ell,i}(u)}{1-\dd\Lambda_{D,\ell}^\star(u)},
\]
with the continuous-hazard simplification \(-S_{a,\ell}^\star(t)\int_0^t\dd\phi_{\Lambda^D,\ell,i}(u)\).  Therefore
\[
  \phi_{\nu,a,\ell,i}(\tau)=\int_0^\tau\phi_{S,a,\ell,i}(t)\dd t.
\]
The quantities \(\phi_{\nu,a,\ell,i}\), \(\phi_{\mu,a,\ell,i}^{\bw}\), \(\phi_{\psi,a,\ell,i}^{\bw}\), and \(\phi_{\Delta,\ell,i}^{\bw}\) are the cluster-level influence functions denoted \(\Phi_{\nu,a,\ell,i}\), \(\Phi_{\mu,a,\ell,i}\), \(\Phi_{\psi,a,\ell,i}\), and \(\Phi_{\Delta,\ell,i}\) in the main text.
For the recurrent-event burden, the Stieltjes product rule yields
\[
  \phi_{\mu,a,\ell,i}^{\bw}(\tau)
  =\int_0^\tau S_{a,\ell}^\star(t-)\dd\phi_{\Lambda^R,\ell,i}(t)
   +\int_0^\tau\phi_{S,a,\ell,i}(t-)\dd\Lambda_{R,\ell}^\star(t).
\]
The while-alive rate influence function is obtained by differentiating \(\mu/\nu\):
\[
  \phi_{\psi,a,\ell,i}^{\bw}(\tau)
  =\frac{\phi_{\mu,a,\ell,i}^{\bw}(\tau)}{\nu_{a,\ell}^\star(\tau)}
   -\frac{\mu_{a,\ell}^\star(\tau;\bw)}{\{\nu_{a,\ell}^\star(\tau)\}^2}\phi_{\nu,a,\ell,i}(\tau).
\]
Finally, because the two arm-specific estimators are formed from the same independent clusters but have disjoint treatment-weighted empirical contributions, the cluster-level contrast influence function is
\[
  \phi_{\Delta,\ell,i}^{\bw}(\tau)=\phi_{\psi,1,\ell,i}^{\bw}(\tau)-\phi_{\psi,0,\ell,i}^{\bw}(\tau),
\]
and hence
\[
  \sqrt M\{\widehat\Delta_\ell^{\mathrm{DR}}(\tau;\bw)-\Delta_\ell(\tau;\bw)\}
  =M^{-1/2}\sum_{i=1}^M\phi_{\Delta,\ell,i}^{\bw}(\tau)+o_p(1).
\]
The analytic cluster-level influence function variance estimator is
\[
  \widehat V_{\mathrm{IF}}\{\widehat\Delta_\ell^{\mathrm{DR}}(\tau;\bw)\}
  =\frac{1}{M(M-1)}\sum_{i=1}^M\left\{\widehat\phi_{\Delta,\ell,i}^{\bw}(\tau)-M^{-1}\sum_{r=1}^M\widehat\phi_{\Delta,\ell,r}^{\bw}(\tau)\right\}^2.
\]
For finite or moderate numbers of clusters, Wald intervals may use a \(t_{M-2}\) critical value.

\section{Details of simulation settings} \label{supp:sec:simulation_settings}

Table \ref{supp:tab:sim_irt_dgp} summarizes the process-specific parameters in the simulation studies for IRT. 
\begin{table}[tbp]
\caption{Continuous-time data-generating parameters for the IRT simulation. Coefficient vectors correspond to \(\bh_i=(Z_{i1},Z_{i2},Z_{i1}Z_{i2})^{\mathsf T}\). The censoring baseline scale \(\lambda_C\) was calibrated separately for each \((\tau,\gamma)\).}
\label{supp:tab:sim_irt_dgp}
\centering
\renewcommand{\arraystretch}{1.08}
\begin{adjustbox}{max width=\textwidth}
\begin{tabular}{lccccc}
\toprule
Process & \(\rho_h\) & \(\lambda_{h,0}\) & \(\lambda_{h,1}\) & \(\bbeta_{h,0}^{\mathsf T}\) & \(\bbeta_{h,1}^{\mathsf T}\) \\
\midrule
Terminal event & \(1.20\) & \(0.055\) & \(0.040\) & \((0.25,0.85,0.75)\) & \((0.25,-0.65,-0.55)\) \\
Recurrent type 1 & \(1.08\) & \(0.22\) & \(0.35\) & \((0.15,1.00,0.80)\) & \((0.15,-0.85,-0.70)\) \\
Recurrent type 2 & \(1.18\) & \(0.14\) & \(0.42\) & \((-0.10,0.75,0.60)\) & \((-0.10,-0.65,-0.55)\) \\
Censoring & \(1.10\) & \(\exp(c_{\tau,\gamma})\) & \(\exp(c_{\tau,\gamma})\) & \((0.12,0.50,0.3125)\) & \((0.12,-0.50,-0.3125)\) \\
\bottomrule
\end{tabular}
\end{adjustbox}
\end{table}

Table \ref{supp:tab:crt_sim_dgp} summarizes the process-specific parameters in the simulation studies for CRT. The coefficient order is \((L_i,Z_{ij1},Z_{ij2},n_i^\star,n_i^\star Z_{ij2},Z_{ij1}Z_{ij2})\).
\begin{table}[tbp]
\caption{Data-generating parameters for the CRT simulation.}
\label{supp:tab:crt_sim_dgp}
\centering
\setlength{\tabcolsep}{3.0pt}
\renewcommand{\arraystretch}{1.08}
\begin{adjustbox}{max width=\textwidth}
\begin{tabular}{lcccll}
\toprule
Process & \(\lambda_0\) & \(\lambda_1\) & $\rho_h$
& \(\bbeta_0^\mathsf{T}\) & \(\bbeta_1^\mathsf{T}\) \\
\midrule
Terminal event
& \(0.045\) & \(0.032\) & \(1.20\)
& \((0.25,0.35,0.45,0.45,0.20,0.20)\)
& \((0.25,0.20,-0.35,-0.35,-0.15,-0.15)\) \\
Recurrent event 1
& \(0.18\) & \(0.13\) & \(1.08\)
& \((0.20,0.35,0.55,0.45,0.25,0.25)\)
& \((0.20,0.20,-0.40,-0.35,-0.20,-0.20)\) \\
Recurrent event 2
& \(0.12\) & \(0.10\) & \(1.18\)
& \((-0.10,0.25,0.45,0.35,0.20,0.20)\)
& \((-0.10,0.15,-0.35,-0.30,-0.15,-0.15)\) \\
Censoring
& \(\lambda_C\) & \(\lambda_C\) & \(1.10\)
& \((0.08,0.12,0.216,0.144,0.048,0.048)\)
& \((0.08,0.10,-0.216,-0.144,-0.048,-0.048)\) \\
\bottomrule
\end{tabular}
\end{adjustbox}
\end{table}

\section{Additional simulation results}
\label{supp:sec:simulation_tables_detailed}

This section reports additional simulation results using the same side-by-side estimator layout as the main-text simulation tables, with all tables kept in portrait orientation. The true value is included directly in each table. The column labeled Cens. is the empirical mean observed censoring percentage. The columns labeled O and C, or the paired column \(O,C\), indicate whether the outcome-side and censoring-side nuisance models are correctly specified, respectively. Relative bias and empirical coverage are reported as percentages, MCSD is the Monte Carlo standard deviation, and ASE is the average estimated standard error.

Table~\ref{supp:tab:irt_n1600} reports the IRT simulation scenarios with \(n=1600\). Table~\ref{supp:tab:irt_n3200} reports the corresponding IRT scenarios with \(n=3200\). The tables include both horizons and both censoring levels, and compare the doubly robust estimator with the IPCW and OR estimators under all four nuisance-model configurations.

Table~\ref{supp:tab:crt_individual_M50} and Table~\ref{supp:tab:crt_individual_M80} report the individual-level CRT estimand with \(M=50\) and \(M=80\) clusters. These tables target the effect for the average participant and therefore use the individual-level weighting scheme from the main text.

The cluster-level CRT estimand with \(M=50\) is reported in Table~\ref{tab:sim_crt_semiparametric} of the main text; Table~\ref{supp:tab:crt_cluster_M80} reports the corresponding results with \(M=80\) clusters. These tables target the effect for the average cluster and are therefore the appropriate counterpart when the scientific question gives equal weight to clusters.

\begin{table}[p]
\caption{Additional Monte Carlo performance for the while-alive treatment contrast in the individually randomized trial with \(n=1600\). Cens. is the mean observed censoring percentage. \(O1\) and \(C1\) denote correct outcome-side and censoring-side working models, respectively. RBias is relative bias in percent, MCSD is the Monte Carlo standard deviation, ASE is the average influence function standard error, and Cov is empirical coverage in percent of the nominal \(95\%\) interval.}
\label{supp:tab:irt_n1600}
\centering
\setlength{\tabcolsep}{2.2pt}
\renewcommand{\arraystretch}{1.04}
\begin{adjustbox}{max width=\textwidth}
\begin{tabular}{@{}cccc*{12}{r}@{}}
\toprule
& & & & \multicolumn{4}{c}{\textsc{dr}}
& \multicolumn{4}{c}{\textsc{ipcw}}
& \multicolumn{4}{c}{\textsc{or}} \\
\cmidrule(lr){5-8}\cmidrule(lr){9-12}\cmidrule(lr){13-16}
\(\tau\) & Cens. & \(O,C\) & \(\Delta\)
& RBias & MCSD & ASE & Cov
& RBias & MCSD & ASE & Cov
& RBias & MCSD & ASE & Cov \\
\midrule
3 & 40.0 & O1C1 & 0.542 & +0.1 & 0.061 & 0.060 & 95.1 & +0.1 & 0.062 & 0.061 & 94.5 & +0.3 & 0.060 & 0.060 & 95.6 \\
 &  & O1C0 & 0.542 & +0.1 & 0.060 & 0.059 & 95.7 & -5.9 & 0.056 & 0.055 & 89.7 & +0.3 & 0.060 & 0.060 & 95.6 \\
 &  & O0C1 & 0.542 & +0.1 & 0.062 & 0.061 & 94.4 & +0.1 & 0.062 & 0.061 & 94.5 & -5.8 & 0.056 & 0.055 & 90.2 \\
 &  & O0C0 & 0.542 & -6.0 & 0.056 & 0.055 & 90.0 & -5.9 & 0.056 & 0.055 & 89.7 & -5.8 & 0.056 & 0.055 & 90.2 \\
\addlinespace[2pt]
3 & 60.0 & O1C1 & 0.542 & -0.9 & 0.069 & 0.069 & 95.7 & -0.6 & 0.071 & 0.071 & 95.6 & -0.3 & 0.067 & 0.067 & 95.4 \\
 &  & O1C0 & 0.542 & -0.7 & 0.067 & 0.067 & 95.8 & -11.5 & 0.057 & 0.057 & 78.1 & -0.3 & 0.067 & 0.067 & 95.4 \\
 &  & O0C1 & 0.542 & -0.5 & 0.071 & 0.071 & 96.0 & -0.6 & 0.071 & 0.071 & 95.6 & -11.2 & 0.057 & 0.057 & 79.0 \\
 &  & O0C0 & 0.542 & -11.5 & 0.057 & 0.057 & 78.2 & -11.5 & 0.057 & 0.057 & 78.1 & -11.2 & 0.057 & 0.057 & 79.0 \\
\addlinespace[2pt]
5 & 40.1 & O1C1 & 0.581 & -0.4 & 0.061 & 0.059 & 94.2 & -0.3 & 0.062 & 0.060 & 93.7 & -0.1 & 0.060 & 0.058 & 94.0 \\
 &  & O1C0 & 0.581 & -0.3 & 0.060 & 0.058 & 94.1 & -7.2 & 0.054 & 0.054 & 85.7 & -0.1 & 0.060 & 0.058 & 94.0 \\
 &  & O0C1 & 0.581 & -0.3 & 0.062 & 0.060 & 93.9 & -0.3 & 0.062 & 0.060 & 93.7 & -7.0 & 0.054 & 0.054 & 86.0 \\
 &  & O0C0 & 0.581 & -7.2 & 0.054 & 0.054 & 85.3 & -7.2 & 0.054 & 0.054 & 85.7 & -7.0 & 0.054 & 0.054 & 86.0 \\
\addlinespace[2pt]
5 & 60.0 & O1C1 & 0.581 & -0.4 & 0.074 & 0.070 & 92.3 & -0.3 & 0.076 & 0.071 & 92.3 & +0.2 & 0.070 & 0.066 & 93.4 \\
 &  & O1C0 & 0.581 & -0.2 & 0.070 & 0.066 & 93.0 & -12.4 & 0.058 & 0.055 & 71.8 & +0.2 & 0.070 & 0.066 & 93.4 \\
 &  & O0C1 & 0.581 & -0.2 & 0.076 & 0.072 & 92.4 & -0.3 & 0.076 & 0.071 & 92.3 & -12.0 & 0.059 & 0.056 & 72.6 \\
 &  & O0C0 & 0.581 & -12.4 & 0.058 & 0.055 & 71.4 & -12.4 & 0.058 & 0.055 & 71.8 & -12.0 & 0.059 & 0.056 & 72.6 \\
\bottomrule
\end{tabular}
\end{adjustbox}
\end{table}

\begin{table}[p]
\caption{Additional Monte Carlo performance for the while-alive treatment contrast in the individually randomized trial with \(n=3200\). Cens. is the mean observed censoring percentage. \(O1\) and \(C1\) denote correct outcome-side and censoring-side working models, respectively. RBias is relative bias in percent, MCSD is the Monte Carlo standard deviation, ASE is the average influence function standard error, and Cov is empirical coverage in percent of the nominal \(95\%\) interval.}
\label{supp:tab:irt_n3200}
\centering
\setlength{\tabcolsep}{2.2pt}
\renewcommand{\arraystretch}{1.04}
\begin{adjustbox}{max width=\textwidth}
\begin{tabular}{@{}cccc*{12}{r}@{}}
\toprule
& & & & \multicolumn{4}{c}{\textsc{dr}}
& \multicolumn{4}{c}{\textsc{ipcw}}
& \multicolumn{4}{c}{\textsc{or}} \\
\cmidrule(lr){5-8}\cmidrule(lr){9-12}\cmidrule(lr){13-16}
\(\tau\) & Cens. & \(O,C\) & \(\Delta\)
& RBias & MCSD & ASE & Cov
& RBias & MCSD & ASE & Cov
& RBias & MCSD & ASE & Cov \\
\midrule
3 & 40.0 & O1C1 & 0.542 & +0.0 & 0.043 & 0.043 & 95.0 & +0.1 & 0.044 & 0.044 & 95.3 & +0.2 & 0.042 & 0.042 & 95.7 \\
 &  & O1C0 & 0.542 & +0.0 & 0.042 & 0.042 & 95.3 & -6.2 & 0.039 & 0.039 & 84.9 & +0.2 & 0.042 & 0.042 & 95.7 \\
 &  & O0C1 & 0.542 & +0.1 & 0.044 & 0.044 & 95.3 & +0.1 & 0.044 & 0.044 & 95.3 & -6.0 & 0.039 & 0.039 & 86.2 \\
 &  & O0C0 & 0.542 & -6.2 & 0.039 & 0.039 & 85.4 & -6.2 & 0.039 & 0.039 & 84.9 & -6.0 & 0.039 & 0.039 & 86.2 \\
\addlinespace[2pt]
3 & 60.0 & O1C1 & 0.542 & -0.4 & 0.050 & 0.050 & 94.3 & -0.3 & 0.051 & 0.051 & 94.2 & +0.1 & 0.048 & 0.048 & 94.3 \\
 &  & O1C0 & 0.542 & -0.3 & 0.047 & 0.047 & 94.1 & -11.3 & 0.040 & 0.040 & 68.1 & +0.1 & 0.048 & 0.048 & 94.3 \\
 &  & O0C1 & 0.542 & -0.3 & 0.051 & 0.051 & 94.3 & -0.3 & 0.051 & 0.051 & 94.2 & -10.9 & 0.040 & 0.040 & 69.7 \\
 &  & O0C0 & 0.542 & -11.3 & 0.040 & 0.040 & 68.2 & -11.3 & 0.040 & 0.040 & 68.1 & -10.9 & 0.040 & 0.040 & 69.7 \\
\addlinespace[2pt]
5 & 40.0 & O1C1 & 0.581 & -0.1 & 0.042 & 0.042 & 94.4 & +0.0 & 0.043 & 0.043 & 94.6 & +0.2 & 0.041 & 0.041 & 94.9 \\
 &  & O1C0 & 0.581 & -0.1 & 0.041 & 0.041 & 94.5 & -7.0 & 0.038 & 0.038 & 80.8 & +0.2 & 0.041 & 0.041 & 94.9 \\
 &  & O0C1 & 0.581 & +0.0 & 0.043 & 0.043 & 94.7 & +0.0 & 0.043 & 0.043 & 94.6 & -6.8 & 0.038 & 0.038 & 81.8 \\
 &  & O0C0 & 0.581 & -7.0 & 0.038 & 0.038 & 80.8 & -7.0 & 0.038 & 0.038 & 80.8 & -6.8 & 0.038 & 0.038 & 81.8 \\
\addlinespace[2pt]
5 & 60.0 & O1C1 & 0.581 & -0.4 & 0.048 & 0.050 & 95.3 & -0.2 & 0.049 & 0.051 & 95.5 & +0.1 & 0.045 & 0.047 & 95.7 \\
 &  & O1C0 & 0.581 & -0.3 & 0.045 & 0.047 & 95.6 & -12.3 & 0.038 & 0.039 & 55.4 & +0.1 & 0.045 & 0.047 & 95.7 \\
 &  & O0C1 & 0.581 & -0.2 & 0.049 & 0.051 & 95.6 & -0.2 & 0.049 & 0.051 & 95.5 & -11.9 & 0.038 & 0.039 & 58.4 \\
 &  & O0C0 & 0.581 & -12.3 & 0.038 & 0.039 & 55.4 & -12.3 & 0.038 & 0.039 & 55.4 & -11.9 & 0.038 & 0.039 & 58.4 \\
\bottomrule
\end{tabular}
\end{adjustbox}
\end{table}

\begin{table}[p]
\caption{Additional Monte Carlo performance for the individual-average while-alive treatment contrast with \(M=50\). Cens. is the mean observed censoring percentage. \(O1\) and \(C1\) denote correct outcome-side and censoring-side working models, respectively. RBias is relative bias in percent, MCSD is the Monte Carlo standard deviation, ASE is the average cluster influence function standard error, and Cov is empirical coverage in percent of the nominal \(95\%\) interval based on a \(t_{M-2}\) critical value.}
\label{supp:tab:crt_individual_M50}
\centering
\setlength{\tabcolsep}{2.2pt}
\renewcommand{\arraystretch}{1.04}
\begin{adjustbox}{max width=\textwidth}
\begin{tabular}{@{}cccc*{12}{r}@{}}
\toprule
& & & & \multicolumn{4}{c}{\textsc{dr}}
& \multicolumn{4}{c}{\textsc{ipcw}}
& \multicolumn{4}{c}{\textsc{or}} \\
\cmidrule(lr){5-8}\cmidrule(lr){9-12}\cmidrule(lr){13-16}
\(\tau\) & Cens. & \(O,C\) & \(\Delta_{\mathrm{ind}}\)
& RBias & MCSD & ASE & Cov
& RBias & MCSD & ASE & Cov
& RBias & MCSD & ASE & Cov \\
\midrule
3 & 40.0 & O1C1 & -0.428 & +0.0 & 0.105 & 0.094 & 91.7 & -0.3 & 0.109 & 0.101 & 91.5 & +0.5 & 0.105 & 0.094 & 92.2 \\
 &  & O1C0 & -0.428 & +0.1 & 0.105 & 0.094 & 92.1 & -6.9 & 0.102 & 0.095 & 88.6 & +0.5 & 0.105 & 0.094 & 92.2 \\
 &  & O0C1 & -0.428 & -0.6 & 0.109 & 0.099 & 91.3 & -0.3 & 0.109 & 0.101 & 91.5 & -6.9 & 0.101 & 0.093 & 88.2 \\
 &  & O0C0 & -0.428 & -7.3 & 0.101 & 0.092 & 88.0 & -6.9 & 0.102 & 0.095 & 88.6 & -6.9 & 0.101 & 0.093 & 88.2 \\
\addlinespace[2pt]
3 & 60.0 & O1C1 & -0.428 & +0.5 & 0.105 & 0.096 & 91.9 & +0.2 & 0.110 & 0.103 & 91.9 & +1.3 & 0.105 & 0.097 & 92.2 \\
 &  & O1C0 & -0.428 & +0.7 & 0.105 & 0.096 & 92.0 & -11.5 & 0.097 & 0.092 & 84.0 & +1.3 & 0.105 & 0.097 & 92.2 \\
 &  & O0C1 & -0.428 & -0.2 & 0.110 & 0.102 & 91.0 & +0.2 & 0.110 & 0.103 & 91.9 & -11.3 & 0.096 & 0.090 & 85.1 \\
 &  & O0C0 & -0.428 & -11.9 & 0.096 & 0.090 & 85.0 & -11.5 & 0.097 & 0.092 & 84.0 & -11.3 & 0.096 & 0.090 & 85.1 \\
\addlinespace[2pt]
5 & 40.0 & O1C1 & -0.403 & -2.2 & 0.101 & 0.089 & 89.2 & -2.4 & 0.104 & 0.095 & 89.5 & -1.4 & 0.101 & 0.090 & 89.7 \\
 &  & O1C0 & -0.403 & -2.0 & 0.101 & 0.089 & 89.4 & -8.3 & 0.098 & 0.090 & 86.7 & -1.4 & 0.101 & 0.090 & 89.7 \\
 &  & O0C1 & -0.403 & -2.7 & 0.103 & 0.093 & 89.2 & -2.4 & 0.104 & 0.095 & 89.5 & -8.1 & 0.098 & 0.088 & 86.6 \\
 &  & O0C0 & -0.403 & -8.6 & 0.097 & 0.088 & 86.3 & -8.3 & 0.098 & 0.090 & 86.7 & -8.1 & 0.098 & 0.088 & 86.6 \\
\addlinespace[2pt]
5 & 60.0 & O1C1 & -0.403 & -1.5 & 0.104 & 0.092 & 90.9 & -1.2 & 0.106 & 0.099 & 91.2 & -0.3 & 0.104 & 0.093 & 91.4 \\
 &  & O1C0 & -0.403 & -1.2 & 0.104 & 0.093 & 91.1 & -11.9 & 0.094 & 0.089 & 85.8 & -0.3 & 0.104 & 0.093 & 91.4 \\
 &  & O0C1 & -0.403 & -2.1 & 0.107 & 0.098 & 90.2 & -1.2 & 0.106 & 0.099 & 91.2 & -12.0 & 0.094 & 0.087 & 84.7 \\
 &  & O0C0 & -0.403 & -12.7 & 0.093 & 0.086 & 84.0 & -11.9 & 0.094 & 0.089 & 85.8 & -12.0 & 0.094 & 0.087 & 84.7 \\
\bottomrule
\end{tabular}
\end{adjustbox}
\end{table}

\begin{table}[p]
\caption{Additional Monte Carlo performance for the individual-average while-alive treatment contrast with \(M=80\). Cens. is the mean observed censoring percentage. \(O1\) and \(C1\) denote correct outcome-side and censoring-side working models, respectively. RBias is relative bias in percent, MCSD is the Monte Carlo standard deviation, ASE is the average cluster influence function standard error, and Cov is empirical coverage in percent of the nominal \(95\%\) interval based on a \(t_{M-2}\) critical value.}
\label{supp:tab:crt_individual_M80}
\centering
\setlength{\tabcolsep}{2.2pt}
\renewcommand{\arraystretch}{1.04}
\begin{adjustbox}{max width=\textwidth}
\begin{tabular}{@{}cccc*{12}{r}@{}}
\toprule
& & & & \multicolumn{4}{c}{\textsc{dr}}
& \multicolumn{4}{c}{\textsc{ipcw}}
& \multicolumn{4}{c}{\textsc{or}} \\
\cmidrule(lr){5-8}\cmidrule(lr){9-12}\cmidrule(lr){13-16}
\(\tau\) & Cens. & \(O,C\) & \(\Delta_{\mathrm{ind}}\)
& RBias & MCSD & ASE & Cov
& RBias & MCSD & ASE & Cov
& RBias & MCSD & ASE & Cov \\
\midrule
3 & 40.0 & O1C1 & -0.428 & -0.7 & 0.082 & 0.075 & 90.3 & -0.6 & 0.086 & 0.081 & 91.7 & -0.1 & 0.083 & 0.076 & 91.2 \\
 &  & O1C0 & -0.428 & -0.6 & 0.082 & 0.075 & 90.4 & -7.2 & 0.080 & 0.076 & 88.5 & -0.1 & 0.083 & 0.076 & 91.2 \\
 &  & O0C1 & -0.428 & -1.0 & 0.085 & 0.079 & 91.5 & -0.6 & 0.086 & 0.081 & 91.7 & -7.2 & 0.079 & 0.074 & 88.1 \\
 &  & O0C0 & -0.428 & -7.6 & 0.079 & 0.074 & 87.7 & -7.2 & 0.080 & 0.076 & 88.5 & -7.2 & 0.079 & 0.074 & 88.1 \\
\addlinespace[2pt]
3 & 60.0 & O1C1 & -0.428 & -1.2 & 0.081 & 0.078 & 93.1 & -1.3 & 0.084 & 0.083 & 91.5 & -0.3 & 0.081 & 0.078 & 93.4 \\
 &  & O1C0 & -0.428 & -1.0 & 0.081 & 0.077 & 92.9 & -13.0 & 0.074 & 0.073 & 83.0 & -0.3 & 0.081 & 0.078 & 93.4 \\
 &  & O0C1 & -0.428 & -1.5 & 0.083 & 0.082 & 91.8 & -1.3 & 0.084 & 0.083 & 91.5 & -12.7 & 0.073 & 0.072 & 83.3 \\
 &  & O0C0 & -0.428 & -13.3 & 0.072 & 0.072 & 82.0 & -13.0 & 0.074 & 0.073 & 83.0 & -12.7 & 0.073 & 0.072 & 83.3 \\
\addlinespace[2pt]
5 & 40.0 & O1C1 & -0.403 & -1.5 & 0.077 & 0.072 & 92.1 & -1.3 & 0.081 & 0.076 & 92.3 & -0.7 & 0.077 & 0.073 & 92.3 \\
 &  & O1C0 & -0.403 & -1.4 & 0.077 & 0.072 & 91.9 & -7.2 & 0.076 & 0.072 & 88.2 & -0.7 & 0.077 & 0.073 & 92.3 \\
 &  & O0C1 & -0.403 & -1.7 & 0.080 & 0.075 & 91.9 & -1.3 & 0.081 & 0.076 & 92.3 & -7.1 & 0.076 & 0.071 & 88.1 \\
 &  & O0C0 & -0.403 & -7.6 & 0.075 & 0.071 & 87.7 & -7.2 & 0.076 & 0.072 & 88.2 & -7.1 & 0.076 & 0.071 & 88.1 \\
\addlinespace[2pt]
5 & 60.0 & O1C1 & -0.403 & -0.3 & 0.078 & 0.075 & 92.9 & -0.3 & 0.081 & 0.078 & 93.2 & +0.9 & 0.078 & 0.075 & 93.0 \\
 &  & O1C0 & -0.403 & +0.0 & 0.078 & 0.074 & 92.8 & -11.1 & 0.073 & 0.070 & 84.7 & +0.9 & 0.078 & 0.075 & 93.0 \\
 &  & O0C1 & -0.403 & -0.6 & 0.082 & 0.078 & 92.9 & -0.3 & 0.081 & 0.078 & 93.2 & -10.7 & 0.073 & 0.069 & 85.3 \\
 &  & O0C0 & -0.403 & -11.5 & 0.072 & 0.069 & 83.9 & -11.1 & 0.073 & 0.070 & 84.7 & -10.7 & 0.073 & 0.069 & 85.3 \\
\bottomrule
\end{tabular}
\end{adjustbox}
\end{table}

\begin{table}[p]
\caption{Additional Monte Carlo performance for the cluster-average while-alive treatment contrast with \(M=80\). Cens. is the mean observed censoring percentage. \(O1\) and \(C1\) denote correct outcome-side and censoring-side working models, respectively. RBias is relative bias in percent, MCSD is the Monte Carlo standard deviation, ASE is the average cluster influence function standard error, and Cov is empirical coverage in percent of the nominal \(95\%\) interval based on a \(t_{M-2}\) critical value.}
\label{supp:tab:crt_cluster_M80}
\centering
\setlength{\tabcolsep}{2.2pt}
\renewcommand{\arraystretch}{1.04}
\begin{adjustbox}{max width=\textwidth}
\begin{tabular}{@{}cccc*{12}{r}@{}}
\toprule
& & & & \multicolumn{4}{c}{\textsc{dr}}
& \multicolumn{4}{c}{\textsc{ipcw}}
& \multicolumn{4}{c}{\textsc{or}} \\
\cmidrule(lr){5-8}\cmidrule(lr){9-12}\cmidrule(lr){13-16}
\(\tau\) & Cens. & \(O,C\) & \(\Delta_{\mathrm{clus}}\)
& RBias & MCSD & ASE & Cov
& RBias & MCSD & ASE & Cov
& RBias & MCSD & ASE & Cov \\
\midrule
3 & 40.0 & O1C1 & -0.309 & -0.3 & 0.070 & 0.065 & 92.2 & +0.2 & 0.073 & 0.070 & 92.8 & +0.4 & 0.070 & 0.066 & 92.3 \\
 &  & O1C0 & -0.309 & -0.2 & 0.070 & 0.065 & 92.3 & -7.0 & 0.069 & 0.065 & 90.4 & +0.4 & 0.070 & 0.066 & 92.3 \\
 &  & O0C1 & -0.309 & -0.1 & 0.072 & 0.069 & 92.9 & +0.2 & 0.073 & 0.070 & 92.8 & +24.7 & 0.077 & 0.072 & 85.5 \\
 &  & O0C0 & -0.309 & -7.5 & 0.068 & 0.065 & 89.9 & -7.0 & 0.069 & 0.065 & 90.4 & +24.7 & 0.077 & 0.072 & 85.5 \\
\addlinespace[2pt]
3 & 60.0 & O1C1 & -0.309 & -0.8 & 0.068 & 0.067 & 94.4 & -0.6 & 0.070 & 0.071 & 94.1 & +0.0 & 0.069 & 0.067 & 94.2 \\
 &  & O1C0 & -0.309 & -0.7 & 0.068 & 0.067 & 94.2 & -13.4 & 0.062 & 0.063 & 87.0 & +0.0 & 0.069 & 0.067 & 94.2 \\
 &  & O0C1 & -0.309 & -0.7 & 0.070 & 0.070 & 94.6 & -0.6 & 0.070 & 0.071 & 94.1 & +17.2 & 0.071 & 0.070 & 91.6 \\
 &  & O0C0 & -0.309 & -13.9 & 0.062 & 0.063 & 86.3 & -13.4 & 0.062 & 0.063 & 87.0 & +17.2 & 0.071 & 0.070 & 91.6 \\
\addlinespace[2pt]
5 & 40.0 & O1C1 & -0.289 & -0.9 & 0.067 & 0.063 & 93.3 & -0.4 & 0.070 & 0.067 & 92.9 & -0.1 & 0.067 & 0.064 & 92.9 \\
 &  & O1C0 & -0.289 & -0.9 & 0.067 & 0.063 & 93.2 & -6.8 & 0.066 & 0.063 & 90.6 & -0.1 & 0.067 & 0.064 & 92.9 \\
 &  & O0C1 & -0.289 & -0.6 & 0.070 & 0.066 & 93.6 & -0.4 & 0.070 & 0.067 & 92.9 & +26.1 & 0.074 & 0.070 & 83.8 \\
 &  & O0C0 & -0.289 & -7.2 & 0.066 & 0.063 & 90.2 & -6.8 & 0.066 & 0.063 & 90.6 & +26.1 & 0.074 & 0.070 & 83.8 \\
\addlinespace[2pt]
5 & 60.0 & O1C1 & -0.289 & +0.4 & 0.067 & 0.065 & 94.3 & +0.8 & 0.070 & 0.068 & 94.4 & +1.6 & 0.067 & 0.066 & 94.6 \\
 &  & O1C0 & -0.289 & +0.6 & 0.067 & 0.065 & 94.0 & -11.0 & 0.063 & 0.061 & 89.0 & +1.6 & 0.067 & 0.066 & 94.6 \\
 &  & O0C1 & -0.289 & +0.7 & 0.070 & 0.068 & 95.0 & +0.8 & 0.070 & 0.068 & 94.4 & +20.9 & 0.071 & 0.068 & 88.3 \\
 &  & O0C0 & -0.289 & -11.5 & 0.063 & 0.061 & 88.4 & -11.0 & 0.063 & 0.061 & 89.0 & +20.9 & 0.071 & 0.068 & 88.3 \\
\bottomrule
\end{tabular}
\end{adjustbox}
\end{table}

\section{Additional data analysis results} \label{supp:sec:data_analysis}

Table \ref{supp:tab:hfaction_while_alive} shows the estimated treatment contrast at several landmark times in HF-ACTION data analysis.

\begin{table}[tbp]
\caption{Doubly robust estimates of the while-alive all-cause hospitalization rate in the HF-ACTION analysis. Rates are expected all-cause hospitalizations per year alive. The contrast is exercise training minus usual care, so a negative value indicates a lower hospitalization burden per year alive under exercise training.}
\label{supp:tab:hfaction_while_alive}
\centering
\setlength{\tabcolsep}{4.5pt}
\renewcommand{\arraystretch}{1.12}
\begin{adjustbox}{max width=\textwidth}
\begin{tabular}{ccccc}
\toprule
Horizon, years
& Usual care, \(\widehat\psi_0\) (95\% CI)
& Exercise training, \(\widehat\psi_1\) (95\% CI)
& Difference, \(\widehat\Delta\) (95\% CI)
& \(p\)-value \\
\midrule
0.5 & 0.706 (0.620, 0.791) & 0.693 (0.607, 0.779) & -0.013 (-0.134, 0.107) & 0.830 \\
1.0 & 0.730 (0.667, 0.793) & 0.760 (0.695, 0.826) &  0.030 (-0.060, 0.120) & 0.509 \\
2.0 & 0.725 (0.675, 0.774) & 0.709 (0.660, 0.758) & -0.016 (-0.084, 0.052) & 0.652 \\
3.0 & 0.608 (0.570, 0.645) & 0.596 (0.557, 0.635) & -0.011 (-0.064, 0.042) & 0.678 \\
4.0 & 0.551 (0.519, 0.583) & 0.526 (0.493, 0.559) & -0.025 (-0.070, 0.020) & 0.271 \\
\bottomrule
\end{tabular}
\end{adjustbox}
\par\vspace{2pt}
{\footnotesize CI, confidence interval. All intervals are pointwise Wald intervals based on the influence function variance estimator.}
\end{table}

\end{document}